\documentclass[a4paper,11pt]{article}
\usepackage[english]{babel}
\usepackage[sectionbib]{natbib}
\usepackage{sectsty, secdot}

\usepackage{chngcntr}
\usepackage{caption}

\usepackage{geometry}
\usepackage{amsmath}
\usepackage{amssymb}
\usepackage{amsfonts}
\usepackage{authblk}
\usepackage{multirow}
\usepackage{amsthm}
\usepackage{graphicx}
\usepackage{mathrsfs}
\usepackage{afterpage}
\usepackage[debugshow]{tabularx}
\usepackage{cases}
\usepackage{enumerate}
\usepackage{url}

\usepackage{appendix}

\usepackage{subfigure}
\usepackage{bm}

\usepackage{abstract}

\usepackage[singlespacing]{setspace}

\usepackage[mathscr]{eucal}
\usepackage{lscape}

\usepackage{threeparttable}
\usepackage[table,xcdraw,usenames,dvipsnames,svgnames]{xcolor}
\usepackage{accents}
\usepackage{array,epsfig,fancyhdr,rotating}
\usepackage[section]{placeins}%
\usepackage{natbib}
\usepackage{epstopdf}%
\usepackage{color}
\usepackage{booktabs}

\usepackage[pdftex,colorlinks=true,hypertexnames=false]{hyperref}
\definecolor{darkblue}{rgb}{0,0,.6}
\hypersetup{citecolor=darkblue,linkcolor=darkblue,urlcolor=darkblue}

\usepackage{indentfirst}  
\newtheorem{theorem}{Theorem}
\newtheorem{lemma}{Lemma}
\newtheorem{corollary}{Corollary}

\theoremstyle{definition}

\newtheorem{assum}{Assumption}

\newtheorem{remark}{Remark}

\fontfamily{ptm}\selectfont 

\numberwithin{equation}{section}

\begin{document}
	\title{\textbf{Adaptive LAD-Based Bootstrap Unit Root Tests under Unconditional Heteroskedasticity}}
	
	\author[a]{Jilin Wu}
	\author[b]{Ruike Wu\thanks{\noindent Correspondence to: Department of Finance, School of Economics, Xiamen University, Xiamen, China. E-mail addresses: rainforest1061@126.com (J. Wu),  ruikewu@foxmail.com (R. Wu), xiaoz@bc.edu (Z. Xiao)}}
	\author[c]{Zhijie Xiao}
	
	\affil[a]{Department of Finance, School of Economics, Gregory and Paula Chow Institute for studies in Economics, Wang Yanan Institute for Studies in Economics (WISE), Xiamen University, China}
	\affil[b]{Department of Finance, School of Economics, Xiamen University, China}
	\affil[c]{Department of Economics, Boston College, USA}
	
	\date{}
	\maketitle
	\begin{abstract}
		This paper explores testing unit roots based on least absolute deviations (LAD) regression under unconditional heteroskedasticity. We first derive the asymptotic properties of the LAD estimator for a first-order autoregressive process with the coefficient (local to) unity under unconditional heteroskedasticity and weak dependence, revealing that the limiting distribution of the LAD estimator (consequently the derived test statistics) is closely associated with unknown time-varying variances. To conduct feasible LAD-based unit root tests under heteroskedasticity and serial dependence, we develop an adaptive block bootstrap procedure, which accommodates time-varying volatility and serial dependence, both of unknown forms, to compute critical values for LAD-based tests. The asymptotic validity is established. We then extend the testing procedure to allow for deterministic components. Simulation results indicate that, in the presence of unconditional heteroskedasticity and serial dependence, the classic LAD-based tests demonstrate severe size distortion, whereas the proposed LAD-based bootstrap tests exhibit good size-control capability. Additionally, the newly developed tests show superior testing power in heavy-tailed distributed cases compared to considered benchmarks. Finally, empirical analysis of real effective exchange rates of 16 EU countries is conducted to illustrate the applicability of the newly proposed tests.
	\end{abstract}
	
	
	
	\par
	\newpage

	\begin{sloppypar}	\section{Introduction}
		
		More and more evidence shows that time variation in unconditional  volatility is a common feature in macroeconomic and financial data, see, \cite{Andreou2002} and \cite{Liu2008}. The presence of unconditional heteroskedasticity may mislead conventional unit root testing procedure, leading to erroneous conclusions. \cite{Hamori1997} found that a simple level shift in innovation variance could invalidate  usual unit root asymptotics. \cite{Kim2002} investigated the impact of permanent variance shift in the innovations on Dickey-Fuller tests, and found they were badly oversized. \cite{Cavaliere2005} systematically explored the effect of time-varying variance  on unit root tests in a general framework, and showed that the size and power of the tests are largely affected by the form of underlying volatility process. Therefore, if unconditional heteroskedasticity is not properly accounted for, spurious inference about the presence or absence of a unit root may be made.
		
		Consequently, subsequent studies have developed numerous modified unit root tests that are robust to unknown unconditional heteroskedasticity. 
		For example, \cite{Cavaliere2007} computed critical values for the M-type unit root tests developed by \cite{Perron1996}  via simulation based on the variance profile estimator. \cite{Cavaliere2008a} proposed wild bootstrap versions of  the M-type unit root tests and established their validity under nonstationary volatility. \cite{Cavaliere2008b} proposed a nonparametric correction method based on variance profiles to deal with heteroskedasticity, demonstrating that the  corrected statistics share the same asymptotic distributions as the corresponding standard tests derived under homoskedasticity. To achieve higher testing power, \cite{Boswijk2018} developed a likelihood ratio (LR) test that accommodates nonstationary volatility and achieves power close to the Gaussian asymptotic power envelope. However, their LR test is derived under the assumption of Gaussianity.
		\cite{Beare2018} proposed a class of modified unit root test statistics that are robust to unstable volatility, which is achieved by purging heteroskedasticity from the data using a kernel estimate of volatility before applying standard tests.
		
		\subsection{Motivation}
		The aforementioned tests  are all derived in the framework of least squares (LS) estimation or rely on Gaussian assumption.  Under departure from normality, particularly for those heavy-tailed innovations, these tests may exhibit relatively poor testing powers. Macroeconomic and financial data are often characterized by heavy tails, which have been extensively documented in the literature, see, for example, \cite{rachev2003} for an overview.
		Moreover, heavy-tailed behavior and time-varying volatility in the data often occur simultaneously and are closely intertwined, making it difficult to clearly distinguish between the two ones.
		Consequently, it is essential to devise a unit root testing procedure that is robust to both heavy tails and time-varying volatility in innovations.

		In the unit root testing literature, considerable effort has been devoted to deriving robust inference \citep{Herce1996, Hasan1997, Koenker2004, Thompson2004}. However, the impact of heteroskedasticity on these robust tests has been rarely discussed and remains unclear, necessitating further investigation. In this paper, we focus on studying unit root tests based on the LAD regression in the context of unconditional heteroskedasticity, as it serves as the most straightforward robust counterpart to LS-based unit root tests.
		The LAD technique utilizes absolute values of errors instead of squared errors, making it less sensitive to extreme values. Consequently, it is regarded as more robust than the LS technique when innovations follow heavy-tailed distributions.

		The unit root processes under the LAD framework have been extensively studied.  \cite{Knight1989} developed asymptotic theory for the LAD estimation of the autoregressive coefficient in  a unit root process with independent and infinite variance errors. \cite{Herce1996} complemented this work by developing the asymptotic theory for the LAD estimator under the assumption of mixing-dependent errors with finite variances. The limiting distribution of the LAD estimator depends on  long-run variance parameters, leading to the proposal of several transformed test statistics with computationally convenient critical values. However, \cite{Thompson2004} pointed out that the transformation causes the problem\footnote{It is also noted by \cite{Herce1996}} that the transformed tests exhibit powers equal to sizes for normal innovations. Hence, he described a  way to compute critical values without transforming the test statistics. The resulting tests are much more powerful than  the LS-based tests for thicker-tailed distributions.
		\cite{Moreno2000} proposed an LAD-based bootstrap procedure to obtain critical values, their simulation results show that LAD-based tests exhibit higher power than the LS counterparts when errors are heavy-tailed distributed. However, their bootstrap procedure is established under the assumption that errors are  independently and identically distributed (i.i.d). \cite{Jansson2008} derived the asymptotic power envelopes for the locally asymptotically $\alpha$-similar  unit root tests, which contains the LAD-based unit root tests, in a zero-mean AR(1) model. \cite{Li2009} considered the LAD-based unit root tests for GARCH-type errors.  In addition, since the LAD estimator can be regarded as a special case of the quantile regression (QR) estimator at median, some attention is also given to unit root tests within the QR framework. \cite{Koenker2004} proposed several QR-based unit root tests which show higher detection ability than the LS-based tests in non-Gaussian cases. \cite{Galvao2009} extended these tests by incorporating stationary covariates and allowing for a linear time trend. \cite{Li2018} further generalized these QR-based tests to a nonlinear quantile autoregressive framework. 
		
		\subsection{Contribution}
		
		To the best of our knowledge, none of the LAD-based studies have explored the impact of unconditional heteroskedasticity on asymptotic theory of a (nearly) unit root process.
		In this paper, we try to fill this gap. We first  derive the asymptotic theory of the LAD estimator of the autoregressive coefficient in a (nearly) unit root process under unconditional heteroskedasticity and weak dependence, showing that the limiting distribution of the LAD estimator is not free of nuisance parameters and can differ substantially (depending on the underlying heteroskedastic processes) from the limiting distribution that is obtained under homoskedasticity \citep{Herce1996}.  To conduct feasible LAD-based unit root tests in the presence of time-varying volatility and serial dependence, we propose an adaptive block bootstrap resampling procedure to compute the critical values for the LAD-based tests. The asymptotic validity of resampling procedure is established.  The newly proposed bootstrap procedure successfully generates unit root pseudo series, adaptively retaining the important heteroskedastic dynamics and dependence structure of the data, without making any parametric assumptions on the two characteristics.
		Next, we  extend our main results to accommodate deterministic components in  the unit root process, the asymptotic properties are also derived for the extended models.

		Simulation results show that, in the presence of unconditional heteroskedasticity and weak dependence, the LAD-based  bootstrap unit root tests exhibit reasonable empirical sizes in all considered situations, whereas the conventional LAD-based tests derived under homoskedasticity suffer from substantial size distortion. Furthermore, compared to other popular LS-based competitors that are robust to both time-varying volatility and serial dependence, the proposed tests also achieve higher testing power in non-Gaussian situations. Finally, we apply the newly proposed LAD-based bootstrap tests to investigate the validity of purchasing power parity condition for 16 European countries using real effective exchange rates.

		
		The main contributions can be summarized into two folds: 
		(1) We are the first to derive the asymptotic theory of the LAD estimator in the presence of both unconditional heteroskedasticity and serial dependence. (2) We develop  feasible LAD-based bootstrap unit root tests that accommodate unconditional heteroskedasticity and serial dependence, both of unknown forms, which is new to the literature. We note that although we focus on solving the problem of heteroskedasticity, we also adapt our testing procedure to allow for serial dependence to enhance its practical applicability, which poses significant challenges for theoretical proof. As a result, the new LAD-based tests are particularly suitable for the situations where errors are heavy-tailed, heteroskedastic, and weakly dependent. Moreover, It is worth noting that the newly developed bootstrap procedure is general and can  also be easily applied to LS-based tests. Lastly, we clarify that the focus of this paper is on addressing unconditional heteroskedasticity in LAD-based unit root tests, the good performance of the LAD-based tests in  heavy-tailed situations has been extensively discussed in the aforementioned literature.

		\subsection{Organization}	
		This paper is organized as follows: In Section \ref{section model}, we consider the basic  unit root process under unconditional heteroskedasticity and serial dependence, and study the asymptotic properties of the LAD estimator. Section \ref{section theorem} proposes an infeasible/feasible adaptive block bootstrap procedure to approximate the limiting null distributions of the LAD-based unit root tests.  In Section \ref{section extension}, we further extend our testing procedure to allow for deterministic components. Section \ref{section simu} reports Monte Carlo simulation results  to assess the finite sample performance of the newly proposed testing procedure. An empirical application to real effective exchange rates of 16 European countries is conducted for checking the validity of purchasing power parity in Section \ref{section empirical}.  Section \ref{section conclusion} offers conclusions. Proofs and some additional simulation results are relegated to the supplementary material.
		
		Throughout the paper, $x^{\prime}$ is the transpose of $x$, $\left\lfloor x\right\rfloor $ represents the largest integer less than or equal to $x$.
		{$\overset{p}{\rightarrow}$ and $\overset{d}{\rightarrow}$ denote
			the convergence in probability and in distribution, respectively. $\overset{p^{*}}{\rightarrow}$ and $\overset{d^{*}}{\rightarrow}$ stand for the convergence in probability and in distribution under the bootstrap law, respectively.  $sgn(x)=I(x>0)-I(x<0)$ is the sign function of $x\in \mathbb{R}$ where $I(\cdot)$ is the indicator function}.
		
		\section{The Basic Model and  Asymptotic Properties}\label{section model}
		Consider the following first order autoregressive model
		\begin{equation}
			y_{t}=\gamma_{0}y_{t-1}+u_{t}, \label{AR}%
		\end{equation}
		where the error satisfies
		\begin{equation}
			u_{t}=\sigma_{t}\varepsilon_{t}\colon=\sigma\left(  \frac{t}{T}\right)
			\varepsilon_{t} \label{HUF}%
		\end{equation}
		for $t=1,\ldots,T$ with $y_{0} = 0$ for convenience. $\left\{\sigma_{t}\right\}  $ is a deterministic positive sequence functioning as a proxy for all factors that affect the unconditional heteroskedasticity, and $\left\{  \varepsilon_{t}\right\}  $ is a strictly stationary mixing process, which accounts for  conditional heterogeneity and temporal dependence in $u_{t}$. We do not assume a particular parametric specification but only impose some mild restrictions on $\sigma\left(
		\cdot\right)  $. This sort of model-free assumptions not only helps to avoid potential misspecifications in $\sigma\left(  \cdot\right)  $, but  {is also} general enough to encompass  abrupt breaks and smooth transitions as special cases. Such rescaling device for $\sigma\left(  \cdot\right)  $ was first introduced by \cite{Dahlhaus1997}, and since then has been widely adopted in \textcolor{black}{the} time series literature; see, for example, \cite{Cavaliere2005}, \cite{Cavaliere2007}, and \cite{Wu2018}, to name only a few. 
		
		Our focus in this paper is on testing the unit root null hypothesis of $\mathbb{H}_{0}:\gamma_{0}=1$ in (\ref{AR}) against the alternative $\mathbb{H}_{A}:\gamma_{0}<1$. The LAD estimator can be calculated by solving
		\begin{equation}
			\hat{\gamma}_{LAD}=\arg\min_{\gamma}\sum_{t=1}^{T}\left\vert
			y_{t}-\gamma y_{t-1}\right\vert . \label{LAD0}%
		\end{equation}
		
		Denote $v=T\left(  \gamma-\gamma_{0}\right)  $, then    {solving the minimization problem} (\ref{LAD0}) is equivalent to minimizing the following criterion
		\begin{equation}
			H_{T}\left(  v\right)  =\sum_{t=1}^{T}  \left\vert u_{t}-vT^{-1}y_{t-1}\right\vert -\left\vert u_{t}\right\vert  . \label{MINC}%
		\end{equation}
		By Knight's identity \citep{Knight1989}, the objective function (\ref{MINC}) can be further rewritten as
		\begin{equation}
			H_{T}\left(  v\right)  =-v\left(  \frac{1}{T}\sum_{t=1}^{T}y_{t-1}sgn\left(u_{t}\right)  \right)  +2\sum_{t=1}^{T}\int_{0}^{vT^{-1}y_{t-1}}\left[I\left(  u_{t}\leq s\right)  -I\left(  u_{t}\leq0\right)  \right]  ds.
			\label{HTV}%
		\end{equation}
		If $\hat{v}$ is the minimizer of $H_{T}\left(  v\right)  $, then we have $\hat{v}=T\left(  \hat{\gamma}_{LAD}-\gamma_{0}\right) $. 
		Therefore, in order to derive the limiting distribution of $\hat{\gamma}_{LAD}$, we need to study the two terms in the above equation $H_{T}\left(  v\right)  $, whose asymptotic analyses are based on the following assumptions.

		\begin{sloppypar}
			\begin{assum}
				\label{assum 1} 
				\begin{enumerate}
					\item[(i)] $\left\{  \varepsilon_{t}\right\}
					$ is a strictly stationary mixing process with $E\left(
					\varepsilon_{t}\right)  =0$ and $E\left( sgn(\varepsilon_{t})\right)=0$;
					\item[(ii)] The density function $f$ of $\varepsilon_{t}$ is continuous and positive at 0 with respect to the Lebesgue measure. The conditional cumulative density function $F_{t-1}(x)=P(\varepsilon_{t}<x|\varepsilon_{t-1},\varepsilon_{t-2},\cdots)$ has the first derivative $f_{t-1}(x)$ a.s. with $f_{t-1}(x_{T})$ uniformly integrable for any sequence $x_{T}\rightarrow0$ and $E(f_{t-1}^{r}(0))<\infty$ for some $r>1$;
					\item[(iii)] The long-run variance-covariance 
					$$\Sigma=\left(
					\begin{array}
						[c]{cc}%
						\delta_{1}^{2} & \delta_{12}\\
						\delta_{12} & \delta_{2}^{2}%
					\end{array}
					\right)  :=\Sigma_{0}+ \Sigma_{1} + \Sigma_{1}^{\prime}$$
					is strictly positive definite  with $\Sigma_{1}=\sum_{s=2}^{\infty}E\left[\left(  \varepsilon_{1},sgn(\varepsilon_{1})\right)^{\prime}\left(  \varepsilon_{s},sgn(\varepsilon_{s})\right)\right]$ and $\Sigma_{0}=E[\left(\varepsilon_{1},sgn(\varepsilon_{1})\right)^{\prime}\left(  \varepsilon_{1},sgn(\varepsilon_{1})\right)]$.
				\end{enumerate}
			\end{assum}
		\end{sloppypar}
		
		The zero-median assumption  {postulated by Assumption \ref{assum 1}(i)} is a general set-up for LAD-type estimators, see e.g., \cite{Herce1996},\cite{zhu2015},  \cite{Zhu2019}, and \cite{Zhang2022}.
		Assumption \ref{assum 1}(ii) gives the limitations for the density function, the conditional density function and the conditional cumulative density function, which are  technique requirements for the LAD regression under serial dependence, see Assumption 1(iii) of \cite{Herce1996}. 
		
		\begin{assum}
			\label{assum 2}
			\textit{There exists $ p > 2$   such that $ E|\varepsilon_t|^{p} <\infty $ and the mixing coefficients $\{\alpha_{m}\}$ satisfy $\sum_{m=1}^{\infty}\alpha_{m}%
				^{1/\beta-1/p}<\infty$ for some $2<\beta< p$.} 
		\end{assum}
		
		Assumption \ref{assum 2} provides strong mixing and moment conditions for $\left\{{\varepsilon_{t}}\right\}  $. It imposes  \textcolor{black}{weaker}  moment restrictions on $\varepsilon_t$ than those required by LS-based unit root tests 
		\citep[see, e.g.,][]{Cavaliere2005, Cavaliere2007, Cavaliere2008a}, and hence admits \textcolor{black}{heavier tails in innovations}.
		The restriction on the mixing coefficients $\alpha_m$ embodies a trade-off between the mixing decay rate and the moment conditions. 
		{Under Assumptions \ref{assum 1} and \ref{assum 2},  a broad scope of data generating mechanisms for $\{\varepsilon_t\}$ is included, such as various stationary ARMA-type and GARCH-type processes.}
		
		
		\begin{assum}
			\label{assum 3} \textit{$\sigma\left(  \cdot\right)  $ is a bounded
				positive and deterministic function, and has continuous second derivatives
				except at finite discontinuity points on $[0,1]$.}
		\end{assum}

		Under Assumption \ref{assum 3}, the heteroskedastic function $\sigma(\cdot)$ is bounded, and can have up to a finite number of jumps. Models of single or multiple volatility shifts satisfy Assumption \ref{assum 3} with $\sigma(\cdot)$  piece-wise constant. 
		The assumption of a deterministic  function $\sigma(\cdot)$ can be weakened by simply assuming stochastic independence between $\{\varepsilon_{t}\}$ and $\{\sigma_{t}\}$, 
		readers can be referred to \cite{Cavaliere2009} for more discussion. 
		
		The following lemma provides a bivariate invariance principle, laying the foundation for this paper.
		\begin{lemma}
			\label{lemma 1} 
			\begin{enumerate}
				\item[(i)] Suppose Assumptions \ref{assum 1}-\ref{assum 2} hold true, then we  have
				\begin{equation}
					\frac{1}{\sqrt{T}}\sum_{t=1}^{\left\lfloor Tr\right\rfloor }(\varepsilon_{t},sgn(\varepsilon_{t}))^{\prime}\overset{d}{\rightarrow}\left(B_{1}(r),B_{2}(r)\right)^{\prime}:=\Sigma^{1/2}\left(  W_{1}(r),W_{2}(r)\right)  ^{\prime};\label{sbip0}
				\end{equation}
				\item[(ii)] Suppose Assumptions \ref{assum 1}-\ref{assum 3} hold true, then we  have
				\begin{equation}
					\frac{1}{\sqrt{T}}\sum_{t=1}^{\left\lfloor Tr\right\rfloor }(u_{t},sgn(u_{t}))^{\prime}\overset{d}{\rightarrow}\left(  B_{\sigma}(r),B_{2}(r)\right)^{\prime}, \label{gbip}%
				\end{equation}
				where $\left(  W_{1}(\cdot),W_{2}(\cdot)\right)  ^{\prime}	$ is \textcolor{black}{a} standard bivariate Brownian motion, $\Sigma$ is defined in Assumption \ref{assum 1}(iii), and $\mathit{B_{\sigma}(r)=\int_{0}^{r}\sigma(s)dB_{1}(s)}$.
			\end{enumerate}
		\end{lemma}

		The first result (\ref{sbip0}) has been well-established in the literature, see, for example, \cite{Hansen1992}. The second result (\ref{gbip}) extends (\ref{sbip0}) to allow for deterministic heteroskedasticity. As a special case of ours, \cite{Cavaliere2005} developed a univariate invariance principle for  {the partial sum of} heteroskedastic errors, see his Lemma 1. It is worth noting that $sgn(u_{t})$ is always correlated with $u_{t}$ even if $\left\{  u_{t}\right\}  $ is an i.i.d. sequence. 
		
		\subsection{Asymptotic properties under the null hypothesis} \label{subsection asy null}
		In this subsection, we  demonstrate the asymptotic results for the LAD estimator   $\hat{\gamma}_{LAD}$ under the  null hypothesis.  The following Lemma \ref{lemma 2} provides the limiting distributions for the two terms in $H_T(v)$, respectively.
		
		\begin{lemma}
			\label{lemma 2} 
			Suppose Assumptions \ref{assum 1}-\ref{assum 3} hold true, 
			let $y_{t}$ be determined by (\ref{AR}) and
			(\ref{HUF}) with $\gamma_{0}=1$, then we have
			$$
			\frac{1}{T}\sum_{t=1}^{T}y_{t-1}sgn\left(  u_{t}\right)  \overset
			{d}{\rightarrow}\int_{0}^{1}B_{\sigma}(r)dB_{2}(r)+\Gamma,$$
			$$2\sum_{t=1}^{T}\int_{0}^{vT^{-1}y_{t-1}}\left[  I\left(  u_{t}\leq s\right)
			-I\left(  u_{t}\leq0\right)  \right]  ds\overset{d}{\rightarrow}v^{2}\left(	f\left(  0\right)  \int_{0}^{1}\sigma^{-1}\left(  r\right)  B_{\sigma}^{2}\left(  r\right)  dr\right), 
			$$
			where $\Gamma=lim_{T\rightarrow\infty}\frac{1}{T}\sum_{i=1}^{T}\sum_{j=i+1}^{\infty}E(u_{i}sgn(u_{j}))$.
		\end{lemma}

		As an immediate result of Lemma \ref{lemma 2}, we obtain
		$$
		H_{T}\left(  v\right)  \overset{d}{\rightarrow}H\left(  v\right)
		:=-v\left(  \int_{0}^{1}B_{\sigma}(r)dB_{2}(r)+\Gamma\right)
		+v^{2}\left(  f\left(  0\right)  \int_{0}^{1}\sigma^{-1}\left( r\right)
		B_{\sigma}^{2}\left(  r\right)  dr\right)  .
		$$
		Note that $H\left(  v\right)  $ is minimized at $\hat{v}=T\left( \hat{\gamma}_{LAD}-\gamma_{0}\right)  $, by the convexity lemma of \cite{Pollard1991}  and the arguments of \cite{Knight1989}, we  have the following theorem.
		\begin{theorem}
			\label{theorem 1} 
			Suppose Assumptions \ref{assum 1}-\ref{assum 3} hold true, 
			let $y_{t}$ be determined by (\ref{AR})
			and (\ref{HUF}) with $\gamma_{0}=1$, then we
			have
			\begin{align}
				T(\hat{\gamma}_{LAD}-1)\overset{d}{\rightarrow}
				\mathcal{D}(\sigma):= \frac{\int_{0}^{1}B_{\sigma	}(r)dB_{2}(r)+\Gamma}{2f(0)\int_{0}^{1}\sigma^{-1}(r){\ B_{\sigma}^{2}(r)}	dr}.
				\label{Limit of coeff under null}
			\end{align}
		\end{theorem}

		To better observe the limiting distribution of $T(\hat{\gamma}_{LAD}-1)$, we choose
		\begin{equation}
			\Sigma^{1/2}=\left[
			\begin{array}
				[c]{cc}%
				\delta_{1} & 0\\
				\delta_{12}/\delta_{1} & \Delta^{1/2}/\delta_{1}%
			\end{array}
			\right]  \label{CHDC}%
		\end{equation}
		as a triangular decomposition of $\Sigma$, where $\Delta=\delta_{1}^{2}%
		\delta_{2}^{2}-\delta_{12}^{2}$. Correspondingly,  we can rewrite $B_{1}(\cdot)=\delta_{1}W_{1}(\cdot)$, $B_{2}(\cdot)=\delta_{12}\delta_{1}^{-1}W_{1}(\cdot)+\Delta^{1/2}\delta_{1}^{-1}W_{2}(\cdot)$ and $B_{\sigma}(\cdot)=\delta_{1}W_{\sigma}(\cdot)$ where $W_{\sigma}(r)=\int_{0}^{r}\sigma(s)dW_{1}(s)$.  It is not hard to  \textcolor{black}{show} that $W_{\sigma}(r)$ is normally distributed with mean 0 and variance $\int_{0}^{r}\sigma^{2}(s)ds$, which corresponds to the diffusion process generated by $dW_{\sigma}(r)=\sigma(r)dW_{1}(r)$ with $W_{\sigma}(0)=0$. As a result, $\int_{0}^{1}B_{\sigma}(r)dB_{2}(r)$ has the following decomposition
		\begin{equation}
			\int_{0}^{1}B_{\sigma}(r)dB_{2}(r) =\delta_{12}\int_{0}^{1}W_{\sigma}(r)dW_{1}(r) +\Delta^{1/2}\left( \int_{0}^{1}{W_{\sigma}^{2}}(r)dr\right)  ^{1/2}\vartheta, \label{BDC}%
		\end{equation}
		where  the second term is derived by the fact that $\int_0^1W_\sigma(r)dW_2(r)$ is a scale mixture of normals,
		and  $\vartheta$ is a standard normal random variable independent of $\int_{0}^{1}W_{\sigma}^{2}(r)dr$. When $\sigma(r)=1$, $\int_{0}^{1}W_{1}(r)d{W_{\sigma}}(r)=(W_{1}^{2}(1)-1)/2$, the decomposition (\ref{BDC}) degenerates to equation (16) in \cite{Herce1996}. 
		
		Based on (\ref{BDC}), the limiting distribution of $T(\hat{\gamma
		}_{LAD}-1)$ can be decomposed into 
		\begin{equation}
			\frac{\delta_{12}\int_{0}^{1}W_{\sigma}(r)dW_{1}(r)  +\Gamma}{2f(0)\delta_{1}^{2}\int_{0}^{1}\sigma^{-1}(r){\ W_{\sigma}^{2}(r)}dr}+\frac{\Delta^{1/2}}{2f(0)\delta_{1}^{2}}\frac{\left( \int_{0}^{1}{W_{\sigma}^{2}}(r)dr\right)  ^{1/2}}{\int_{0}^{1}\sigma^{-1}(r){W_{\sigma}^{2}}(r)dr}\vartheta. \label{DC1}%
		\end{equation}
		
		Recall the formula of the LS estimator $\hat{\gamma}_{LS}$ 
		\citep[see,][]{Phillipsperron1987}
		{under the conditions of Theorem \ref{theorem 1}}. By using Lemma \ref{lemma 1}, we can easily show that 
		\begin{equation}
			T(\hat{\gamma}_{LS}-1)\overset{d}{\rightarrow}\frac{\delta_{1}^{2}\int_{0}^{1}W_{\sigma}(r)dW_{\sigma}(r)+\Gamma_{2}}{\delta_{1}^{2}\int_{0}^{1}{W_{\sigma}^{2}}(r)dr}, \label{DC2}
		\end{equation}
		where $\Gamma_{2}=lim_{T\rightarrow\infty}\frac{1}{T}\sum_{i=1}^{T}\sum_{j=i+1}^{\infty}E(u_{i}u_{j})$. Therefore, by comparing (\ref{DC1}) with  (\ref{DC2}),  it is shown  that  the limiting distribution of $T(\hat{\gamma}_{LAD}-1)$ can be roughly regarded as the combination of  a ``some sort of unit root" distribution  and a  normal distribution. In addition,  the second term in (\ref{DC1}) vanishes in (\ref{DC2}) since $\Delta=0$ in the LS  case. We also note that the denominator $\int_{0}^{1}\sigma^{-1}(r){W_{\sigma}^{2}}(r)dr$ in (\ref{DC1}) has an extra term $\sigma^{-1}(\cdot)$, which plays the role of a weighting function and helps reduce the heteroskedastic impact from outliers, hence enhancing the robustness of the LAD estimator to some degree.
		
		\subsection{Asymptotic properties under the local to unity alternative hypothesis}
		Next, we derive the asymptotic local power function of $\hat{\gamma}_{LAD}-1$, which can  contribute to  {a} deeper understanding of  {the origins of the testing powers of unit root tests.}
		Specifically, we consider the following local departure from the null hypothesis
		\begin{equation}
			\mathbb{H}_{A}:\gamma_{0}=\exp\left(  -\frac{c}{T}\right)  \approx1-\frac{c}{T}%
			,c\geq0, \label{ALT}%
		\end{equation}
		where the non-negative parameter $c$ controls the degree of departure and  $T^{-1}$ is  the speed at which the local alternative deviates from the null. The null holds when $c=0$ and holds ``locally" as $T\rightarrow\infty$ for $c>0$. 
		
		Under $\mathbb{H}_{A}$, the data generated by  (\ref{AR}) and (\ref{HUF}) has the representation
		$
		y_{t}=\sum_{j=1}^{t}e^{-c\left(  t-j\right)  /T}\sigma_{j}\varepsilon
		_{j}
		$. As such, by applying Lemma \ref{lemma 1} (ii), the asymptotic behavior of $y_{\lfloor Tr \rfloor}$ scaled by $\sqrt{T}$  for $r \in [0,1]$ is 
		\begin{equation}
			\frac{1}{\sqrt{T}}y_{\left\lfloor Tr\right\rfloor }\overset{d}{\rightarrow
			}\mathcal{J}_{c,\sigma}(r)=\int_{0}^{r}e^{-c(r-s)}dB_{\sigma}(s),\label{AOU}%
		\end{equation}
		where $\mathcal{J}_{c,\sigma}(\cdot)$ is an Ornstein--Uhlenbeck (OU)-type process \textcolor{black}{with heteroskedastic increments} 
		and satisfies
		$d\mathcal{J}_{c,\sigma}(r)=-c\mathcal{J}_{c,\sigma}(r)+dB_{\sigma}(r)$.  Obviously, \textcolor{black}{a} larger value of $c$ implies \textcolor{black}{a} larger deviation of $\mathcal{J}_{c,\sigma}(r)$ from $B_\sigma(r)$.
		By using  {the} similar arguments to those of proving Theorem \ref{theorem 1}, we obtain the following  local to unity asymptotic result for $\hat{\gamma}_{LAD}$.
		\begin{theorem}
			\label{theorem 2}
			Suppose Assumptions \ref{assum 1}-\ref{assum 3} hold true,  let $y_{t}$ be determined by (\ref{AR})	and (\ref{HUF}) with $\gamma_{0}$ satisfying (\ref{ALT}), we then have
			\begin{equation}
				T(\hat{\gamma}_{LAD}-\gamma_{0})\overset{d}{\rightarrow}\mathcal{D}(c,\sigma) := \frac{\int_{0}	^{1}\mathcal{J}_{c,\sigma}(r)dB_{2}(r)+\Gamma}{2f(0)\int_{0}^{1}\sigma^{-1}(r){\ \mathcal{J}_{c,\sigma}^{2}(r)}dr}. \label{TH2}%
			\end{equation}
		\end{theorem}
		
		Compared with (\ref{Limit of coeff under null}), the asymptotic result (\ref{TH2}) additionally depends on the location parameter $c$. When $c=0$, Theorem \ref{theorem 2} degenerates to Theorem \ref{theorem 1}; when $c > 0$, by using the approximation $exp{\left( -\frac{c}{T}\right) }\approx 1- \frac{c}{T}$, we obtain
		\begin{align}
			T(\hat{\gamma}_{LAD}-1)\overset{d}{\rightarrow}-c+\frac{\int_{0}^{1}\mathcal{J}_{c,\sigma}(r)dB_{2}(r)+\Gamma}{2f(0)\int_{0}^{1}\sigma^{-1}(r){\ \mathcal{J}_{c,\sigma}^{2}(r)}dr}.
			\label{limit coeff local power rewrite}
		\end{align}
		This result is considerably helpful in studying the behavior of $T(\hat{\gamma}_{LAD}-1)$ under  {the} alternatives. The larger value of  $c$ means  {the} larger negative shift in the center (the first term) and  {the} larger deviation  {(in shape)} from the null distribution (the second term).
		Hence, it is natural to develop unit root tests by utilizing  the difference performance of $\hat{\gamma}_{LAD} - 1$ under the null and  under the alternatives.

		
		\section{The LAD-based bootstrap unit root tests}
		\label{section theorem} 
		From Theorem \ref{theorem 1}, 
		the presence of heteroskedasticity    complicates the unit root testing because the limiting distribution of $T(\hat{\gamma}_{LAD}-1)$ is unobtainable in practice, and, explicitly or implicitly, involves the time-varying volatility $\sigma(\cdot)$ that enters in a highly nonlinear fashion.  Therefore, in this section, we  develop an adaptive block bootstrap (ABB)  procedure to calculate the critical values  for conducting feasible LAD-based unit root tests in the presence of unconditional heteroskedasticity.
		
		In this paper, for convenience, we pay our attention  to the classic coefficient-based test statistic $L_T$ and the t-ratio test statistic $t_T$ \citep{Koenker2004}, which are defined as follows respectively: 
		\begin{align}
			&L_T=T(\hat{\gamma}_{LAD}-1), \label{LT stat}	
			\\
			&t_T = 2\widehat{f(0)}\left(Y_{-1}^\prime P_C Y_{-1} \right)^{1/2}\left(\hat{\gamma}_{LAD} - 1 \right), \label{tT stat}
		\end{align}
		where $\widehat{f(0)}$ is  a consistent estimator of $f(0)$, $Y_{-1}$ is  the vector of the lagged dependent variables ${y_t}$,  $P_C = I_{T-1} - \textbf{1}_{T-1}(\textbf{1}_{T-1}^\prime \textbf{1}_{T-1})^{-1}\textbf{1}_{T-1}^\prime$, $\textbf{1}_{T-1}$ and $I_{T-1}$ are $T-1$ dimensional vector of ones and identity matrix, respectively.  
		
		\subsection{Infeasible adaptive block bootstrap}
		\label{Sec:infeasible ABB}
		The core issue of designing a suitable bootstrap procedure is to generate a pseudo series that can achieve the same functional results as in Lemma \ref{lemma 1}.
		Clearly, if we can mimic the behavior of the process $\{\varepsilon_{t}\}$ by some pseudo series $\{\varepsilon_{t}^{*}\}$, then the  {properties} of $\{u_{t}\}$ can be replicated by $\{\sigma_{t}\varepsilon_{t}^{*}\}$ immediately. Since $\{\varepsilon_{t}\}$ is a stationary mixing process, we apply  the moving block bootstrap method established by \cite{Kunsch1989} to help realize this goal, which  samples the blocks with replacement and then pastes the resampled blocks together to form 
		a pseudo series. This method makes no parametric assumptions, and  has been widely applied to deal with serial dependence in time series, see \cite{Paparoditis2003} and \cite{Hounyo2023}.

		Let $E^{*}(\cdot)$ be the expectation operator  {taken with respect to} the bootstrap probability space conditional on the data $\{y_{t}\}_{t=1}^{T}$. For the time being, $\{\sigma_{t}\}_{t=1}^{T}$ is assumed to be known in advance. Formally, the ABB procedure for approximating the distributions of the $L_{T}$ ($t_T$) statistic  takes the following steps:
		
		(i) Estimate $\hat{\gamma}_{LAD}$ by (\ref{LAD0}), compute the test statistic $L_{T}$ ($t_T$) according to \eqref{LT stat} (\eqref{tT stat}), and obtain the standardized residual $\hat{\varepsilon}_{t}=\hat{u}_{t}/\sigma_{t}$, where $\hat{u}_{t}=y_{t} -\hat{\gamma}_{LAD}y_{t-1}$.
		
		(ii) Choose a positive integer $b_{T}<T$, and let $i_{0},i_{1},\cdots,i_{k}$ be drawn independently and  {uniformly from}
		the set $\mathcal{W}=\{-T,-T+1,\cdots,-1-b_{T},1,2,\cdots,T-b_{T}\}$, where $k=\lfloor(T-1)/b_{T}\rfloor$ is the number of blocks and  $b_{T}$ is the block size.
		
		(iii) The bootstrap pseudo series $\{y_{t}^{*}\}_{t=1}^{T}$ is generated by $y_{t}^{*}=y_{t-1}^{*}+\sigma_{t}\varepsilon_{t}^{*}$ with $y_{0}^{*}=0$ and
		\begin{equation}
			\varepsilon_{t}^{*}=\left\{
			\begin{array}
				[c]{ll}%
				\hat{\varepsilon}_{i_{m}+s}, & i_{m}>0\\
				-\hat{\varepsilon}_{i_{m}+s+T+1}, & i_{m}<0
			\end{array}
			\right., \label{btr}%
		\end{equation}
		where $m=\lfloor(t-1)/b_{T}\rfloor$ and $s=t-mb_{T}-1$, $t=1,2,\ldots,T.$
		
		(iv) Compute the bootstrap test statistic $L_{T}^{*}$ ($t_T^*$) based on the bootstrap sequence $\{y_{t}^{*}\}_{t=1}^{T}$.
		
		(v) Repeat (ii)-(iv) $B$ times (e.g., $B=499$), and  then obtain the bootstrap statistic sequence $\{ L_{T,j}^{*}\}  _{j=1}^{B}$($\{ t_{T,j}^{*}\}  _{j=1}^{B}$). The bootstrap $p$-value is  given by  $p^{*}=B^{-1}\sum_{j=1}^{B}I(  L_{T,j}^{*}<L_{T})$ ($p^{*}=B^{-1}\sum_{j=1}^{B}I(  t_{T,j}^{*}<t_{T})$), and the null is rejected whenever $p^{*}$ is smaller than the given significance level.
		

		In (ii), we randomly draw indices from the set $\mathcal{W}$, which, by observing (\ref{btr}), has a one-to-one mapping relationship with $\mathcal{E}=\{-\hat{\varepsilon}_{1},\cdots,-\hat{\varepsilon}_{T-b_{T}},\hat{\varepsilon}_{1},\cdots,\hat{\varepsilon}_{T-b_{T}}\}$. We augment the sampled pool by adding $\{-\hat{\varepsilon}_{i}\}_{t=1}^{T-b_{T}}$ to  ensure that the bootstrap errors satisfy the conditions of both zero mean and zero median,  {a similar construction can be found in \cite{Moreno2000}.} 
		
		
		
		The following bootstrap invariance principle shows that the \textcolor{black}{partial sum of} pseudo vector series $\left\{\left(   \sigma_{t}\varepsilon_{t}^{*},sgn( \sigma_{t}\varepsilon_{t}^{*})\right)  ^{\prime}\right\}  _{t=1}^{T}$  \textcolor{black}{behaves} the same as \textcolor{black}{that of the} true ones $\left\{  \left( u_{t},sgn(u_{t})\right)  ^{\prime}\right\}  _{t=1}^{T}$  {asymptotically}.

		\begin{lemma}
			\label{lemma 3}
			Suppose Assumptions \ref{assum 1}-\ref{assum 3} hold true, let $y_{t}$ be determined by (\ref{AR}) and
			(\ref{HUF}) with $\gamma_{0}$ satisfying (\ref{ALT}), ${1}/{b_{T}}+b_{T}/{T}\rightarrow0$ as
			$T\rightarrow\infty$, we have
			\begin{align}
				\frac{1}{\sqrt{T}}\sum_{t=1}^{\lfloor Tr \rfloor}\left(  \varepsilon_{t}^{*	},sgn(\varepsilon_{t}^{*})\right)  ^{\prime}\overset{d^{*}}{\rightarrow}\left(  B_{1}(r),B_{2}(r)\right)  ^{\prime}; \label{bsbip0}	\\	\frac{1}{\sqrt{T}}\sum_{t=1}^{\lfloor Tr \rfloor}\left(  \sigma_{t}\varepsilon_{t}^{*	},sgn(\sigma_{t}\varepsilon_{t}^{*})\right)^{\prime}\overset{d^{*}}{\rightarrow}\left(  B_{\sigma}(r),B_{2}(r)\right)^{\prime}. \label{bgbip}
			\end{align}
		\end{lemma}
		
		Based on Lemma \ref{lemma 3}, it is evident that the unit root pseudo series $\{y_t^*\}$, generated by the ABB procedure, adaptively retains the same unconditional heteroskedasticity and mixing dependence as in $\{y_t\}$, without assuming any parametric structures on them. Hence, the resampling procedure is data-adaptive.
		
		\begin{remark}
			Dependent wild bootstrap (DWB) of \cite{Shao2010} can also  be adopted to capture serial dependence  {and unconditional heteroskedasticity}, however, it does not work effectively in \textcolor{black}{the LAD} framework.  The key to  the validity of the  bootstrap procedure is to  {ensure that the pseudo series achieves the same asymptotic result}
			as displayed in (\ref{bgbip}). For DWB, the  bootstrap pseudo series drawn from $\{u_t\}_{t=1}^T$  takes the form of  $w_t {u_t}$, where $w_t$  is a random weight drawn from the multivariate normal distribution with mean 0 and covariance matrix $\Sigma$, whose $(i,j)${-th}  element  is given by $ K\left( (i-j)/q_T\right) $ with the  kernel function $K(\cdot)$ and the truncation parameter $q_T$. It is not hard to find that $T^{-1/2}\sum_{t=1}^{\lfloor Tr \rfloor}sgn(w_tu_t)$ differs from $T^{-1/2}\sum_{t=1}^{\lfloor Tr \rfloor}sgn(\varepsilon_t)$ since 	$$E^*\left(\frac{1}{\sqrt{T}} \sum_{t=1}^{\lfloor Tr \rfloor}\left( sgn(\varepsilon_t) -  sgn(w_t {u_t})\right) \right)  = E^*\left(\frac{1}{\sqrt{T}} \sum_{t=1}^{\lfloor Tr \rfloor}2 {I}(w_t < 0) sgn(\varepsilon_t)\right) = \frac{1}{\sqrt{T}}\sum_{t=1}^{\lfloor Tr \rfloor}sgn(\varepsilon_t).$$
		\end{remark}
		
		Before illustrating the asymptotic properties of the bootstrap test statistics, we first present the limiting distribution of the t-ratio test statistic in the following corollary, which follows directly from Theorem \ref{theorem 2}.
		
		\begin{corollary}
			Suppose Assumptions \ref{assum 1}-\ref{assum 3} hold true,  let $y_{t}$ be determined by (\ref{AR})	and (\ref{HUF}) with $\gamma_{0}$ satisfying (\ref{ALT}), we then have
			$$
			t_T \overset{d}{\rightarrow}\mathcal{B}(c,\sigma):=\left(\int_{0}^{1}\mathcal{J}_{c,\sigma}^2(r)dr - \left(\int_{0}^{1}\mathcal{J}_{c,\sigma}(r)dr \right)^2  \right)^{1/2}\left[-2f(0)c+\frac{\int_{0}^{1}\mathcal{J}_{c,\sigma}(r)dB_{2}(r)+\Gamma}{\int_{0}^{1}\sigma^{-1}(r){\ \mathcal{J}_{c,\sigma}^{2}(r)}dr}\right].$$
		\end{corollary}

		When  $c = 0$, $\mathcal{J}_{c,\sigma}(r)$ becomes $B_{\sigma}(r)$ and the limiting null distribution $\mathcal{B}(0,\sigma)$  of  $t_T^d$ is obtained. The following Theorem \ref{theorem 3} shows the asymptotic results of the bootstrap test statistics $L_T^*$ and $t_T^*$.

		\begin{theorem}
			\label{theorem 3}
			Suppose Assumptions \ref{assum 1}-\ref{assum 3} hold true,
			let $y_{t}$ be determined by (\ref{AR})	and (\ref{HUF}) with $\gamma_{0}$  satisfying (\ref{ALT}), ${1}/{b_{T}}+b_{T}/{T}\rightarrow0$ as
			$T\rightarrow\infty$, we then have $ L_{T}^{*}\overset{d^*}{\rightarrow} \mathcal{D}(\sigma)$ and  $t_T^* \overset{d^*}{\rightarrow} \mathcal{B}(0,\sigma)$.
		\end{theorem}

		The bootstrap  test statistics $L_{T}^{*}$ and $t_T^*$ converge in distribution to the  the same limits as $L_{T}$ and $t_T$ under the null hypothesis, respectively. 
		We highlight that the proposed ABB procedure can be easily modified to approximate the null distribution of other unit root test statistics, such as any other LAD-based test statistics constructed from $\hat{\gamma}_{LAD}-1$ and the classic DF test, see a similar argument in \cite{Paparoditis2003}.
		
		\subsection{Feasible adaptive block bootstrap}
		However, the bootstrap test statistics $L_{T}^{*}$ and $t_T^*$ are unobtainable since the true sequence of volatility $\left\{  \sigma_{t}\right\}  _{t=1}^{T}$ is  unknown. To produce feasible testing procedures, $\sigma_{t}$ is required to be estimated in advance. Here we propose to estimate unconditional volatility nonparametrically and then implement the feasible bootstrap algorithm by replacing $\sigma_{t}$ with its nonparametric estimate. 
		
		Specifically, let  $\hat{u}_{t}=y_{t}-\hat{\gamma}_{LAD}y_{t-1}$ be the LAD residual,  {and} denote 
		\begin{equation}
			\hat{\sigma}_{t}=\sum_{s=1}^{T}w_{t,s}\left\vert \hat{u}_{s}\right\vert,
			\label{npe}%
		\end{equation}
		where $w_{t,s}=\left(  \sum_{s=1}^{T}k_{t,s}\right)  ^{-1}k_{t,s}$, $k_{t,s}=k\left(  {\left( t-s\right) }/{Th}\right)  $ is  {the} kernel function and $h$ is  {the} bandwidth parameter shrinking to zero at an appropriate rate. 
		The role of $\hat{\sigma}_{t}$ is to deputize for $\sigma_{t}$, which is the local mean of $\{\left\vert u_{t}\right\vert \}_{t=1}^T$ based on the following assumption. 
		
		\begin{assum}
			\label{assum 4} \textit{$E \left\vert \varepsilon_{t}\right\vert=1$.}
		\end{assum}
		
		Assumption \ref{assum 4} is imposed for the identification of ${\sigma}_{t}$, and it is substituted by  $E (\varepsilon_{t}^{2})  =1$ or $E\left(\varepsilon_{t}^{2}|\mathcal{F}_{t-1}\right) =1$ in the LS framework, see \cite{Cavaliere2005} and \cite{Cavaliere2007}. Note that Assumption \ref{assum 4} can be  relaxed to $E|\varepsilon_t| = c$, where $c$ is an arbitrary positive value.  The following two assumptions are required to ensure that the nonparametric estimator $\hat{\sigma}_{t}$ is asymptotically equivalent to $\sigma_{t}$.
		
		\begin{assum}
			\label{assum 5}
			\textit{(i) $k(\cdot)$ is a bounded continuous nonnegative function defined on the real line with
				$\int_{-\infty}^{\infty
				}k(v)dv=1$; (ii) $h + 1/(Th^{9/2})\rightarrow 0$  as $T\rightarrow \infty$.}
		\end{assum}
		
		\begin{assum}
			\label{assum 6} \textit{There exists $ p > 4$   such that $ E|\varepsilon_t|^{p}  <\infty $ and the mixing coefficients $\{\alpha_m\}$  satisfy   $\sum_{m=1}^{\infty}\alpha_m^{1/\beta - 1/p}<\infty $ for some $4 < \beta < p$.}
		\end{assum}

		Assumption \ref{assum 5}(i) is a standard assumption on kernel functions in \textcolor{black}{the} nonparametric literature, and holds for \textcolor{black}{many} commonly used kernels such as Epanechnikov, Gaussian and uniform ones. Assumption 5(ii) requires that the bandwidth parameter $h$ shrinks to zero at a slower rate than $T^{-2/9}$, ensuring that the feasible bootstrap test statistics have the same limiting distributions as the infeasible ones. It can be relaxed if  $\{\varepsilon_t \}$ is an m.d.s. or i.i.d. random variable. Assumption \ref{assum 6} imposes a stronger moment condition on  $\{\varepsilon_t \}$ than Assumption \ref{assum 2} due to  {the} nonparametric estimation of unconditional volatility,  {some} similar treatments can be found in \cite{Harvey2018} and \cite{Zhu2019}.
		
		The steps of feasible bootstrap procedure are totally the same as the infeasible ones outlined in Subsection \ref{Sec:infeasible ABB} except that $\sigma_{t}$ is replaced by $\hat{\sigma}_{t}$. Now denote $\check{\varepsilon}_{t}=\hat{u}_{t}/\hat{\sigma}_{t}$ and let $\{\check{y}_{t}^{\ast}\}_{t=1}^{T}$ be the pseudo series generated by  $\check{y}_t^* = \check{y}_{t-1}^* + \check{u}_t^*$ where $\check{u}_t^* = \hat{\sigma}_t\check{\varepsilon}_t^*$ and $\check{\varepsilon}_t^*$ is the bootstrap error term drawn from the feasible ABB procedure. The LAD-based bootstrap test statistics $\check{L}_{T}^{*}$ and $\check{t}_T^{*}$ are computed by using the sequence $\{\check{y}_{t}^{\ast}\}_{t=1}^{T}$. 
		The following Lemma \ref{lemma 4} and Theorem \ref{theorem 4} are the feasible versions  corresponding to Lemma \ref{lemma 3} and Theorem \ref{theorem 3}, respectively.
		
		\begin{lemma}
			\label{lemma 4}
			Suppose Assumptions \ref{assum 1},\ref{assum 3}-\ref{assum 6} hold true, let $y_{t}$ be determined by (\ref{AR}) and (\ref{HUF}) with $\gamma_{0}$ satisfying (\ref{ALT}), $1/b_{T}+b_{T}/{T}\rightarrow0$ as $T\rightarrow\infty$, we have
			\begin{align*}
				\frac{1}{\sqrt{T}}\sum_{t=1}^{\lfloor Tr \rfloor}\left(  \check{\varepsilon}_{t}^{*
				},sgn(\check{\varepsilon}_{t}^{*})\right)  ^{\prime}\overset{d^{*}%
				}{\rightarrow}\left(  B_{1}(r),B_{2}(r)\right)  ^{\prime};
				\\ \frac{1}{\sqrt{T}}\sum_{t=1}^{\lfloor Tr \rfloor}\left(  \hat{\sigma}_{t}\check{\varepsilon}_{t}^{*},sgn(\check{\varepsilon}_{t}^{*})\right)  ^{\prime}\overset{d^{*}}{\rightarrow}\left(B_{\sigma}(r),B_{2}(r)\right)
				^{\prime}.
			\end{align*} 
		\end{lemma}
		
		\begin{theorem}
			\label{theorem 4} 
			Suppose Assumptions  \ref{assum 1},\ref{assum 3}-\ref{assum 6} hold true, 
			let $y_{t}$ be determined by (\ref{AR})	and (\ref{HUF}) with $\gamma_{0}$ satisfying (\ref{ALT}), ${1}/{b_{T}}+b_{T}/{T}\rightarrow0$ as
			$T\rightarrow\infty$, we have $
			\check{L}_{T}^{*}\overset{d^*}{\rightarrow} \mathcal{D}(\sigma)$ and $ \check{t}_T^{*}  \overset{d^*}{\rightarrow} \mathcal{B}(0,\sigma).$
		\end{theorem}
		
		Theorem \ref{theorem 4} implies that replacing $\sigma_{t}$ with $\hat{\sigma}_{t}$ has negligible effects on the asymptotic distribution. Therefore, by comparing  $L_T$ ($t_T$) with the empirical quantile of $\{ \check{L}_{T,j}^{*}\}_{j=1}^{B}$ ($\{ \check{t}_{T,j}^{*}\}_{j=1}^{B}$ ), we can judge whether the null hypothesis is rejected or not.

		\subsection{Choices of tuning parameters}
		\label{Sec:choice of para}
		Note that our asymptotic results for the ABB procedure hold true for any bandwidth $h$ satisfying $h + 1/(Th^{9/2})\rightarrow 0$ and  any block length $b_{T}$ satisfying $1/b_{T}+b_{T}/{T}\rightarrow0$, which, however, does not give clues on how to choose $h$ and $b_T$ in practice. Here we apply the leave-one-out cross-validation (CV) \citep[see,][]{Wong1983}
		to select $h$, and the  chosen value is given by
		\begin{equation}
			h_{cv}=\underset{h}{\arg\min}\sum
			_{t=1}^{T}\left(  \left\vert \hat{u}_{t}\right\vert -\hat{\sigma}_{-t}(h)\right)^{2}, \label{CV}%
		\end{equation}
		where $\hat{\sigma}_{-t}(h)$ is the leave-one-out estimate,  calculated in the same way as $\hat{\sigma}_{t}$ except that the summation in (\ref{npe}) is taken for $s\neq t$, and the    magnitude orders of candidates are specified as $T^{-1/5}$, which is the optimal order that minimizes the  asymptotic mean integrated squared error.
		
		For the choice of the block length, we employ the data-driven method suggested by \cite{Hall1995} to determine $b_{T}$.  Let $\mathscr{S}$ denote the set of all $T-m$ runs of length $m$ obtained from the standardized residuals $\{\hat{\varepsilon}_t\}_{t=1}^T$. 
		Let $\hat{\psi}_{m,i}(b)$ ($1 \leq i \leq T-m$) denote the bootstrap estimates of $\psi=\operatorname{var}\left(T^{-1 / 2} \sum_{t=1}^T \varepsilon_t\right)$  computed from  the $T-m$  {elements} in $\mathscr{S}$,  {using block size $b$}. Let  $\hat{\psi}_T( {b_\ast})$  be the estimate of $\psi$ computed from the entire data set $\{\hat{\varepsilon}_t\}_{t=1}^T$, using a plausible pilot block length  {$b_\ast$}. Then the optimal block length\footnote{The block bootstrap technique applied here is tailored to restore the asymptotic behavior of the partial sum of ${\varepsilon_t}$, with focus on the variance. It is well documented that the optimal block length (in terms of minimizing the mean squared error) is proportional to $T^{1/3}$ when handling bias or variance, see  \cite{Lahiri2003} for a review}.  $b_{opt}$ is given by $(T/m)^{1/3}b_m$, where $b_m$ is the optimal solution that minimizes $\sum_{i}\left(\hat{\psi}_{m,i}(b) -\hat{\psi}_T( {b_\ast})\right)^2$. This procedure may be iterated if desired, replacing the original pilot  {$b_\ast$ by  $b_{opt}$}. 
		
		\section{Extensions}
		\label{section extension}
		The autoregressive model without intercept is  too restrictive in empirical application.  Therefore, in this section we show how to extend our testing procedure to further  include deterministic components.   Specifically, suppose that the observed data  $\{x_{t}\}_{t=1}^{T}$ are generated by 
		\begin{align}
			x_t = \mu^\prime d_t + y_t,
			\label{AR with mean}
		\end{align}
		where $y_t$ is generated according to \eqref{AR}, and $d_t$ is a vector of deterministic functions of time $t$, with $\mu$ being a conformable parameter vector.  Two leading cases for  $d_t$ in the unit root literature are  $d_t = 1$ (constant mean $\mu$) and $d_t = (1,t)^\prime$ (linear trend $\mu^\prime d_t = \mu_1 + \mu_2 t$). Following \cite{Elliott1996}, we consider the following regression
		\begin{align} 
			\Delta x_t +  {\frac{\bar{c}}{T}} x_{t-1} &= \mu^\prime \left(\Delta d_t +  {\frac{\bar{c}}{T}}d_{t-1}\right)+u_t, t = 2,\cdots,T,
		\end{align}
		where $\bar{c}$ is a preset constant. Let $\varrho_t=\Delta d_t + \frac{\bar{c}}{T} d_{t-1}$, then the OLS estimator  for  $\mu$  is given by  
		\begin{align}
			\label{olsmu}
			\hat{\mu}(\bar{c})  = \left(\sum_{t=1}^{T}\varrho_t\varrho_t^{\prime}\right)^{-1} 
			\sum_{t=1}^{T}\varrho_t\left(\Delta x_t + \frac{\bar{c}}{T}x_{t-1} \right)^\prime,
		\end{align}
		and  the  demeaned/detrended series of  $\{x_{t}\}_{t=1}^{T}$ is $\{y_t^d \}_{t=1}^T$ where $y_t^d = x_t - \hat{\mu}(\bar{c})^\prime d_t$.
		
		\begin{remark}
			In the  demeaning/detrending procedure,  the OLS estimator $\hat{\mu}(\bar{c}) $ can be substituted by a weighted LS estimator or an LAD-based estimator but without much benefit  only at cost of more tedious proofs.  Our aim is to  study how to extend the proposed  ABB procedure to include deterministic components and demonstrate its validity in  the extended model. Hence, using the OLS estimator $\hat{\mu}(\bar{c})$ is sufficient to  achieve this goal.   In addition, the value of  $\bar{c}$  is required to be pre-specified in advance according to the forms of $d_t$. \cite{Elliott1996} suggested  $\bar{c}=7$ for a constant mean  and  $\bar{c}=13.5$ for a linear trend, which  {have been demonstrated}  to perform well in empirical  {studies},  see e.g., \cite{Ng2001} and \cite{Boswijk2018}.
		\end{remark}
		
		Based on $\{y_t^d\}_{t=1}^{T}$,  the LAD estimator can be obtained by solving the following minimization problem:
		\begin{equation}
			\label{LAD obj mean}
			\hat{\gamma}_{LAD}^d =\arg\min_{\gamma}\sum_{t=1}^{T}\left\vert
			y_{t}^d-\gamma y_{t-1}^d\right\vert .
		\end{equation}
		
		Let $L_T^d=T(\hat{\gamma}_{LAD}^d-1) $ and $
		t^d_T = 2\widehat{f(0)}\left(Y^{d\prime}_{-1} P_C Y^d_{-1} \right)^{1/2}\left(\hat{\gamma}_{LAD}^d - 1 \right)$, where $Y^d_{-1}$ is the vector of lagged dependent variables $\{y_t^d\}$.  The following Theorem \ref{theorem 5} provides  the limiting distributions of both $L_T^d$ and $t_T^d$ under local-to-unity  alternative hypothesis.

		\begin{theorem}
			\label{theorem 5} Suppose
			Assumptions \ref{assum 1}-\ref{assum 3} hold true, let $x_{t}$  be generated according to (\ref{AR with mean}) with $y_t$ determined by (\ref{AR}) and (\ref{HUF}) and  $\gamma_{0}$ satisfying (\ref{ALT}), we have
			\begin{itemize}
				\item[(a)] if $d_t =1$, then 
				$L_T^d\overset{d}{\rightarrow} -c+\mathcal{D}(c,\sigma)$ and $t_T^d \overset{d}{\rightarrow} \mathcal{B}(c,\sigma)$;
				\item[(b)] if $d_t = (1,t)^\prime$, then  $L_T^d\overset{d}{\rightarrow}-c+\mathcal{D}^d(\bar{c},c,\sigma)$ and $t_T^d \overset{d}{\rightarrow} \mathcal{B}^d(\bar{c},c,\sigma) := 2f(0)\mathcal{G}^{1/2}(\bar{c},c,\sigma)(-c+\mathcal{D}^d(\bar{c},c,\sigma))$,
				where 	{\small\begin{align*}
						\mathcal{D}^d(\bar{c},c,\sigma) &:= \left[{2f(0)\left(  \int_{0}^{1}\sigma^{-1}\left(r\right) \mathcal{J}^{2}_{c,\sigma}\left(r\right)dr + \mathcal{M}^2(\bar{c},c,\sigma) \int_{0}^{1}\sigma^{-1}(r) r^2dr - 2\mathcal{M}(\bar{c},c,\sigma)\int_{0}^{1}\sigma^{-1}(r)r\mathcal{J}_{c,\sigma}(r)dr \right)}\right]^{-1}
						\\ & \times
						\left[\int_{0}^{1}\mathcal{J}_{c,\sigma}(r)dB_2(r) -\mathcal{M}(\bar{c},c,\sigma)\int_{0}^{1}rdB_2(r) +\Gamma  - 2f(0)\mathcal{M}(\bar{c},c,\sigma)\int_0^1\sigma^{-1}(r)(1+cr)\mathcal{J}_{c,\sigma}(r)dr
						\right. 
						\\& \left.		+2f(0)\mathcal{M}^2(\bar{c},c,\sigma)\int_{0}^{1}\sigma^{-1}(r)r(1+cr)dr\right],
				\end{align*}  }
				$ \mathcal{M}(\bar{c},c,\sigma) = \left( \int_{0}^{1}\left({ 1+\bar{c}r }\right)^2dr  \right)^{-1} \int_{0}^{1}{\left( 1+\bar{c}r\right) }\left( d\mathcal{J}_{c,\sigma}(r) + \bar{c}\mathcal{J}_{c,\sigma}(r)dr \right)$, and
				$$
				\mathcal{G}(\bar{c},c,\sigma) = \int_0^1\mathcal{J}^2_{c,\sigma}(r)dr + \frac{1}{3}\mathcal{M}^2(\bar{c},c,\sigma )- 2\mathcal{M}(\bar{c},c,\sigma)\int_0^1\mathcal{J}_{c,\sigma}(r)rdr - \left( \int_0^1 \mathcal{J}_{c,\sigma}(r)dr - \frac{1}{2}\mathcal{M}(\bar{c},c,\sigma)\right)^2.
				$$
			\end{itemize}
			
		\end{theorem}
		The forms of the limiting distributions   of the $L_T^d$ and $t_T^d$ statsitics depend on the specific $d_t$. When  $d_t = 1$, the asymptotic results of $L_T^d$ and $t_T^d$ are not affected by  the replacement of the unknown $\mu$ with  its estimator $\hat{\mu}(\bar{c})$.
		However, the situation is quite different when $d_t$ includes a linear trend.	 In this case, the limiting distribution  is apparently  affected by the  estimator  $\hat{\mu}(\bar{c})$, and $\mathcal{D}^d(\bar{c},c,\sigma)$  involves $\bar{c}$  in a highly nonlinear form.  Fortunately, we do not have to  concern  with the explicit form of the limits of  $L_T^d$ and $t_T^d$ due to the readily available bootstrap algorithm.

		Next, we apply the bootstrap  algorithm  for  the extended models. The adaptive  {block} bootstrap procedure  is defined almost identical to the one discussed in the previous sections.  The only difference here is that we employ the \textcolor{black}{demeaned/}detrended series $\{y_t^d\}_{t=1}^{T}$  to  replace the  observed series $\{y_{t}\}_{t=1}^{T}$ \textcolor{black}{throughout the entire} procedure. Of course, the unknown $\sigma_t$  is still estimated by using the nonparametric kernel method in   (\ref{npe}), but instead based on the \textcolor{black}{demeaned/}detrended LAD residual $\hat{u}_{t}^d=y_t^d-\hat{\gamma}_{LAD}^d y_{t-1}^d$.
		Denote the feasible  bootstrap test statistics calculated using the bootstrap sequence as $\check{L}_{T}^{d*}$ and $\check{t}_T^{d*}$.  The following  theorem presents their asymptotic properties.

		\begin{theorem}
			\label{theorem 6} 
			Suppose	Assumptions \ref{assum 1},\ref{assum 3}-\ref{assum 6} hold true, let $x_t$ be generated from (\ref{AR with mean}) with	$y_{t}$ determined by (\ref{AR}) and	(\ref{HUF})  and $\gamma_{0}$ satisfying (\ref{ALT}),  $1/b_{T}+b_{T}/{T}\rightarrow 0$ as $T\rightarrow\infty$,
			\begin{itemize}
				\item[(a)] if $d_t =1$, then $\check{L}_{T}^{d\ast}\overset{d^*}{\rightarrow} \mathcal{D}(\sigma)$ and $\check{t}_{T}^{d\ast} \overset{d^*}{\rightarrow}  \mathcal{B}(0,\sigma)$;
				
				\item[(b)] if $d_t = (1,t)^\prime$, then $\check{L}_{T}^{d\ast}\overset{d^*}{\rightarrow} \mathcal{D}^d(\bar{c},0,\sigma)$ and $\check{t}_{T}^{d\ast} \overset{d^*}{\rightarrow}  \mathcal{B}^d(\bar{c},0,\sigma) $. 
			\end{itemize}
		\end{theorem}
		
		Once again,  Theorem \ref{theorem 6} implies  that  the  bootstrap test statistics $\check{L}_{T}^{d*}$ and $\check{t}_{T}^{d*}$ have the same asymptotic null distributions with $L_T^d$ and $t_T^d$ in both cases, which also shows the validity of the ABB procedure.

		\section{Monte Carlo simulations}
		\label{section simu}
		This section investigates the finite-sample performance of the LAD-based bootstrap unit root tests. Specifically, we consider the following data generating process (DGP): 
		\begin{align}
			y_{t}=\exp\left(  -\frac{c}{T}\right)  y_{t-1}+\sigma_{t}\varepsilon_{t},\text{
			}\varepsilon_{t}=\theta\varepsilon_{t-1}+\phi\eta_{t-1}+\eta_{t},
			\label{DGPs}
		\end{align}
		where  $\varepsilon_t$ is generated by an ARMA(1,1) process, and $\sigma_{t} =\sigma\left(  t/T\right)$ is assumed to take three different forms of structural changes:
		
		\begin{itemize}
			\item[]Case (1)---One shift in volatility:%
			$$\sigma\left(  \tau\right)  =\sigma_{1}+\left(  \sigma_{1}-\sigma_{0}\right)   {I}(\tau>0.5),$$
			\item[]Case (2)---Two shifts in volatility:
			$$\sigma(\tau)=\sigma_{0}+(\sigma_{1}-\sigma_{0}) {I}(0.3<\tau<0.7),$$
			\item[]Case (3)---Smooth changing volatility:
			$$\sigma(\tau)=\sigma_{0}+(\sigma_{1}-\sigma_{0})G\left(  \tau\right),G\left(\tau\right)  =1-\exp\left[  -15\left(  \tau-0.5\right)  ^{2}\right]  .$$
		\end{itemize}

		The volatility dynamics in Case (1) correspond to one abrupt change in volatility from $\sigma_{0}$ to $\sigma_{1}$ at the middle time $\lfloor  0.5T\rfloor  $. In Case (2), the volatility jumps from the level of $\sigma_{0}$ to $\sigma_{1}$ at time $\lfloor  0.3T\rfloor  $ and stays for a period before jumping back at time $\lfloor  0.7T\rfloor  $. In Case (3), the volatility changes smoothly by exhibiting a symmetric `$U$' shape.  Without loss of generality, we specify $\sigma_{0}=1$.  We are interested in evaluating the performance of the proposed tests for heavy-tailed data, while also reporting results under Gaussian conditions as benchmarks, thus considering three types of innovations $\left\{  \eta_{t}\right\}$: (1) The standard normal distribution, $N(0,1)$; (2) The Student-t distribution\ with 3\ degrees of freedom, $t(3)$; (3) The standard double exponential distribution, $DE\left(  0,1\right)  $. It is worth noting that although  $\varepsilon_t$ generated by (\ref{DGPs})  may not satisfy $E \left\vert \varepsilon_{t}\right\vert  =1$  for a majority of $\theta$-$\phi$ combinations, our testing procedure is unaffected by noting that  {${u}_t=\sigma_{t}^{\dagger}\varepsilon_{t}^{\dagger}$ with $\sigma_{t}^{\dagger}=\sigma_{t}E\left\vert \varepsilon_{t}\right\vert$ and $\varepsilon_{t}^{\dagger}=\varepsilon_{t}/E\left\vert \varepsilon_{t}\right\vert$}.
		
		\subsection{Size}
		In this subsection, we first study the empirical size performance of the proposed LAD-based tests.  We denote the LAD-based coefficient-based test statistic and the t-ratio test statistic  as $L^B_T$ and $t^B_T$, respectively, using the proposed ABB procedure to compute corresponding p-values. 
		The results of quantile t-ratio test (denoted as $t_T^{Q}$) developed by \cite{Koenker2004} at median,  and the LAD-based test (denoted as $L_T^M$) established by \cite{Moreno2000}, are also reported for comparison. Neither $t_T^{Q}$ nor $L_T^M$ accounts for unconditional heteroskedasticity. The $t_T^Q$ test can handle high-order autocorrelation, while the $L_T^M$ test is only applicable to AR(1) models with i.i.d. errors. 
		
		\begin{figure}[!htb]
			\centering	
			\includegraphics[width=16cm, trim={0cm 0 0 0}]{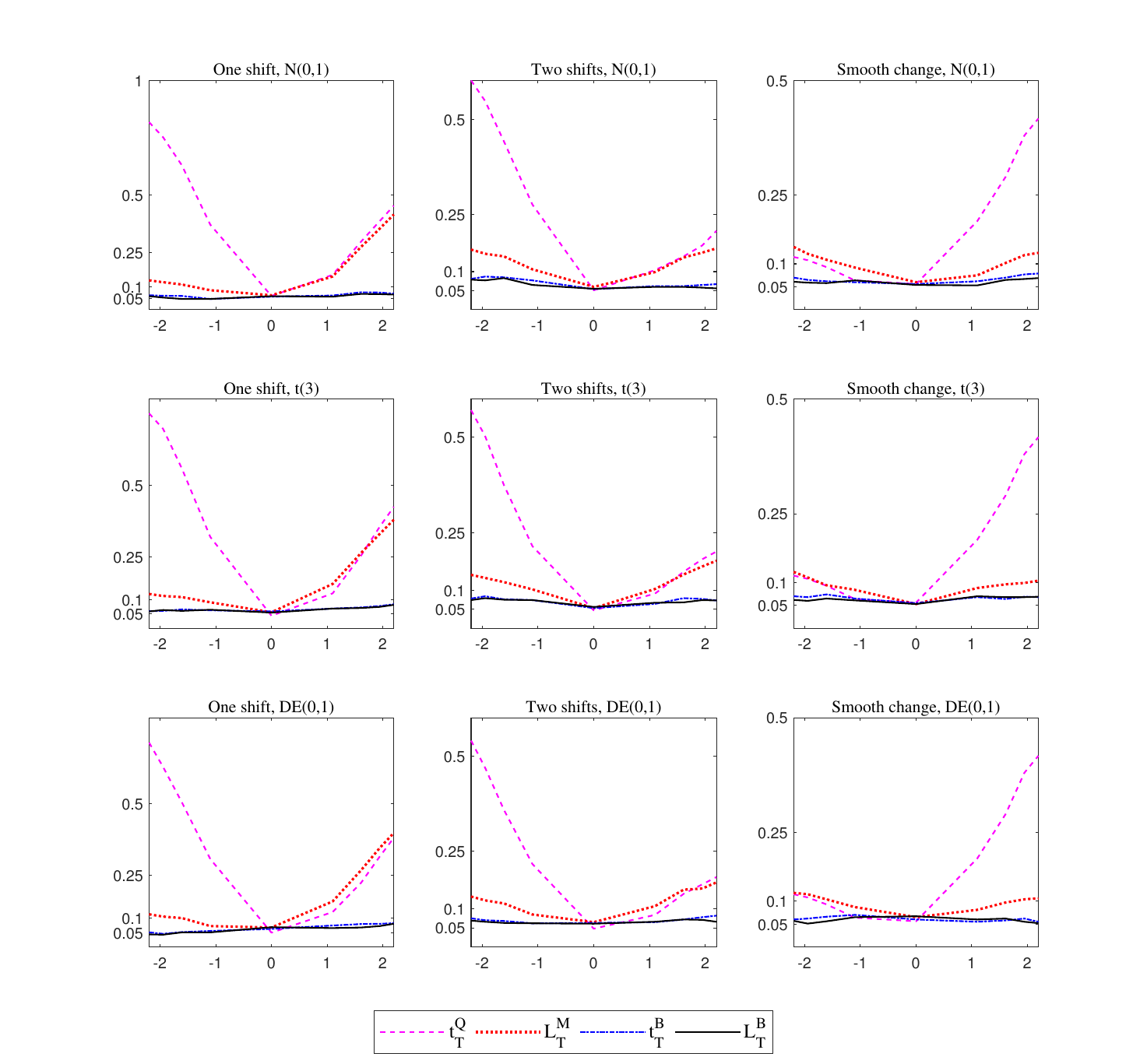}	
			\caption{Empirical sizes (y-axis) of the $t_T^Q$, $L_T^M$, $t_T^B$, and $L_T^B$ tests for $\sigma_1$ between $1/9$ and 9 (x-axis, log scale) when the error term follows an i.i.d process. Nominal size is 0.05 and sample size $T = 100$. }\label{Figure size plot iid}
		\end{figure}
		
		For the error $\varepsilon_t$ in (\ref{DGPs}),  we  consider the following cases: (1) an i.i.d sequence, (2) an MA(1) process with $\theta=0.5$ and $\phi=0$, and (3) an AR(1) process with $\theta=0$ and $\phi=0.5$.  The results for the cases of i.i.d errors are mainly reported to evaluate the impact of heteroskedasticity uncontaminated by serial dependence. Accordingly,  we set the block length  $b_T = 1$ for the ABB procedure and let the lag order  $p = 0$ for the $t_T^Q$ test to eliminate any potential impacts caused by incorrect  tuning parameters choices. For the cases of dependent errors, the block length $b_{T}$ is selected by using the data-driven method proposed in Subsection \ref{Sec:choice of para} and the lag length $p$  is chosen by re-scaled MAIC rule  \citep[see,][]{Cavaliere2015}
		with $p\leq \lbrack12(T/100)^{0.25}]$. 
		The t-ratio tests $t_T^B$ and $t_T^Q$ require estimation of the density function, which are calculated  by using  the classic kernel  method with  a Gaussian kernel and Silverman's rule of thumb bandwidth \citep{Zhao2014}.
		We generate 1000 datasets of random sample $\left\{  y_{t}\right\}_{t=1}^{T}$ for $T=100$ and $250$, and employ $B=499$ iterations for those  bootstrap-based tests. All testing results are run at  the 5\% significance level.  
		
		\begin{figure}[!htb]
			\centering	
			\includegraphics[width=16cm, trim={0cm 0 0 0}]{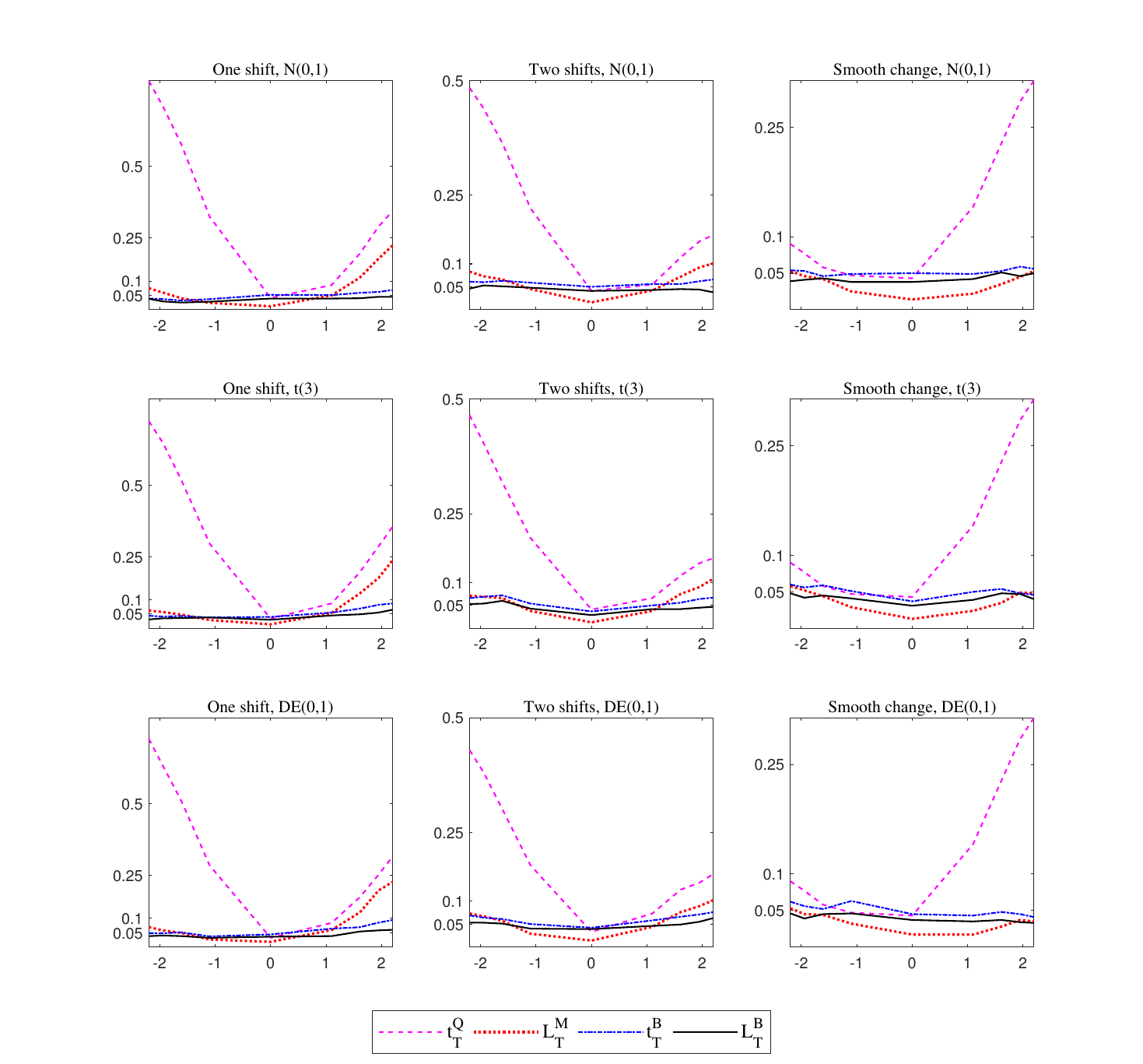}	
			\caption{Empirical sizes (y-axis) of the $t_T^Q$, $L_T^M$, $t_T^B$, and $L_T^B$ tests for $\sigma_1$ between $1/9$ and 9 (x-axis, log scale) when the error term follows an MA(1) process. Nominal size is 0.05 and sample size $T = 100$. }\label{Figure size plot ma}
		\end{figure}

		Figures \ref{Figure size plot iid}-\ref{Figure size plot ar} plot  the empirical sizes of the four considered LAD-based tests with respect  to  $\sigma_1 \in \{1/9,1/7,1/5,1/3,1,3,5,7,9 \}$  for the cases where the error follows an i.i.d. sequence, an MA process, and an AR process, respectively, with $T = 100$. The corresponding figures for $T = 250$ are provided in the supplementary material. We first focus on Figure \ref{Figure size plot iid}. It can be observed that the $t_T^Q$ and $L^M_T$ tests are severely oversized under heteroskedasticity. For example, in the cases of a single shift in volatility, the sizes of the $t_T^Q$ test exceeds 0.5  when $ \sigma_1 = 1/9$. Compared to the results for $T = 250$ in the supplementary material,  increasing  the sample size does not  appear to mitigate the size distortion caused by heteroskedasticity. In contrast,  it is evident that the proposed ABB procedure effectively corrects the size distortion in $t_T^Q$ and $L^M_T$ tests in the presence of time-varying volatility. The sizes of the proposed bootstrap $t_T^B$ and $L_T^B$ tests are largely satisfactory all the time, though they also exhibit  slightly oversized distortions. And the size performance improves to some extent as the sample size $T$ increases from 100 to 250.  In addition, we find that the $L_T^B$ test  generally performs slightly better than the $t_T^B$ test, which may be attributed to the fact that the $L_T^B$ statistic does not require  estimating the density function $f(0)$.
		
		Now we turn to Figures \ref{Figure size plot ma} and \ref{Figure size plot ar}. First, it is evident that the size properties of the  $t_T^B$ and $L_T^B$ tests are satisfactory across different $\sigma_1$ values, error distributions, and  volatility types,  remaining close to the nominal level 0.05, which indicates that the proposed LAD-based bootstrap tests can control size well in the presence of both heteroskedasticity and serial dependence. Then, we can further observe that the $t_T^Q$ test still exhibits severe upward size distortion, as indicated by the similarity of its size curve patterns to those shown in Figure \ref{Figure size plot iid}. We point out that the $t_T^Q$ test allows for high-order correlation in the data, consequently the size distortion is only attributed to time-varying volatility. Next, for the $L_T^M$ test, the dependence in the error results in a downward bias in its size. This is evident from the observation that the empirical sizes of the $L_T^M$ test are largely underestimated, with all values around 1\% under homoskedasticity ($\sigma_1 = 1$). In the presence of both heteroskedasticity and serial dependence, the distortion in the size of the $L_T^M$ test is ambiguous and varies depending on specific scenarios. For example, the $L_T^M$ test  exhibits undersized distortion in most cases when the error follows an AR(1) process, but shows upward distortion when the error follows an MA(1) process with one large upward shift in volatility.

		\begin{figure}[!htb]
			\centering	
			\includegraphics[width=16cm, trim={0cm 0 0 0}]{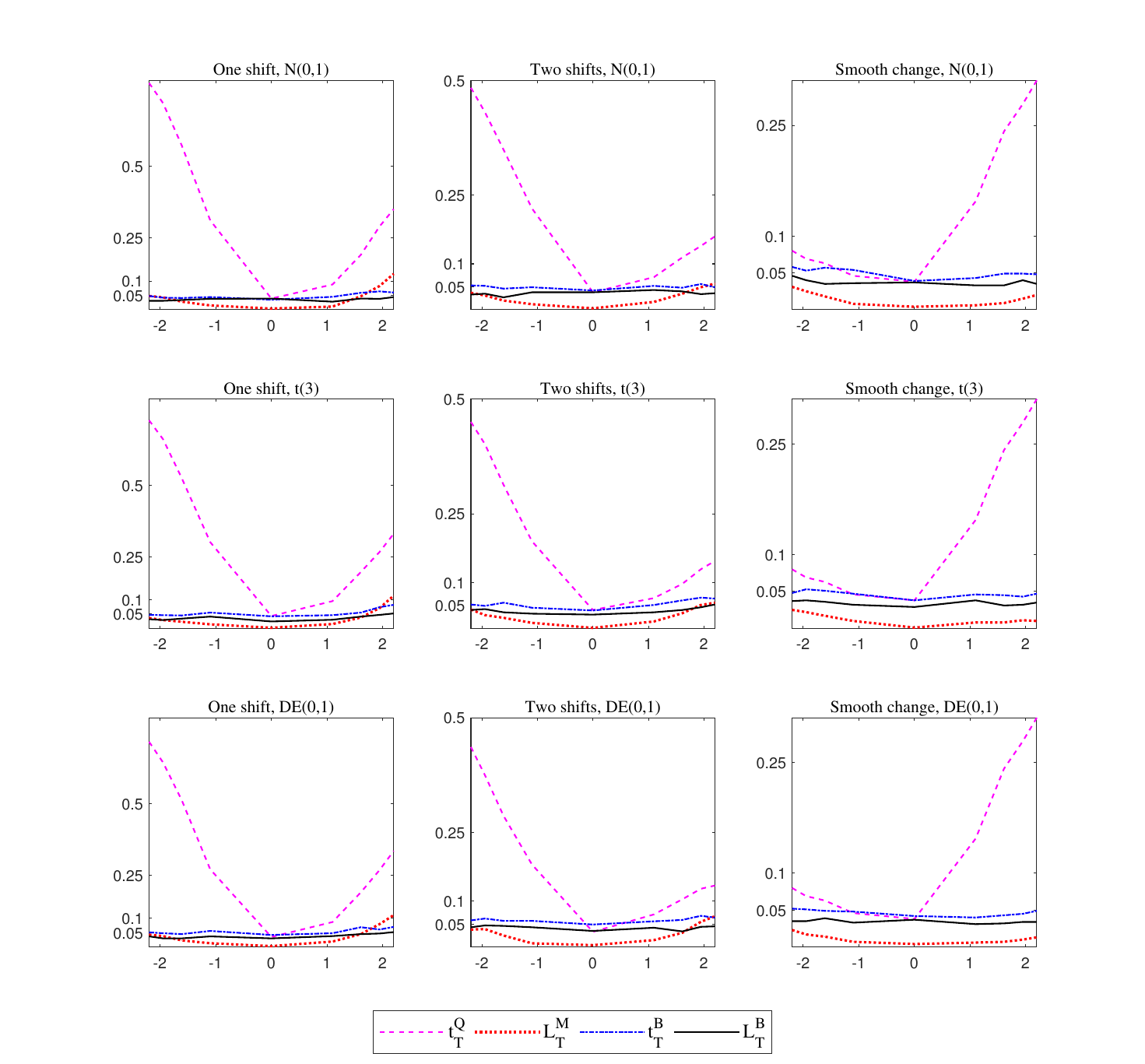}	
			\caption{Empirical sizes (y-axis) of the $t_T^Q$, $L_T^M$, $t_T^B$, and $L_T^B$ tests for $\sigma_1$ between $1/9$ and 9 (x-axis, log scale) when the error term follows an AR(1) process. Nominal size is 0.05 and sample size $T = 100$. }\label{Figure size plot ar}
		\end{figure}

		\subsection{Power} 
		In this subsection, we further investigate the testing power of the proposed LAD-based tests in comparison with popular existing approaches that theoretically accommodate both time-varying variance and serial correlation. Specifically, we consider the following benchmarks:
		
		(1) $MZ^{T}$: The time-transformed $MZ_\alpha$ test proposed by  \cite{Cavaliere2008b}.
		
		(2) $MZ^{B}$: The $MZ_\alpha$ test using the critical values from wild bootstrap procedure  { \citep{Cavaliere2008a,Cavaliere2009}}.
		
		(3) $MZ^{S}$: The $MZ_\alpha$ test using simulated critical values based on the variance profile estimator, see \cite{Cavaliere2007}.
		
		(4) $MZ^{M}$: The $MZ_\alpha$ test calculated by purging heteroskedasticity, see \cite{Beare2018}.
		
		(5) $ALR$: The adaptive likelihood ratio test using wild bootstrap critical values proposed by \cite{Boswijk2018}.

		The first four tests  take the form of the classical M-statistics developed by \cite{Perron1996} and \cite{Ng2001}\footnote{The M statistics generally comprise three types of test statistics, which are defined as		
			$$	MZ_{\alpha}:=\frac{T^{-1}y_{T}^{2}-s_{ar}^{2}(p)}{2T^{-2}\sum_{t=1}%
				^{T}y_{t-1}^{2}},\ MSB := \left(  T^{-2}\sum_{t=1}	^{T}y_{t-1}^{2}/s_{ar}^{2}(p)\right)  ^{1/2},\ MZ_{t}:= MZ_{\alpha}\times MSB, $$
			where $s_{ar}^{2}(p)=\hat{\sigma}^{2}/(1-\sum_{i=1}^{p}\hat{\beta}_{i})$,$\left\{  \hat{\beta}_{i}\right\}  _{i=1}^{p}$ and $\hat{\sigma}^{2}$ are the OLS slope and variance estimators from the regression $\Delta y_{t}=\beta_{0}y_{t-1}+\sum_{i=1}^{p}\beta_{i}\Delta y_{t-i}+e_{t}$. Because the performance of three M tests are quite similar, for brevity we only report the results of $MZ_\alpha$-type tests here.}.   The lag orders in $MZ$-type tests are  selected by the re-scaled MAIC rule.  For the kernel bandwidth used in computing the test statistic $MZ^{M}$ of \cite{Beare2018}, we set it to 0.1,  as this value yields the best size and power  according to his simulation results. The $ALR$ test is a powerful unit root test which can reach the asymptotic power envelope under Gaussianity.  We use the exponential kernel
		$k(x) = exp(-5|x|)$ with the window width selected by their cross-validation procedure. In this subsection, we do not compare the results of the $t_T^Q$ and $L_T^M$ tests since they do not exhibit adequate size control and are not expected to provide reliable inference under heteroskedasticity.

		\begin{table}[!htp]
			\centering
			\caption{Empirical powers under heteroskedasticity with i.i.d. errors, $\sigma_1 = 5$, and $c = 10$.}
			\label{Table power hetero}%
			\begin{tabular}{ccccccccc}
				\toprule 
				$\sigma_t$           & T                    & $MZ^T$ & $MZ^B$ & $MZ^S$ & $MZ^M$ & $ALR$ & $t_T^{B}$ & $L_T^B$ \\ \hline
				\multicolumn{1}{l}{} & \multicolumn{1}{l}{} & \multicolumn{7}{c}{N(0,1)}                                     \\
				One shift            & 100                  & 0.299 & 0.463 & 0.293 & 0.564   & 0.605  & 0.554     & 0.578   \\
				& 250                  & 0.371 & 0.430 & 0.386 & 0.531   & 0.615  & 0.511     & 0.554   \\
				Two shifts           & 100                  & 0.296 & 0.483 & 0.328 & 0.497   & 0.811  & 0.805     & 0.763   \\
				& 250                  & 0.348 & 0.504 & 0.444 & 0.488   & 0.826  & 0.798     & 0.765   \\
				Smooth change        & 100                  & 0.483 & 0.668 & 0.561 & 0.498   & 0.752  & 0.641     & 0.686   \\
				& 250                  & 0.449 & 0.635 & 0.601 & 0.544   & 0.766  & 0.620     & 0.670   \\ \hline
				\multicolumn{1}{l}{} & \multicolumn{1}{l}{} & \multicolumn{7}{c}{t(3)}                                       \\
				One shift            & 100                  & 0.333 & 0.494 & 0.247 & 0.538   & 0.717  & 0.782     & 0.809   \\
				& 250                  & 0.361 & 0.464 & 0.327 & 0.609   & 0.704  & 0.832     & 0.839   \\
				Two shifts           & 100                  & 0.337 & 0.524 & 0.286 & 0.418   & 0.879  & 0.947     & 0.938   \\
				& 250                  & 0.280 & 0.550 & 0.393 & 0.544   & 0.907  & 0.959     & 0.964   \\
				Smooth change        & 100                  & 0.456 & 0.692 & 0.445 & 0.547   & 0.846  & 0.884     & 0.904   \\
				& 250                  & 0.456 & 0.692 & 0.556 & 0.577   & 0.831  & 0.904     & 0.928   \\ \hline
				\multicolumn{1}{l}{} & \multicolumn{1}{l}{} & \multicolumn{7}{c}{DE(0,1)}                                    \\
				One shift            & 100                  & 0.301 & 0.429 & 0.233 & 0.474   & 0.667  & 0.863     & 0.867   \\
				& 250                  & 0.347 & 0.437 & 0.357 & 0.613   & 0.654  & 0.880     & 0.883   \\
				Two shifts           & 100                  & 0.409 & 0.509 & 0.304 & 0.405   & 0.873  & 0.972     & 0.972   \\
				& 250                  & 0.336 & 0.524 & 0.423 & 0.537   & 0.880  & 0.986     & 0.980   \\
				Smooth change        & 100                  & 0.497 & 0.663 & 0.490 & 0.502   & 0.814  & 0.922     & 0.942   \\
				& 250                  & 0.495 & 0.656 & 0.587 & 0.633   & 0.810  & 0.944     & 0.961  
				\\ \bottomrule
			\end{tabular}
		\end{table}

		For the alternative specification, we set $c  = 10$ in (\ref{DGPs}) and fix $\sigma_1 = 5$. Qualitatively similar conclusions can be drawn for other values of $c$ and $\sigma_1$\footnote{In our supplementary material, we also report the testing results under homoskedasticity.}.
		The empirical sizes of the considered tests are tabulated in the supplementary material, indicating that all tests demonstrate good size control ability in most cases. Tables \ref{Table power hetero}-\ref{Table power ARH} report the empirical powers of the considered tests at the 5\% significance level under local to unity alternative.   
		
		Similarly, let us first focus on the cases of i.i.d errors. Under normality, the $ALR$ test exhibits the highest power in all cases. Note that when volatility has two shifts, the powers of the  proposed $t_T^B$ and $L_T^B$ tests are  close to that of the $ALR$ test. When the error follows heavy-tailed distributions, our two LAD-based bootstrap tests exhibit significant advantages over the other competing methods. For example, the powers of the $t_T^B$ and $L_T^B$ tests are 0.863 and 0.867, respectively, in the  $DE(0,1)$-distributed cases with one shift in volatility and $T = 100$, while the highest power among other competitors is only 0.654. 
		It is worth noting that the powers of all the tests vary to some extent depending on the forms of heteroskedasticity,  consistent with the finding that the asymptotic distributions of the LAD-based test statistics depend on the form of $\sigma(\cdot)$.

		\begin{table}[!htp]
			\centering
			\caption{Empirical powers under heteroskedasticity with  MA(1) errors, $\sigma_1 = 5$, and $c = 10$.}
			\label{Table power MAH}
			\begin{tabular}{ccccccccc}
				\toprule 
				$\sigma_t$           & T                    & $MZ^T$ & $MZ^B$ & $MZ^S$ & $MZ^M$ & $ALR$ & $t_T^{B}$ & $L_T^B$ \\ \hline
				\multicolumn{1}{l}{} & \multicolumn{1}{l}{} & \multicolumn{7}{c}{N(0,1)}                                     \\
				One shift            & 100                  & 0.106 & 0.127 & 0.265 & 0.282   & 0.259  & 0.428     & 0.436   \\
				& 250                  & 0.154 & 0.212 & 0.331 & 0.386   & 0.402  & 0.434     & 0.463   \\
				Two shifts           & 100                  & 0.234 & 0.238 & 0.229 & 0.342   & 0.645  & 0.755     & 0.680   \\
				& 250                  & 0.230 & 0.287 & 0.334 & 0.353   & 0.726  & 0.735     & 0.718   \\
				Smooth change        & 100                  & 0.193 & 0.228 & 0.320 & 0.235   & 0.491  & 0.521     & 0.530   \\
				& 250                  & 0.334 & 0.389 & 0.452 & 0.431   & 0.611  & 0.524     & 0.550   \\ \hline
				\multicolumn{1}{l}{} & \multicolumn{1}{l}{} & \multicolumn{7}{c}{t(3)}                                       \\
				One shift            & 100                  & 0.042 & 0.163 & 0.218 & 0.248   & 0.344  & 0.622     & 0.617   \\
				& 250                  & 0.142 & 0.236 & 0.290 & 0.427   & 0.491  & 0.700     & 0.700   \\
				Two shifts           & 100                  & 0.232 & 0.271 & 0.182 & 0.375   & 0.731  & 0.895     & 0.859   \\
				& 250                  & 0.262 & 0.365 & 0.318 & 0.430   & 0.816  & 0.911     & 0.907   \\
				Smooth change        & 100                  & 0.167 & 0.286 & 0.269 & 0.304   & 0.596  & 0.761     & 0.765   \\
				& 250                  & 0.342 & 0.443 & 0.441 & 0.379   & 0.687  & 0.810     & 0.830   \\ \hline
				\multicolumn{1}{l}{} & \multicolumn{1}{l}{} & \multicolumn{7}{c}{DE(0,1)}                                    \\
				One shift            & 100                  & 0.074 & 0.148 & 0.223 & 0.240   & 0.267  & 0.591     & 0.582   \\
				& 250                  & 0.198 & 0.231 & 0.340 & 0.453   & 0.423  & 0.649     & 0.674   \\
				Two shifts           & 100                  & 0.286 & 0.249 & 0.194 & 0.323   & 0.696  & 0.876     & 0.846   \\
				& 250                  & 0.294 & 0.329 & 0.346 & 0.432   & 0.764  & 0.913     & 0.903   \\
				Smooth change        & 100                  & 0.214 & 0.245 & 0.309 & 0.271   & 0.561  & 0.750     & 0.745   \\
				& 250                  & 0.434 & 0.416 & 0.459 & 0.476   & 0.641  & 0.791     & 0.805  
				\\ \bottomrule
			\end{tabular}
		\end{table}

		Next, we study the empirical powers  in the presence of autocorrelated errors. From Tables \ref{Table power MAH} and \ref{Table power ARH}, it can be observed that the proposed two tests show  pronounced  superiority over  other benchmarks in all heavy-tailed cases. Moreover,  our tests also achieve the highest power in some normal cases. For example, with one shift in volatility and $T = 100$, the $t_T^B$ and $L_T^B$ tests achieve powers of 0.428 and 0.436, respectively, when the error follows an MA(1) process,  outperforming the other competing tests. We note that the poor performance of the $ALR$ and $MZ$ type tests under normality with serial correlation  may result from the inaccurate selection of lag length, given that these tests perform better than our tests in i.i.d. cases where the lag order is correctly specified.

		In the supplementary material, the testing results under homoskedasticity are reported for all considered tests, including the $t_T^Q$ and $L_T^M$ tests. Several interesting findings are summarized as follows: (1)  There is little difference between the powers of the tests $L_T^B$ and $L_T^M$ all the time, which indicates that the heteroskedasticity-accommodating modification does not appear to bring about any significant power loss in $L_T^B$ under homoskedasticity. (2) The t-ratio  $t_T^B$ test exhibits much higher power than the $t_T^{Q}$ test, which can be attributed to the fact that the ABB procedure  fully utilizes the  information at the median, while the sieve bootstrap algorithm for the $t_T^{Q}$ test  does not. (3) The $L_T^B$ test exhibits significantly higher power than $L_T^M$ in the presence of serial dependence, as it benefits from the block bootstrap algorithm specifically designed to account for mixing dependence.

		To sum up, the proposed  LAD-based bootstrap tests exhibit reasonable sizes  in the presence of various unconditional heteroskedasticity and weak dependence, and also  achieve higher  powers than other popular unit root tests when the error term is heavy-tailed distributed.

		\begin{table}[!htp]
			\centering
			\caption{Empirical powers under heteroskedasticity with  AR(1) errors, $\sigma_1 = 5$, and $c = 10$.}
			\label{Table power ARH}
			\begin{tabular}{ccccccccc}
				\toprule 
				$\sigma_t$           & T                    & $MZ^T$ & $MZ^B$ & $MZ^S$ & $MZ^M$ & $ALR$ & $t_T^{B}$ & $L_T^B$ \\ \hline
				\multicolumn{1}{l}{} & \multicolumn{1}{l}{} & \multicolumn{7}{c}{N(0,1)}                                     \\
				One shift            & 100                  & 0.090 & 0.140 & 0.264 & 0.228   & 0.230  & 0.334     & 0.320   \\
				& 250                  & 0.141 & 0.229 & 0.320 & 0.402   & 0.391  & 0.339     & 0.362   \\
				Two shifts           & 100                  & 0.210 & 0.262 & 0.227 & 0.321   & 0.612  & 0.730     & 0.664   \\
				& 250                  & 0.256 & 0.301 & 0.305 & 0.372   & 0.698  & 0.724     & 0.702   \\
				Smooth change        & 100                  & 0.184 & 0.251 & 0.342 & 0.264   & 0.453  & 0.468     & 0.449   \\
				& 250                  & 0.347 & 0.391 & 0.448 & 0.421   & 0.599  & 0.456     & 0.466   \\ \hline
				\multicolumn{1}{l}{} & \multicolumn{1}{l}{} & \multicolumn{7}{c}{t(3)}                                       \\
				One shift            & 100                  & 0.060 & 0.170 & 0.215 & 0.234   & 0.339  & 0.480     & 0.478   \\
				& 250                  & 0.205 & 0.265 & 0.305 & 0.415   & 0.499  & 0.572     & 0.578   \\
				Two shifts           & 100                  & 0.238 & 0.296 & 0.183 & 0.336   & 0.736  & 0.827     & 0.802   \\
				& 250                  & 0.270 & 0.362 & 0.303 & 0.443   & 0.808  & 0.845     & 0.838   \\
				Smooth change        & 100                  & 0.168 & 0.319 & 0.302 & 0.280   & 0.595  & 0.637     & 0.629   \\
				& 250                  & 0.351 & 0.444 & 0.439 & 0.390   & 0.684  & 0.701     & 0.717   \\ \hline
				\multicolumn{1}{l}{} & \multicolumn{1}{l}{} & \multicolumn{7}{c}{DE(0,1)}                                    \\
				One shift            & 100                  & 0.061 & 0.170 & 0.228 & 0.226   & 0.259  & 0.465     & 0.450   \\
				& 250                  & 0.214 & 0.254 & 0.330 & 0.477   & 0.429  & 0.495     & 0.510   \\
				Two shifts           & 100                  & 0.283 & 0.267 & 0.192 & 0.331   & 0.674  & 0.811     & 0.777   \\
				& 250                  & 0.283 & 0.320 & 0.317 & 0.406   & 0.752  & 0.819     & 0.826   \\
				Smooth change        & 100                  & 0.253 & 0.274 & 0.309 & 0.273   & 0.540  & 0.591     & 0.593   \\
				& 250                  & 0.410 & 0.440 & 0.480 & 0.488   & 0.657  & 0.637     & 0.643
				\\ \bottomrule
			\end{tabular}
		\end{table}

		\section{Empirical application}
		\label{section empirical}
		In this section, we apply the proposed   {LAD-based bootstrap} unit root tests to study the validity of the purchasing power parity (PPP) hypothesis in 16 EU countries. The PPP hypothesis is an important cornerstone for many open-economy macroeconomic models. It demonstrates that foreign currencies should have the same purchasing power in the long run in a well-functioning world market, which also indicates that the real exchange rate should be stationary. Therefore, the rejection of the unit root is often used as evidence to support the PPP hypothesis, see \cite{Taylor2001}, \cite{Bahmani2007} and \cite{Boswijk2018}. However, there is conflicting empirical evidence. The studies by \cite{OC1998}, \cite{Galvao2009} and \cite{Bahmani2017} among others  {all} suggest that the issue has not been completely solved yet. Hence we revisit this problem by using the new proposed unit root tests.  {A total of} 240 observations of monthly real effective exchange rates (REERS) for each country were collected over the period from 01/2000 to 12/2019, sourced from the website of the Bank for International Settlements\footnote{http://www.bis.org/statistics/eer.htm} (BIS).
		\begin{figure}[htbp!]
			\centering
			\includegraphics[width=1.0\textwidth, trim=3.2cm 0cm 0cm 0cm]
			{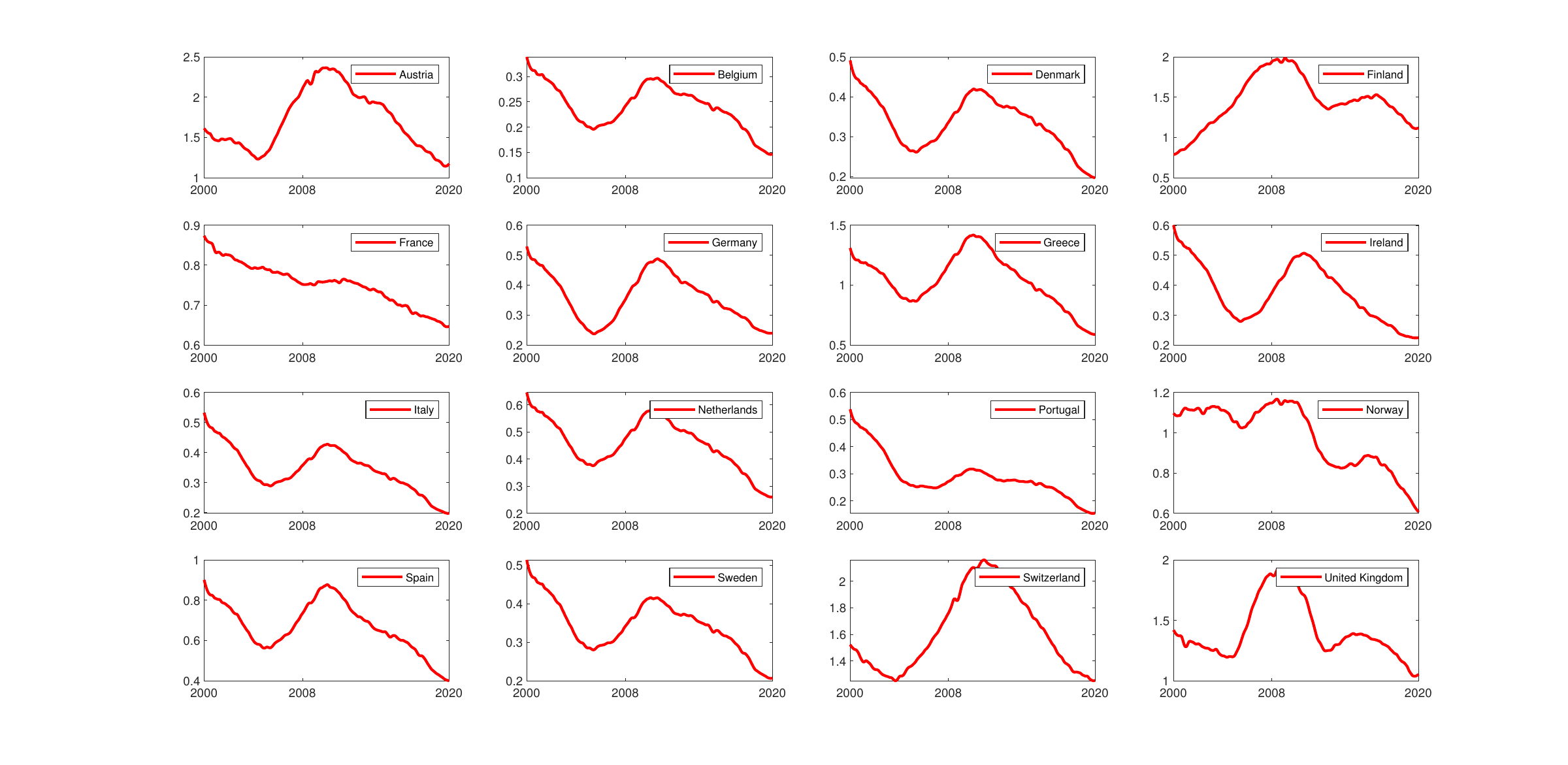}%
			\caption{\small{Non-parametric volatility estimate for 16 EU countries from 01/2000 to 12/2019.}}
			\label{plot volatlity}
		\end{figure}
		
		\begin{figure}[htbp!]
			\centering
			\includegraphics[width=1.0\textwidth, trim=3.2cm 0cm 0cm 0cm]
			{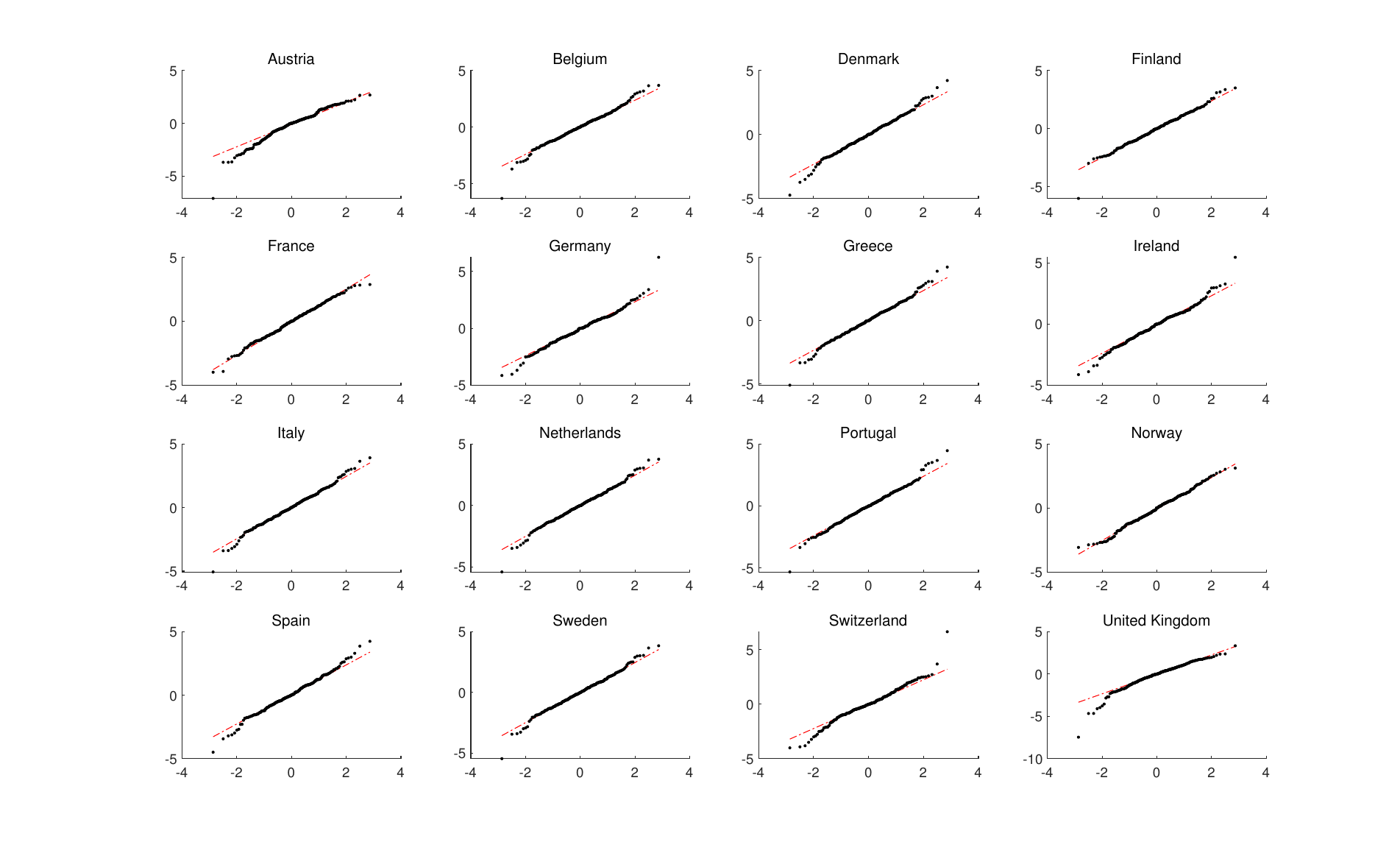}%
			\caption{\small{Q-Q plots of errors, filtered by estimated volatility, for 16 EU countries.}}
			\label{plot-QQplots}
		\end{figure}
		
		Following \cite{Papell1997} and \cite{OC1998}, we adopt the tests of unit roots with an unknown mean, and demean the original data based on (\ref{olsmu}) with $\bar{c} = 7$. The other setups for all the tests are the same as those specified in the simulation section. As the by-products of our tests, nonparametric estimates of the volatilities are also plotted in Figure \ref{plot volatlity}, which shows that the volatility estimates are obviously non-constant, exhibiting inverted "N"-shape patterns for most countries. We also note that the two highly volatile periods observed from the volatility paths coincide with the well-known Dot-com bubble around 2000 and the financial crisis around 2008. This suggests the homoskedastic assumption imposed on traditional unit root tests might be violated. In addition, we also provide the Q-Q plots of the standardized errors (filtered by the estimated volatility) against standard normal distribution in Figure \ref{plot-QQplots}. It is seen that most tail behaviors of the errors  {show clear deviations} from normality. The presence of heteroskedasticity \textcolor{black}{along} with heavy tails,  {as shown} in Figures \ref{plot volatlity}-\ref{plot-QQplots}, suggests  {that it would be} more suitable to employ the newly proposed tests in examining the PPP hypothesis.
		
		\begin{table}[htbp!]
			\centering
			\caption{Unit root testing results for monthly real effective exchange rates of 16 EU countries.}
			\label{Table empirical results}
			\begin{threeparttable}
				\begin{tabular}{cccccccccc}
					\noalign{\global\arrayrulewidth1pt}\hline \noalign{\global\arrayrulewidth0.3pt} & $MZ^T$ & $MZ^B$ & $MZ^S$ & $MZ^{M}$ & $t_T^Q$ & $L_T^M$ & $ALR$ & $t_T^{B}$ & $L_T^B$ \\ \hline
					Austria        & 0.270 & 0.259 & 0.563 & 0.322   & 0.413 & 0.619  & 0.447  & 0.501          & 0.501          \\
					Belgium        & 0.327 & 0.307 & 0.621 & 0.427  & 0.281 & 0.559   &   0.307 & 0.335          & 0.279          \\
					Denmark        & 0.459 & 0.413 & 0.756 & 0.566  & \textbf{0.084} & 0.729 & 
					0.477 &0.812          & 0.870          \\
					Finland        & 0.337 & 0.293 & 0.577 & 0.395   & 0.174 & 0.373 &  0.375& 0.228          & 0.220          \\
					France         & 0.554 & 0.525 & 0.818 & 0.635   & \textbf{0.044} & 0.639 &  0.517 & 0.557          & 0.557          \\
					Germany        & 0.359 & 0.317 & 0.679 & 0.406  & 0.214 & 0.331 &  0.417& 0.206          & 0.186          \\
					Greece         & 0.342 & 0.315 & 0.635 & 0.404  & 0.325 & 0.420 & 0.400 & 0.319          & 0.279          \\
					Ireland        & 0.245 & 0.156 & 0.511 & 0.269  & 0.224 & 0.130 &  \textbf{0.052} &\textbf{0.050} & \textbf{0.066} \\
					Italy          & 0.395 & 0.381 & 0.681 & 0.439   & 0.166 & 0.353 &  0.303& 0.307          & 0.315          \\
					Netherlands    & 0.315 & 0.365 & 0.649 & 0.388  & 0.188 &  0.307 &  0.365& 0.238          & 0.182          \\
					Portugal       &0.453 & 0.465 & 0.707 & 0.561  & 0.309 &  0.928 &  0.499&0.808          & 0.888          \\
					Norway         &0.862  & 0.844 & 0.968 & 0.881   & 0.930 &  0.976  & 0.904 & 1.000          & 0.996          \\
					Spain          & \textbf{0.061} & \textbf{0.094} & 0.152 & 0.200  & 0.116 & \textbf{0.092} &  \textbf{0.032} & \textbf{0.046} & \textbf{0.026} \\
					Sweden         & 0.817  & 0.878 & 0.972 & 0.874  & 0.343 &  0.505  &  0.818 & 0.717          & 0.711          \\
					Switzerland    & 0.694  & 0.756 & 0.916 & 0.666 & 0.281 & 0.333 & 0.160 &0.192          & 0.232          \\
					United Kingdom & 0.279 & 0.317 & 0.673 & 0.271 & 0.731 & 0.641 & 0.235 & 0.491          & 0.493        \\
					\noalign{\global\arrayrulewidth1pt}\hline \noalign{\global\arrayrulewidth0.3pt}  
				\end{tabular}
			\end{threeparttable}
		\end{table}
		
		The p-values for all \textcolor{black}{the} tests are reported in Table \ref{Table empirical results}. It is shown  that in most cases, all \textcolor{black}{the} tests \textcolor{black}{yield} consistent testing results, \textcolor{black}{failing to} reject the unit root hypothesis. Hence the PPP hypothesis does not hold true for most EU countries. \textcolor{black}{The} rejection results are observed for Denmark, France, Ireland and Spain. For Denmark and France, the test $t_T^Q$ rejects the null at the 10\% significance level. However,  {it should be noted} that  {these} rejections  {may} be spurious, given the seriously oversized distortion of $t_T^Q$  under heteroskedasticity demonstrated in the simulation part, meanwhile our LAD-based heteroskedasticity-robust tools do not lead to rejections\footnote{The rejections of the LAD-based bootstrap tests are crucial to conclude a false identification of the $t_T^Q$ test. This is because the $t_T^Q$ test exhibits higher power in heavy-tailed cases than other methods, which may also  lead to a pattern where the $t_T^Q$ rejects the null while others do not.}. For Ireland, only our two tests $t_B^{T}$ and $L_B^T$, along with the $ALR$ test, consistently reject the null with the p-values less than 7\%. For Spain, our two tests  reject the null at the 5\% significance level, which is also supported by the tests $MZ^T$, $MZ^B$, and $ALR$ with the p-values  of 6.1\%, 9.4\% and 3.2\% respectively. Consequently, we can conclude that  the PPP hypothesis holds true in  Ireland and Spain since 2000, and we do not have strong evidence supporting the validity of PPP conditions for other EU countries. 
		
		\section{Conclusion}
		\label{section conclusion}
		This paper investigates LAD-based unit root tests in the presence of unconditional heteroskedasticity of unknown form. We  start by studying the unit root tests  in a basic model, where the data are generated by an AR(1) process without an intercept, and then extend our testing procedure to allow for deterministic components. We establish the asymptotic properties of the LAD estimator  and the corresponding test statistics under the (nearly) unit root hypothesis, showing that the unconditional volatility affects the limiting distributions in a highly nonlinear manner.  
		To develop feasible LAD-based unit root tests in the presence of time-varying variance,  we further propose an adaptive block bootstrap procedure to approximate the null distribution of the LAD-based test statistics.  The asymptotic validity of the bootstrap procedure is studied, and the bootstrap pseudo sequence effectively preserves the unconditional heteroskedasticity and mixing dependence without requiring any parametric assumptions.
		Simulation results indicate that the LAD-based bootstrap tests  exhibit reasonable empirical sizes in the presence of time-varying volatility. Furthermore, they show higher power than the other benchmarks considered when dealing with heavy-tailed data.
		Lastly, we investigate the PPP hypothesis in 16 EU countries since 2000 using real effective exchange rates. The results show that our newly developed tests can serve as  effective supplements to existing tests when the data are characterized by both unconditional heteroskedasticity and non-Gaussian distributions. 
		
		Several extensions of this paper are conceivable. 
		One possible extension involves introducing  the multiplicative volatility model to simultaneously capture both conditional and unconditional heteroskedasticity in errors. 
		Another extension is to generalize our proposed tests to include covariates to improve testing power. We leave these extensions for future work.

\newpage

	\setcounter{page}{1}
\counterwithin{figure}{section}
\counterwithin{table}{section}
\setcounter{lemma}{0}
\counterwithin{lemma}{section}

{\Large \bf 
	\begin{center}
		Supplementary Material to ``Adaptive LAD-Based Bootstrap Unit Root Tests under Unconditional Heteroskedasticity''
	\end{center}
}

\begin{center}
	\large
	{Jilin Wu}$^{*}$, { Ruike Wu}$^{\dag}$, { and Zhijie Xiao}$^{\ddag}$
\end{center}

\begin{center}
	\textit{Department of Finance, School of
		Economics, Gregory and Paula Chow Institute for studies in {Economics}, and Wang
		Yanan Institute for Studies in Economics (WISE), Xiamen University$^{*}$}\\
	\textit{Department of Finance, School of Economics, Xiamen University$^{\dag}$} \\ 
	\textit{Department of Economics, Boston College$^{\ddag}$} 
\end{center}

In this supplementary material, Section \ref{app:tables} provides additional simulation results.  Sections \ref{section: proof 2}, \ref{section: proof 3}, and \ref{section: proof 4} collect the proofs for Lemmas 1-4 and Theorems 1-6, corresponding to Sections 2, 3, and 4 of the main paper, respectively.
Section \ref{section other lemma} presents three technical lemmas that serve as essential prerequisites for the proofs in Sections \ref{section: proof 2}, \ref{section: proof 3}, and \ref{section: proof 4}.

In what follows, the notation $C$ is a generic constant which may take different values at its different occurrences; $x^{\prime}$ is the transpose of $x$, $||x||_k = \left( E|x|^k \right)^{1/k}$ is the $L^{k}$-norm of random variable $x$, \textcolor{black}{and $\left\vert\left\vert x \right\vert\right\vert$ applied to matrix or vector $x$ denotes its norm $||x|| = \sqrt{trace(A^\prime A)}$ where $trace(\cdot)$ is the trace operator}; $\overset{p}{\rightarrow}$ and  $\overset{d}{\rightarrow}$ denote the convergence in probability and in distribution, respectively; The notations $E^*(\cdot)$, $\overset{d^*}{\rightarrow}$, $\overset{p^*}{\rightarrow}$, $O_{p^*}(\cdot)$ and $o_{p^*}(\cdot)$ are the corresponding operators under the bootstrap probability measure.

\begin{appendices}
	
	\section{ Additional Tables and Figures} \label{app:tables}
	
	In this section, we present the empirical sizes of the $MZ^T$, $MZ^B$, $MZ^S$, $MZ^M$, $ALR$, $t_T^B$, and $L_T^B$ tests for the cases where $\sigma_1 = 5$ under heteroskedasticity. The results for the i.i.d. errors, MA(1) errors, and AR(1) errors are reported in Tables \ref{Table size iid hetero} to \ref{Table size AR hetero}, respectively. Furthermore, we present the empirical sizes and powers of all considered tests, including $t_T^Q$ and $L_T^M$, under homoskedastic volatility. The corresponding results are reported in Table \ref{Table size homo} and \ref{Table power homo}, respectively. 
	Moreover, Figures \ref{Figure size plot iid}-\ref{Figure size plot AR} provide the empirical sizes of the $t_T^Q$, $L_T^M$, $t_T^B$ and $L_T^B$ tests with respect  to  $\sigma_1 \in \{1/9,1/7,1/5,1/3,1,3,5,7,9 \}$  for the cases where the error term follows an i.i.d. sequence, an MA process, and an AR process, respectively, with $T = 250$.

	\begin{table}[!htp]
		\centering
		\caption{Empirical sizes under heteroskedasticity with i.i.d. errors and $\sigma_1 = 5$.}
		\label{Table size iid hetero}
		\begin{tabular}{ccccccccc}
			\toprule 
			$\sigma_t$           & T                    & $MZ^T$ & $MZ^B$ & $MZ^S$ & $MZ^M$ & $ALR$ & $t_T^{B}$ & $L_T^B$ \\ \hline
			\multicolumn{1}{l}{} & \multicolumn{1}{l}{} & \multicolumn{7}{c}{N(0,1)}                                     \\
			One shift            & 100                  & 0.047 & 0.052 & 0.032 & 0.048   & 0.055  & 0.078     & 0.068   \\
			& 250                  & 0.038 & 0.047 & 0.040 & 0.056   & 0.048  & 0.070     & 0.064   \\
			Two shifts           & 100                  & 0.033 & 0.060 & 0.029 & 0.042   & 0.040  & 0.058     & 0.053   \\
			& 250                  & 0.032 & 0.052 & 0.041 & 0.060   & 0.052  & 0.061     & 0.061   \\
			Smooth change        & 100                  & 0.049 & 0.065 & 0.042 & 0.053   & 0.053  & 0.068     & 0.066   \\
			& 250                  & 0.050 & 0.068 & 0.053 & 0.054   & 0.053  & 0.054     & 0.047   \\ \hline
			\multicolumn{1}{l}{} & \multicolumn{1}{l}{} & \multicolumn{7}{c}{t(3)}                                       \\
			One shift            & 100                  & 0.031 & 0.061 & 0.023 & 0.052   & 0.064  & 0.077     & 0.077   \\
			& 250                  & 0.022 & 0.051 & 0.030 & 0.046   & 0.055  & 0.066     & 0.054   \\
			Two shifts           & 100                  & 0.025 & 0.059 & 0.027 & 0.055   & 0.056  & 0.070     & 0.070   \\
			& 250                  & 0.037 & 0.045 & 0.033 & 0.042   & 0.047  & 0.062     & 0.050   \\
			Smooth change        & 100                  & 0.037 & 0.061 & 0.030 & 0.047   & 0.059  & 0.057     & 0.064   \\
			& 250                  & 0.029 & 0.053 & 0.037 & 0.046   & 0.056  & 0.054     & 0.043   \\ \hline
			\multicolumn{1}{l}{} & \multicolumn{1}{l}{} & \multicolumn{7}{c}{DE(0,1)}                                    \\
			One shift            & 100                  & 0.038 & 0.062 & 0.034 & 0.057   & 0.055  & 0.079     & 0.073   \\
			& 250                  & 0.035 & 0.055 & 0.045 & 0.044   & 0.046  & 0.062     & 0.049   \\
			Two shifts           & 100                  & 0.023 & 0.058 & 0.030 & 0.059   & 0.059  & 0.070     & 0.069   \\
			& 250                  & 0.030 & 0.046 & 0.036 & 0.045   & 0.038  & 0.053     & 0.044   \\
			Smooth change        & 100                  & 0.041 & 0.062 & 0.040 & 0.046   & 0.053  & 0.056     & 0.067   \\
			& 250                  & 0.041 & 0.052 & 0.040 & 0.041   & 0.049  & 0.049     & 0.051  
			\\ \bottomrule
		\end{tabular}
	\end{table}

	\begin{table}[!htp]
		\centering
		\caption{Empirical sizes under heteroskedasticity with MA(1) errors and $\sigma_1 = 5$.}
		\label{Table size MA hetero}
		\begin{tabular}{ccccccccc}
			\toprule 
			$\sigma_t$           & T                    & $MZ^T$ & $MZ^B$ & $MZ^S$ & $MZ^M$ & $ALR$ & $t_T^{B}$ & $L_T^B$ \\ \hline
			\multicolumn{1}{l}{} & \multicolumn{1}{l}{} & \multicolumn{7}{c}{N(0,1)}                                     \\
			One shift            & 100                  & 0.116 & 0.038 & 0.098 & 0.056   & 0.035  & 0.057     & 0.037   \\
			& 250                  & 0.091 & 0.046 & 0.082 & 0.055   & 0.035  & 0.059     & 0.055   \\
			Two shifts           & 100                  & 0.042 & 0.025 & 0.027 & 0.035   & 0.033  & 0.057     & 0.042   \\
			& 250                  & 0.046 & 0.033 & 0.045 & 0.057   & 0.049  & 0.060     & 0.049   \\
			Smooth change        & 100                  & 0.076 & 0.036 & 0.063 & 0.065   & 0.050  & 0.055     & 0.045   \\
			& 250                  & 0.054 & 0.038 & 0.052 & 0.053   & 0.046  & 0.056     & 0.056   \\ \hline
			\multicolumn{1}{l}{} & \multicolumn{1}{l}{} & \multicolumn{7}{c}{t(3)}                                       \\
			One shift            & 100                  & 0.126 & 0.045 & 0.102 & 0.057   & 0.037  & 0.063     & 0.049   \\
			& 250                  & 0.072 & 0.047 & 0.066 & 0.052   & 0.043  & 0.057     & 0.038   \\
			Two shifts           & 100                  & 0.039 & 0.033 & 0.018 & 0.032   & 0.036  & 0.053     & 0.041   \\
			& 250                  & 0.044 & 0.038 & 0.038 & 0.039   & 0.035  & 0.046     & 0.037   \\
			Smooth change        & 100                  & 0.083 & 0.040 & 0.055 & 0.052   & 0.048  & 0.052     & 0.046   \\
			& 250                  & 0.050 & 0.045 & 0.043 & 0.059   & 0.046  & 0.042     & 0.035   \\ \hline
			\multicolumn{1}{l}{} & \multicolumn{1}{l}{} & \multicolumn{7}{c}{DE(0,1)}                                    \\
			One shift            & 100                  & 0.110 & 0.052 & 0.104 & 0.058   & 0.030  & 0.070     & 0.049   \\
			& 250                  & 0.075 & 0.037 & 0.064 & 0.045   & 0.029  & 0.049     & 0.041   \\
			Two shifts           & 100                  & 0.028 & 0.022 & 0.017 & 0.048   & 0.033  & 0.060     & 0.043   \\
			& 250                  & 0.040 & 0.038 & 0.042 & 0.043   & 0.041  & 0.044     & 0.038   \\
			Smooth change        & 100                  & 0.060 & 0.032 & 0.051 & 0.059   & 0.030  & 0.040     & 0.036   \\
			& 250                  & 0.041 & 0.028 & 0.040 & 0.046   & 0.046  & 0.041     & 0.036  
			\\ \bottomrule
		\end{tabular}
	\end{table}

	\begin{table}[!htp]
		\centering
		\caption{Empirical sizes under heteroskedasticity with AR(1) errors and $\sigma_1 = 5$.}
		\label{Table size AR hetero}
		\begin{tabular}{ccccccccc}
			\toprule 
			$\sigma_t$           & T                    & $MZ^T$ & $MZ^B$ & $MZ^S$ & $MZ^M$ & $ALR$ & $t_T^{B}$ & $L_T^B$ \\ \hline
			\multicolumn{1}{l}{} & \multicolumn{1}{l}{} & \multicolumn{7}{c}{N(0,1)}                                     \\
			One shift            & 100                  & 0.119 & 0.042 & 0.104 & 0.060   & 0.030  & 0.052     & 0.036   \\
			& 250                  & 0.097 & 0.054 & 0.089 & 0.057   & 0.032  & 0.049     & 0.041   \\
			Two shifts           & 100                  & 0.049 & 0.031 & 0.026 & 0.042   & 0.033  & 0.047     & 0.036   \\
			& 250                  & 0.038 & 0.041 & 0.034 & 0.054   & 0.050  & 0.058     & 0.049   \\
			Smooth change        & 100                  & 0.095 & 0.040 & 0.073 & 0.069   & 0.052  & 0.047     & 0.033   \\
			& 250                  & 0.056 & 0.047 & 0.060 & 0.053   & 0.049  & 0.058     & 0.047   \\ \hline
			\multicolumn{1}{l}{} & \multicolumn{1}{l}{} & \multicolumn{7}{c}{t(3)}                                       \\
			One shift            & 100                  & 0.124 & 0.041 & 0.094 & 0.061   & 0.035  & 0.059     & 0.037   \\
			& 250                  & 0.080 & 0.041 & 0.063 & 0.053   & 0.040  & 0.048     & 0.038   \\
			Two shifts           & 100                  & 0.044 & 0.042 & 0.019 & 0.035   & 0.040  & 0.059     & 0.039   \\
			& 250                  & 0.043 & 0.028 & 0.029 & 0.038   & 0.033  & 0.041     & 0.026   \\
			Smooth change        & 100                  & 0.088 & 0.044 & 0.071 & 0.059   & 0.051  & 0.051     & 0.031   \\
			& 250                  & 0.053 & 0.047 & 0.049 & 0.056   & 0.043  & 0.038     & 0.029   \\ \hline
			\multicolumn{1}{l}{} & \multicolumn{1}{l}{} & \multicolumn{7}{c}{DE(0,1)}                                    \\
			One shift            & 100                  & 0.134 & 0.065 & 0.111 & 0.065   & 0.035  & 0.058     & 0.040   \\
			& 250                  & 0.079 & 0.041 & 0.062 & 0.043   & 0.028  & 0.054     & 0.036   \\
			Two shifts           & 100                  & 0.039 & 0.033 & 0.019 & 0.043   & 0.033  & 0.060     & 0.035   \\
			& 250                  & 0.042 & 0.036 & 0.034 & 0.037   & 0.035  & 0.050     & 0.032   \\
			Smooth change        & 100                  & 0.070 & 0.038 & 0.060 & 0.062   & 0.026  & 0.040     & 0.033   \\
			& 250                  & 0.051 & 0.036 & 0.041 & 0.045   & 0.049  & 0.042     & 0.038  
			\\ \bottomrule
		\end{tabular}
	\end{table}

	\begin{table}[!htp]
		\centering
		\caption{Empirical sizes under homoskedasticity.}
		\label{Table size homo}
		\begin{tabular}{ccccccccccc}
			\toprule
			\multicolumn{1}{c}{Distribution} & \multicolumn{1}{c}{T} & $MZ^T$ & $MZ^B$ & $MZ^S$ & $MZ^M$ & $t_T^Q$ & $L_T^M$ & \multicolumn{1}{l}{$ALR$} & $t_T^{B}$ & $L_T^B$ \\ \hline
			&                       & \multicolumn{9}{c}{i.i.d}                                                                             \\
			N(0,1)                           & 100                   & 0.050 & 0.056 & 0.045 & 0.050   & 0.055   & 0.061   & 0.051                     & 0.059     & 0.058   \\
			& 250                   & 0.050 & 0.053 & 0.056 & 0.052   & 0.048   & 0.066   & 0.056                     & 0.063     & 0.060   \\
			t(3)                             & 100                   & 0.031 & 0.059 & 0.024 & 0.044   & 0.048   & 0.053   & 0.048                     & 0.057     & 0.052   \\
			& 250                   & 0.027 & 0.049 & 0.034 & 0.041   & 0.050   & 0.052   & 0.045                     & 0.055     & 0.048   \\
			DE(0,1)                          & 100                   & 0.041 & 0.054 & 0.035 & 0.042   & 0.047   & 0.065   & 0.061                     & 0.061     & 0.067   \\
			& 250                   & 0.040 & 0.042 & 0.035 & 0.044   & 0.048   & 0.047   & 0.041                     & 0.041     & 0.048   \\ \hline
			&                       & \multicolumn{9}{c}{MA(1)}                                                                             \\
			N(0,1)                           & 100                   & 0.050 & 0.025 & 0.044 & 0.057   & 0.038   & 0.015   & 0.033                     & 0.048     & 0.039   \\
			& 250                   & 0.050 & 0.034 & 0.049 & 0.052   & 0.036   & 0.011   & 0.043                     & 0.050     & 0.047   \\
			t(3)                             & 100                   & 0.052 & 0.024 & 0.032 & 0.045   & 0.042   & 0.014   & 0.032                     & 0.040     & 0.029   \\
			& 250                   & 0.056 & 0.039 & 0.039 & 0.046   & 0.047   & 0.013   & 0.036                     & 0.045     & 0.032   \\
			DE(0,1)                          & 100                   & 0.062 & 0.032 & 0.043 & 0.058   & 0.031   & 0.015   & 0.036                     & 0.045     & 0.035   \\
			& 250                   & 0.043 & 0.027 & 0.039 & 0.039   & 0.049   & 0.009   & 0.040                     & 0.042     & 0.034   \\ \hline
			&                       & \multicolumn{9}{c}{AR(1)}                                                                             \\
			N(0,1)                           & 100                   & 0.054 & 0.025 & 0.051 & 0.054   & 0.037   & 0.004   & 0.033                     & 0.038     & 0.037   \\
			& 250                   & 0.054 & 0.038 & 0.049 & 0.046   & 0.036   & 0.001   & 0.037                     & 0.044     & 0.039   \\
			t(3)                             & 100                   & 0.067 & 0.035 & 0.037 & 0.047   & 0.041   & 0.002   & 0.031                     & 0.040     & 0.029   \\
			& 250                   & 0.056 & 0.044 & 0.039 & 0.050   & 0.052   & 0.002   & 0.047                     & 0.044     & 0.038   \\
			DE(0,1)                          & 100                   & 0.065 & 0.036 & 0.052 & 0.057   & 0.041   & 0.003   & 0.034                     & 0.044     & 0.032   \\
			& 250                   & 0.045 & 0.034 & 0.039 & 0.042   & 0.043   & 0.001   & 0.033                     & 0.042     & 0.031  \\ \bottomrule
		\end{tabular}
	\end{table}

	\begin{table}[!htp]
		\centering
		\caption{Empirical powers under homoskedasticity.}
		\label{Table power homo}
		\begin{tabular}{ccccccccccc}
			\toprule
			\multicolumn{1}{c}{Distribution} & \multicolumn{1}{c}{T} & $MZ^T$ & $MZ^B$ & $MZ^S$ & $MZ^M$ & $t_T^Q$ & $L_T^M$ & \multicolumn{1}{l}{$ALR$} & $t_T^{B}$ & $L_T^B$ \\ \hline
			&                       & \multicolumn{9}{c}{i.i.d}                                                                             \\
			N(0,1)                           & 100                   & 0.653 & 0.715 & 0.621 & 0.649   & 0.194   & 0.657   & 0.703                     & 0.526     & 0.637   \\
			& 250                   & 0.688 & 0.731 & 0.700 & 0.696   & 0.194   & 0.652   & 0.717                     & 0.524     & 0.624   \\
			t(3)                             & 100                   & 0.712 & 0.756 & 0.521 & 0.669   & 0.490   & 0.876   & 0.816                     & 0.825     & 0.877   \\
			& 250                   & 0.667 & 0.757 & 0.651 & 0.759   & 0.596   & 0.922   & 0.812                     & 0.872     & 0.926   \\
			DE(0,1)                          & 100                   & 0.699 & 0.718 & 0.561 & 0.676   & 0.428   & 0.935   & 0.774                     & 0.876     & 0.936   \\
			& 250                   & 0.725 & 0.747 & 0.686 & 0.778   & 0.596   & 0.963   & 0.768                     & 0.931     & 0.955   \\ \hline
			&                       & \multicolumn{9}{c}{MA(1)}                                                                             \\
			N(0,1)                           & 100                   & 0.398 & 0.272 & 0.364 & 0.341   & 0.061   & 0.272   & 0.354                     & 0.395     & 0.469   \\
			& 250                   & 0.549 & 0.449 & 0.527 & 0.543   & 0.109   & 0.268   & 0.497                     & 0.401     & 0.518   \\
			t(3)                             & 100                   & 0.339 & 0.334 & 0.322 & 0.431   & 0.251   & 0.543   & 0.504                     & 0.682     & 0.716   \\
			& 250                   & 0.427 & 0.491 & 0.492 & 0.584   & 0.389   & 0.593   & 0.627                     & 0.742     & 0.785   \\
			DE(0,1)                          & 100                   & 0.333 & 0.285 & 0.339 & 0.374   & 0.166   & 0.481   & 0.414                     & 0.650     & 0.704   \\
			& 250                   & 0.583 & 0.474 & 0.531 & 0.657   & 0.363   & 0.503   & 0.562                     & 0.710     & 0.767   \\ \hline
			&                       & \multicolumn{9}{c}{AR(1)}                                                                             \\
			N(0,1)                           & 100                   & 0.396 & 0.292 & 0.355 & 0.366   & 0.070   & 0.089   & 0.342                     & 0.329     & 0.371   \\
			& 250                   & 0.529 & 0.484 & 0.537 & 0.584   & 0.114   & 0.072   & 0.518                     & 0.356     & 0.411   \\
			t(3)                             & 100                   & 0.353 & 0.367 & 0.336 & 0.399   & 0.274   & 0.213   & 0.504                     & 0.548     & 0.588   \\
			& 250                   & 0.483 & 0.512 & 0.499 & 0.555   & 0.389   & 0.224   & 0.622                     & 0.623     & 0.685   \\
			DE(0,1)                          & 100                   & 0.341 & 0.321 & 0.335 & 0.344   & 0.186   & 0.163   & 0.399                     & 0.500     & 0.530   \\
			& 250                   & 0.580 & 0.486 & 0.535 & 0.584   & 0.358   & 0.163   & 0.557                     & 0.537     & 0.601  \\ \bottomrule
		\end{tabular}
	\end{table}

	\begin{figure}[!htb]
		\centering	
		\includegraphics[width=16cm, trim={0cm 0 0 0}]{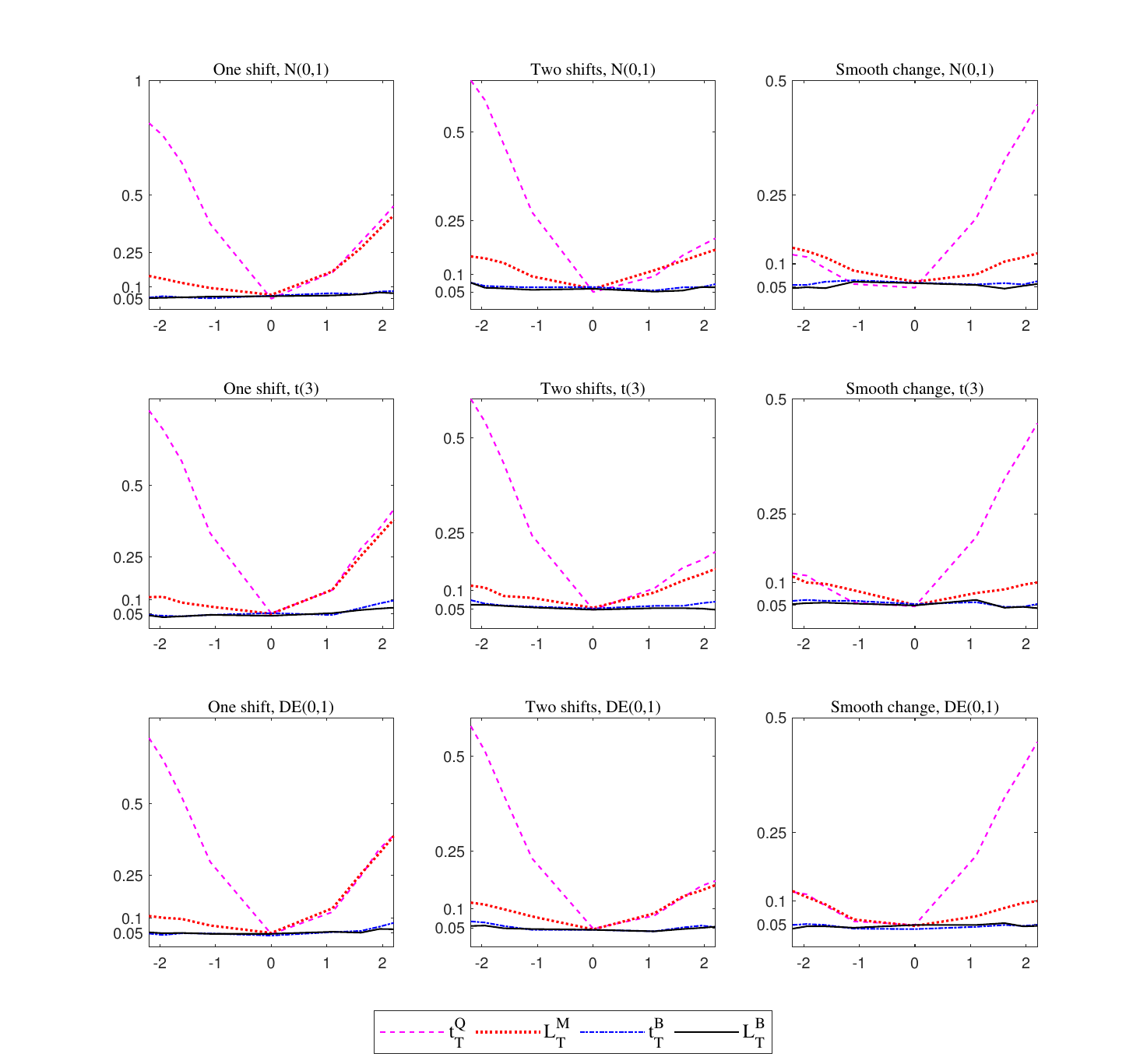}	
		\caption{Empirical sizes (y-axis) of the $t_T^Q$, $L_T^M$, $t_T^B$, and $L_T^B$ tests for $\sigma_1$ between $1/9$ and 9 (x-axis, log scale) when the error term follows an i.i.d process. Nominal size is 0.05 and sample size $T = 250$. }\label{Figure size plot iid}
	\end{figure}

	\begin{figure}[!htb]
		\centering	
		\includegraphics[width=16cm, trim={0cm 0 0 0}]{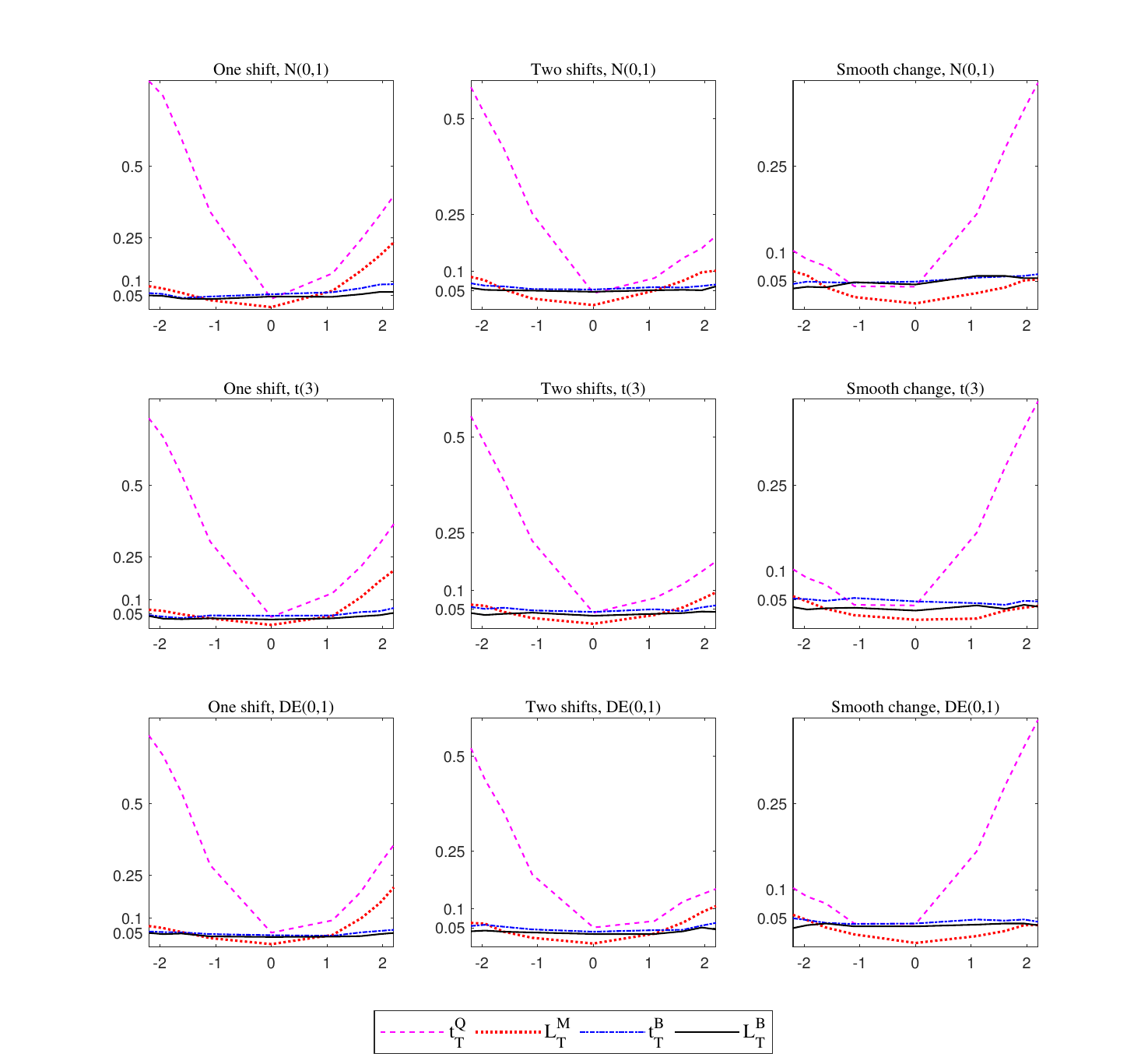}	
		\caption{Empirical sizes (y-axis) of the $t_T^Q$, $L_T^M$, $t_T^B$, and $L_T^B$ tests for $\sigma_1$ between $1/9$ and 9 (x-axis, log scale) when the error term follows an MA(1) process. Nominal size is 0.05 and sample size $T = 250$. }\label{Figure size plot MA}
	\end{figure}

	\begin{figure}[!htb]
		\centering	
		\includegraphics[width=16cm, trim={0cm 0 0 0}]{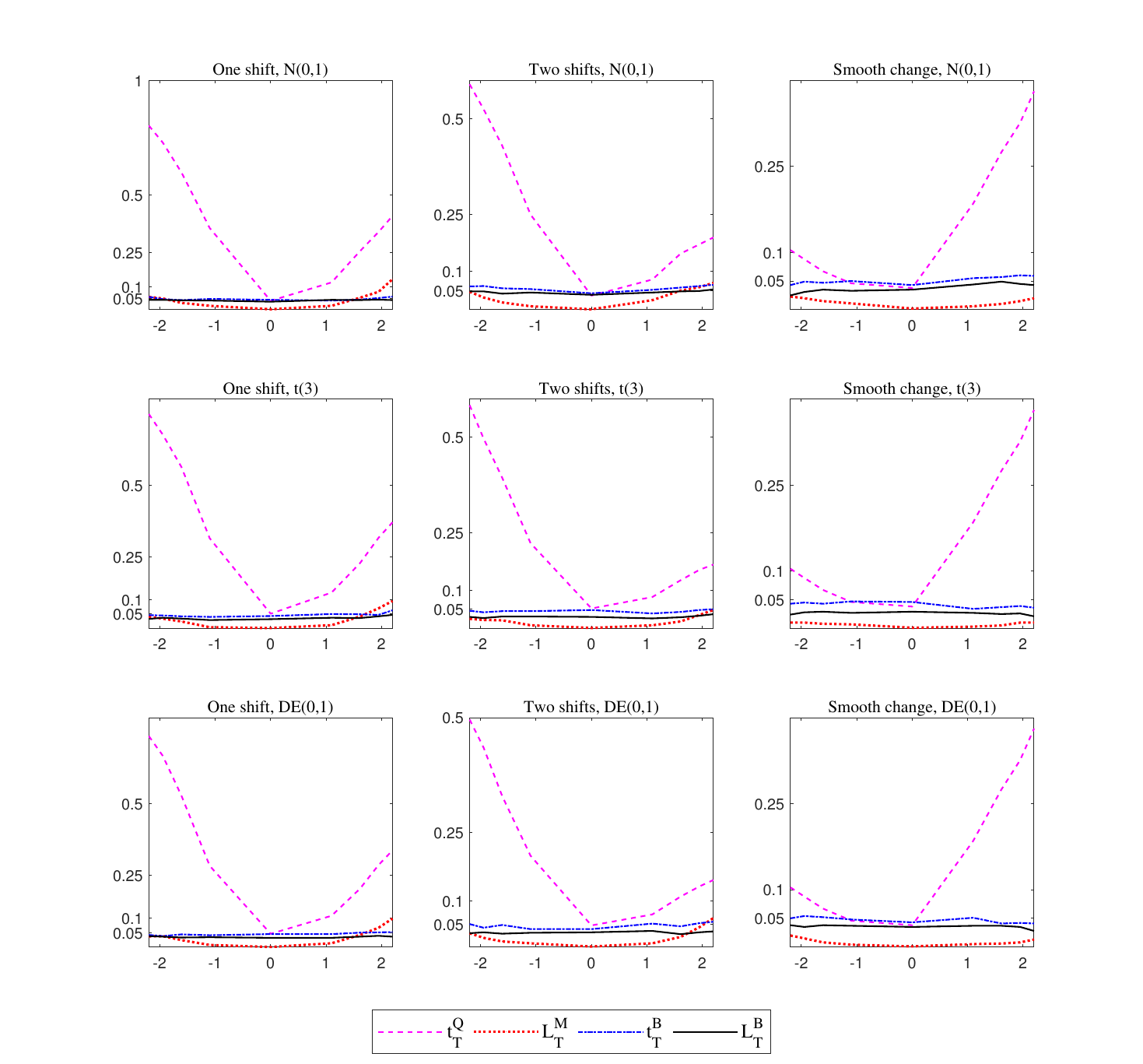}	
		\caption{Empirical sizes (y-axis) of the $t_T^Q$, $L_T^M$, $t_T^B$, and $L_T^B$ tests for $\sigma_1$ between $1/9$ and 9 (x-axis, log scale) when the error term follows an AR(1) process. Nominal size is 0.05 and sample size $T = 250$. }\label{Figure size plot AR}
	\end{figure}


	\section{Proofs of Theorems and Lemmas in Section 2}
	\label{section: proof 2}

	\subsection{Proof of Lemma 1}
	\begin{proof}[Proof of Lemma 1:] By Assumptions 1(i), 1(iii), and 2, $\{(\varepsilon_{t}, sgn(\varepsilon_{t}))\}_{t=1}^{\infty}$ is a strictly stationary mixing sequence with zero mean and the long-run variance-covariance $\Sigma$. By applying Corollary 2.2 of \cite{Phillips1986}, we immediately have equation (2.6) in the main paper.
		
		For equation (2.7), its proof is similar to that of Lemma 2 in \cite{Cavaliere2010}.  Let $U_t = diag(\sigma_t,1)$ and $\nu_{t}=(\varepsilon_{t},sgn(\varepsilon_{t}))^{\prime}$, denote $E(\cdot|\mathcal{G}_{t})$ by $E_{t}(\cdot)$, where $\mathcal{G}_{t}$ \textcolor{black}{is $\sigma$-filed generated by $\{\varepsilon_{t},\varepsilon_{t-1},\ldots \}$}, and define $\zeta_{t}=\sum_{l=0}^{\infty}\left(  E_{t}\nu_{t+l}-E_{t-1}\nu_{t+l}\right)$ and $z_{t}=\sum_{l=1}^{\infty}E_{t}\nu_{t+l}$. Following \cite{Hansen1992}, we have $\nu_{t}=\zeta_{t}+z_{t-1}-z_{t}$ and $E_{t-1}\zeta_{t}=0$, which leads to 
		\begin{equation}
			\frac{1}{\sqrt{T}}\sum_{t=1}^{\lfloor Tr \rfloor}U_{t}\nu_{t}=\frac{1}{\sqrt{T}}\sum
			_{t=1}^{\lfloor Tr \rfloor}U_{t}\zeta_{t}+\frac{1}{\sqrt{T}}\sum_{t=1}^{\lfloor Tr \rfloor}\left(
			U_{t}-U_{t-1}\right)  z_{t}-\frac{1}{\sqrt{T}}U_{\lfloor Tr \rfloor}z_{\lfloor Tr \rfloor}
			=L_{1}+L_{2}-L_{3}. 
		\end{equation}
		By Theorem 3.1 of \cite{Hansen1992}, we have (i) \textcolor{black}{$\frac{1}{\sqrt{T}}\mathop{sup}_{t\leq T}||z_t||\overset{p}{\rightarrow}0$}, (ii) $\frac{1}{{T}}\sum_{t=1}^{T}E\zeta_{t}^{2}<\infty$ and (iii) $\frac{1}{\sqrt{T}}\sum_{t=1}^{\lfloor Tr \rfloor}\zeta_{t} = \frac{1}{\sqrt{T}}\sum_{t=1}^{\lfloor Tr \rfloor}\nu_{t} + o_p(1)$. Based on the above results (i)-(iii), Theorem 2.1 of \cite{Hansen1992} and  $\frac{1}{\sqrt{T}}\sum_{t=1}^{\lfloor Tr \rfloor}\nu_{t}\overset{d}{\rightarrow}\left(  B_{1}(r),B_{2}(r)\right)^{\prime}$, we obtain
		\begin{equation}
			L_{1}\overset{d}{\rightarrow}\int_{0}^{r}\left(\begin{array}[c]{cc}\sigma(s) & 0\\0 & 1\end{array}\right)  d\left(\begin{array}[c]{c}B_{1}(s)\\B_{2}(s)\end{array}\right)  =\left(\begin{array}[c]{c}B_{\sigma}(r)\\B_{2}(r)\end{array}\right)  .
		\end{equation}
		Furthermore, it is not hard to prove that both $L_{2}$ and $L_{3}$ are $o_{p}(1)$ by using  result (i)  and the fact that $\sigma(\cdot)$ is positively bounded and nonstochastic. This completes the proof of (2.7). \end{proof}
	
	
	\subsection{Proof of Lemma 2}
	\begin{proof}[Proof of Lemma 2:]
		Let $\varsigma_{t}=(u_{t},sgn(u_{t}))^{\prime}$,  $V_{t}=\sum_{i=1}^{t}\varsigma_{i}$ and $V_{T,t}=\frac{1}{\sqrt{T}}V_{t}$. By Lemma 1, we have $V_{T,\lfloor Tr \rfloor}\overset {d}{\rightarrow}(B_{\sigma}(r),B_{2}(r))^{\prime}$. Similarly, define $\epsilon_{t}=\sum_{l=0}^{\infty}\left(  E_{t}\varsigma_{t+l}-E_{t-1}\varsigma_{t+l}\right)$ and  $m_{t}=\sum_{l=1}^{\infty}E_{t}\varsigma_{t+l}$, thus we have $\varsigma_{t}=\epsilon_{t}+m_{t-1}-m_{t}$ and $E_{t-1}\epsilon_{t}=0$, which leads to
		\begin{equation}
			\frac{1}{T}\sum_{t=1}^{\lfloor Tr \rfloor}V_{t}\varsigma_{t+1}^{\prime}=\frac{1}{T}\sum_{t=1}%
			^{\lfloor Tr \rfloor}V_{t}\epsilon_{t+1}^{\prime}+\frac{1}{T}\sum_{t=1}^{\lfloor Tr \rfloor}\varsigma_{t}m_{t}^{\prime}-\frac{1}{T}V_{\lfloor Tr \rfloor}m_{\lfloor Tr \rfloor+1}^{\prime}=\mathcal{L}%
			_{1}+\mathcal{L}_{2}-\mathcal{L}_{3}.
			\label{decompose l1l2l3} 
		\end{equation}
		Because $\varsigma_{t}$ can be rewritten as $\varsigma_{t}=U_{t}\nu_{t}$, where $U_{t}$ and $\nu_{t}$ have been defined in the proof of Lemma 1.  By Theorem 3.1 of \cite{Hansen1992}, we  also have (i') \textcolor{black}{$\frac{1}{\sqrt{T}}\mathop{sup}_{t\leq T}||m_{t}||\overset{p}{\rightarrow}0$}, (ii') $\frac{1}{T}\sum_{t=1}^{T}E\epsilon_{t}^{2}<\infty$ and (iii') $\frac{1}{\sqrt{T}}\sum_{t=1}^{\lfloor Tr \rfloor}\epsilon_{t}= V_{T,\lfloor Tr \rfloor} + o_p(1)$. Based on the results (i')-(iii'),  Theorem 2.1 of \cite{Hansen1992} and $V_{T,\lfloor Tr \rfloor}\overset{d}{\rightarrow}(B_{\sigma}(r),B_{2}(r))^{\prime}$, we obtain
		\begin{equation}
			\mathcal{L}_{1}\overset{d}{\rightarrow}\int_{0}^{r}(B_{\sigma}(r),B_{2}(r))^{\prime}d(B_{\sigma}(r),B_{2}(r)). 
			\label{L1 limit}
		\end{equation}
		For $\mathcal{L}_{2}$, according to Theorem 3.2 and Theorem 3.3 of \cite{Hansen1992}, we have 
		\begin{align}
			\frac{1}{T}\sum_{t=1}^{\lfloor Tr \rfloor}\varsigma_{t}m_{t}^{\prime}
			\overset{p}{\rightarrow}\frac{1}{T}\sum_{i=1}^{\lfloor Tr \rfloor}\sum_{j=i+1}^{\infty}E\left(
			\varsigma_{i}\varsigma_{j}^{\prime}\right) .
			\label{L2 limit}
		\end{align}
		For $\mathcal{L}_{3}$, we note that
		\textcolor{black}{		
			\begin{align}
				\mathop{sup}_{t\leq T}\frac{1}{T}\Vert V_{t}m_{t+1}^{\prime}\Vert
				\leq\mathop \frac{1}{\sqrt{T}}{sup}_{t\leq T}\left\Vert V_{t}\right\Vert
				\frac{1}{\sqrt{T}}\mathop{sup}_{t\leq T}\left\Vert m_{t}\right\Vert
				\overset{p}{\rightarrow}0 
				\label{L3 limit}
		\end{align} }
		since $\frac{1}{\sqrt{T}}{sup}_{t\leq T}\left\Vert V_{t}\right\Vert=O_{p}\left(  1\right)  $ and $\frac{1}{\sqrt{T}}\mathop{sup}_{t\leq T}\left\Vert m_{t}\right\Vert =o_{p}\left(  1\right)  $. Thus, by (\ref{decompose l1l2l3})-(\ref{L3 limit}), it follows that 
		\begin{equation}
			\frac{1}{T}\sum_{t=1}^{T}V_{t}\varsigma_{t+1}^{\prime}\overset{d}{\rightarrow}\int_{0}^{1}(B_{\sigma}(r),B_{2}(r))^{\prime}d(B_{\sigma}(r),B_{2}(r))+lim_{T\rightarrow \infty}\frac{1}{T}\sum_{i=1}^{T}\sum_{j=i+1}^{\infty}E\left(  \varsigma_{i}\varsigma_{j}^{\prime}\right),
			\label{lemma 2 equation}
		\end{equation}
		and  the first row and second column element of the above formula leads to 
		\begin{equation}
			T^{-1}\sum_{t=1}^{T}y_{t-1}sgn(u_t) \overset{d}{\rightarrow}\int_{0}^{1}B_\sigma(r)dB_2(r) + lim_{T\rightarrow \infty}\frac{1}{T}\sum_{i=1}^{T}\sum_{j=i+1}^{\infty}E\left(  u_{i}sgn(u_{j})\right).
			\label{FT}
		\end{equation}
		
		Now we define $M_{t}=\int_{0}^{vT^{-1}y_{t-1}}\left[  I\left(  u_{t}\leq s\right)  -I\left(u_{t}\leq0\right) \right]ds$ and $\mu_{t}=E\left(  M_{t}|\mathcal{F}_{t-1}\right)$, where $\mathcal{F}_{t-1}$ is $\sigma$-field generated by $\{u_{t-1}, u_{t-2}, \ldots\}$, then
		\begin{equation}
			\sum_{t=1}^{T}M_{t} = \sum_{t=1}^{T}\mu_{t}+\sum_{t=1}^{T}\left(  M_{t}-\mu_{t}\right).
		\end{equation}
		Subsequently, under Assumptions 1(ii) and 3, we can show that for each $v\in R$,
		\begin{align}
			\sum_{t=1}^{T}\mu_{t}  &  =\sum_{t=1}^{T}\int_{0}^{vT^{-1}y_{t-1}}\left[E\left(  I\left(  u_{t}\leq s\right)  |\mathcal{F}_{t-1}\right)  -E\left( I\left(  u_{t}\leq0\right)  |\mathcal{F}_{t-1}\right)  \right]  ds\nonumber\\
			&  =\sum_{t=1}^{T}\sigma_{t}^{-1}f_{t-1}(0)\frac{v^{2}y_{t-1}^{2}}{2T^{2}}+o_{p}(1).
		\end{align}
		Next, integrability of $f_{t-1}^{r}(0)$ and stationarity of $f_{t-1}(0)$ in Assumption 1(ii) ensure that, as $T\rightarrow\infty$,
		\[
		\mathop{sup}_{0\leq s\leq1}\left\vert \frac{1}{T^{1-\epsilon}}\sum
		_{t=1}^{\lfloor Ts\rfloor}(f_{t-1}(0)-f(0))\right\vert \overset{p}{\rightarrow}0,
		\]
		for some $\epsilon>0$, which in turn leads to
		\begin{equation}
			2\sum_{t=1}^{T}\mu_{t}\overset{d}{\rightarrow}v^{2}\left(
			f(0)\int_{0}^{1}\sigma^{-1}(r)B_{\sigma}^{2}(r)dr\right), 
		\end{equation}
		where we have used $T^{-1/2}y_{\lfloor Tr\rfloor}\overset{d}{\rightarrow}B_{\sigma}(r)$. 
		
		The remaining work is to prove that $\sum_{t=1}^{T}\left(  M_{t}-\mu_{t}\right)  =o_{p}(1)$. Note that $\{M_{t}-\mu_{t}\}$ is a martingale difference sequence, which leads to
		\begin{align*}
			var\left[  \sum_{t=1}^{T}\left(  M_{t}-\mu_{t}\right)  \right]  \leq
			\sum_{t=1}^{T}E(M_{t}^{2})\leq f(0)\sum_{t=1}^{T}\frac{1}{\sigma_t}\int_0^{vT^{-1}y_{t-1}}\int_0^{vT^{-1}y_{t-1}} min(x,y) dxdy = O_{p}(T^{-1/2}), 
		\end{align*}
		where $\int_0^{vT^{-1}y_{t-1}}\int_0^{vT^{-1}y_{t-1}}min(x,y)dxdy = \frac{1}{3}(vT^{-1}y_{t-1})^3$.
		As a result, we have $\sum_{t=1}^{T}\left(  M_{t}-\mu_{t}\right)  =o_{p}(1)$, and 
		$$2\sum_{t=1}^{T}\int_{0}^{vT^{-1}y_{t-1}}\left[  I\left(  u_{t}\leq s\right)
		-I\left(  u_{t}\leq0\right)  \right]  ds\overset{d}{\rightarrow}v^{2}\left(	f\left(  0\right)  \int_{0}^{1}\sigma^{-1}\left(  r\right)  B_{\sigma}^{2}\left(  r\right)  dr\right).  $$
	\end{proof}
	
	\subsection{Proof of Theorem 2}
	\begin{proof}[Proof of Theorem 2:]

		Under $\mathbb{H}_{A}:\gamma_{0}=\exp\left(  -\frac{c}{T}\right),c\geq0,$ by taking the similar arguments to those of proving Lemma 2, we can show
		\begin{align}
			T^{-1}\sum_{t=1}^{T}&y_{t-1}sgn(u_{t})\overset{d}{\rightarrow}\int_{0}^{1}\mathcal{J}_{c,\sigma}(r)dB_{2}(r)+\Gamma,
			\label{yt-1 sgn alter}
			\\
			2\sum_{t=1}^{T}\int_{0}^{vT^{-1}y_{t-1}} &\left[  I\left(  u_{t}\leq s\right)-I\left(  u_{t}\leq0\right)  \right]  ds\overset{d}{\rightarrow}v^{2}\left(f(0)\int_{0}^{1}\sigma^{-1}(r){\ \mathcal{J}^2_{c,\sigma}(r)}dr \right),
			\label{yt-1 sgn alter longer}
		\end{align}
		where $\Gamma=lim_{T\rightarrow\infty}\frac{1}{T}\sum_{i=1}^{T}\sum_{j=i+1}^{\infty}E(u_{i}sgn(u_{j}))$.  Therefore,
		\[
		H_{T}\left(  v\right)  \overset{d}{\rightarrow}\tilde{H}\left(  v\right)
		:=-v\left(  \int_{0}^{1}\mathcal{J}_{c,\sigma}(r)dB_{2}(r)+\Gamma\right)
		+v^{2}\left(  f\left(  0\right)  \int_{0}^{1}\sigma^{-1}(r){\ \mathcal{J}^2_{c,\sigma}(r)}dr\right)  .
		\]
		Since $\tilde{H}\left(  v\right)  $ is minimized at $\hat{v}=T\left(  \hat{\gamma
		}_{LAD}-\gamma_{0}\right)  $, we then have Theorem 2.
	\end{proof}

	
	\section{Proofs of Theorems and Lemmas in Section 3}\label{section: proof 3}
	
	\subsection{Proof of Lemma 3}\label{subsection proof of lemma 3}
	
	\begin{proof}[Proof of Lemma 3:]
		For ease of exposition, $b_{T}$ is abbreviated as $b$. Because $k=\lfloor(T-1)/b\rfloor$, we write
		\begin{equation}
			\frac{1}{\sqrt{T}}\sum_{t=1}^{T}\varepsilon_{t}^{*}=\frac{1}{\sqrt{T}}%
			\sum_{t=1}^{bk}{\varepsilon}_{t}^{*}+\frac{1}{\sqrt{T}}\sum_{t=bk+1}%
			^{T}{\varepsilon}_{t}^{*}=P+G, 
			\label{boot decompose equ}
		\end{equation}
		where $T-bk\leq b$. Let $K^+=\{0,1,\cdots,k-1\}$ and define the integer sets $\mathcal{P}=\{m|m\in K^+,i_{m}>0\}$ and $\mathcal{N}=\{m|m\in K^+,i_{m}<0\}$. Clearly, $\left\langle \mathcal{P}\right\rangle +\left\langle \mathcal{N}\right\rangle = k$, where $\left\langle A\right\rangle $ is the cardinality of the set $A$. Then $P$ has the decomposition
		\begin{equation}
			P=\frac{1}{\sqrt{T}}\sum_{m\in\mathcal{P}}\sum_{s=0}^{b-1}\hat{\varepsilon
			}_{i_{m}+s}+\frac{1}{\sqrt{T}}\sum_{m\in\mathcal{N}}\sum_{s=0}^{b-1}\left(
			-\hat{\varepsilon}_{i_{m}+s+T+1}\right)  =P_{1}+P_{2}. 
		\end{equation}
		For $P_{1}$ and $P_2$,
		\begin{align*}
			&P_1=\frac{1}{\sqrt{T}}\sum_{m\in\mathcal{P}}\sum_{s=0}^{b-1}%
			{\varepsilon}_{i_{m}+s}+(\gamma_{0}-\hat{\gamma}_{LAD})\frac{1}{\sqrt{T}}%
			\sum_{m\in\mathcal{P}}\sum_{s=0}^{b-1}\frac{{y}_{i_{m}+s-1}}{\sigma_{i_{m}+s}%
			}=P_{11}+(\gamma_{0}-\hat{\gamma}_{LAD})P_{12}; 
			\\
			& 
			P_2 =-\frac{1}{\sqrt{T}}\sum_{m\in\mathcal{N}}\sum_{s=0}^{b-1}{\varepsilon}_{i_{m}+s+T+1} - (\gamma_{0}-\hat{\gamma}_{LAD})\frac{1}{\sqrt{T}}\sum_{m\in\mathcal{N}}\sum_{s=0}^{b-1}\frac{{y}_{i_{m}+s+T}}{\sigma_{i_{m}+s+T+1}%
			}=P_{21}+(\gamma_{0}-\hat{\gamma}_{LAD})P_{22}. 
		\end{align*}
		\textcolor{black}{Let $Z_m = \sum_{s=0}^{b-1}y_{i_m+s-1}/\sigma_{i_m+s}$ if $m \in \mathcal{P}$ and ${Z}_m = \sum_{s=0}^{b-1}y_{i_m+s+T}/\sigma_{i_m+s+T+1}$ if $m\in\mathcal{N}$. Note that $Z_m$ is independent and identically distributed.} 
		Since the index $i_m,m=0,\ldots,k-1$, is randomly drawn from \textcolor{black}{the} set $\mathcal{W} = \{-T,-T+1,\ldots,-1-b, 1,2,\ldots,T-b  \}$, \textcolor{black}{where the number of negative indices is equal to the number of positive indices}.  Then $E^*\left( P_{12} + P_{22} \right) = 0 $, and we further have
		\textcolor{black}{	\begin{align}
				E^*\left(P_{12} + P_{22} \right)^2 & =  \frac{1}{T} E^*\left( \sum_{m\in\mathcal{P}}Z_m - \sum_{m\in\mathcal{N}}{Z}_m \right)^2 \nonumber
				\\ & = \frac{1}{T}\left[ \left(\frac{1}{2}\right)^k  
				+ \left(\frac{1}{2}\right)^kC_k^1
				+ \left(\frac{1}{2}\right)^kC_k^2 + \ldots + \left(\frac{1}{2}\right)^k C_k^k
				\right] E^*\left(k Z_m^2 \right)  \nonumber
				\\ & 
				= \frac{k}{T}E^*(Z_m^2) = O_p(Tb),
				\label{squared of P12 + P22}
		\end{align} }
		where $C_k^i = \frac{k!}{i!(k-i)!}$, and $E^*(Z_m^2) = \frac{1}{T-b}\sum_{t=1}^{T-b}\left( \sum_{s=0}^{b-1}\frac{y_{t+s-1}}{\sigma_{t+s}} \right)^2 = O_p(b^2T)$. Furthermore, by Theorem 2 we have    $T\left(\gamma_{0}-\hat{\gamma}_{LAD} \right) = O_p(1)$.  As a result, $\left(\gamma_{0}-\hat{\gamma}_{LAD} \right)(P_{12} + P_{22})= O_{p^*}(b^{1/2}T^{-1/2}) $, which leads to
		\begin{equation}
			P=\frac{1}{\sqrt{T}}\sum_{m\in\mathcal{P}}\sum_{s=0}^{b-1}%
			{\varepsilon}_{i_{m}+s}+\frac{1}{\sqrt{T}}\sum_{m\in\mathcal{N}}\sum_{s=0}^{b-1}\left(-{\varepsilon}_{i_{m}+s+T+1}\right)+o_{p^*}(1). 
		\end{equation}
		By taking the same arguments to those of proving $P$, we also have $G=o_{p^*}(1)$. Thus
		\begin{equation}
			\frac{1}{\sqrt{T}}\sum_{t=1}^{T}\varepsilon_{t}^{*}=\frac{1}{\sqrt{T}}\sum_{m\in\mathcal{P}}\sum_{s=0}^{b-1}%
			{\varepsilon}_{i_{m}+s}+\frac{1}{\sqrt{T}}\sum_{m\in\mathcal{N}}\sum_{s=0}^{b-1}\left(-{\varepsilon}_{i_{m}+s+T+1}\right)+o_{p^*}(1).\label{PPQ}
		\end{equation}
		In a similar manner, we rewrite 
		\begin{align}
			\frac{1}{\sqrt{T}}\sum_{t=1}^{T}sgn(\varepsilon_{t}^{*})  &  =\frac
			{1}{\sqrt{T}}\sum_{m\in\mathcal{P}}\sum_{s=0}^{b-1}sgn(\hat{\varepsilon
			}_{i_{m}+s})+\frac{1}{\sqrt{T}}\sum_{m\in\mathcal{N}}\sum_{s=0}^{b-1}%
			sgn(-\hat{\varepsilon}_{i_{m}+s+T+1})+\frac{1}{\sqrt{T}}\sum_{t=bk+1}%
			^{T}sgn(\varepsilon_{t}^{*})\nonumber\\
			&  =W_{1}+W_{2}+W_{3}.
			\label{eq:sgn sum decompose}
		\end{align}
		Here we mainly focus on the proofs of $W_1$ and $W_2$ since $W_3=o_{p^*}(1)$ by a simple calculation. Because $sgn\left(  \cdot\right)  $ is not everywhere differentiable and we cannot directly take a Taylor expansion with $sgn\left(  \cdot\right)  $, we proceed by treating the function $sgn\left(  \cdot\right)  $ as a generalized function with a smooth regular sequence $sgn_{x}\left(  \cdot\right)  $ defined on an appropriate set of test functions (see \cite{Phillips1995} and \cite{Xiao2012}), satisfying $sgn_{x}\left(  \cdot\right) \rightarrow sgn\left(  \cdot\right)  $, and $sgn_{x}^{\prime}\left( \cdot\right)  \rightarrow sgn^{\prime}\left(  \cdot\right)  $ as $x\rightarrow\infty$. By applying the mean value theorem, we decompose $sgn_x(\hat{\varepsilon}_t)$ as  
		\begin{equation}
			sgn_{x}\left(  \hat{\varepsilon}_{t}\right)  =sgn_{x}\left(  \varepsilon
			_{t}\right)  -\left(  \hat{\gamma}_{LAD}-\gamma_{0}\right)  \frac{{y}_{t-1}%
			}{\sigma_{t}}sgn_{x}^{\prime}\left(  \varepsilon_{*t}\right), \label{app}
		\end{equation}
		where $\varepsilon_{*t}=\varepsilon_{t}-\lambda  \left(\hat{\gamma}_{LAD}-\gamma_{0}\right){{y}_{t}}/{\sigma_{t}}  $, $\lambda\in\left[  0,1\right]  $. Thus, 
		\begin{align}
			W_{1} &=\frac{1}{\sqrt{T}}\sum_{m\in\mathcal{P}}\sum_{s=0}^{b-1}sgn_{x}\left({\varepsilon}_{i_{m}+s}\right)  -(\hat{\gamma}_{LAD}-\gamma_{0})\frac {1}{\sqrt{T}}\sum_{m\in\mathcal{P}}\sum_{s=0}^{b-1}\frac{{y}_{i_{m}+s-1}}{\sigma_{i_{m}+s}}sgn_{x}^{\prime}({\varepsilon}_{*i_{m}+s}) \nonumber
			\\&=W_{11}+(\gamma_{0}- \hat{\gamma}_{LAD})W_{12}, 
		\end{align}
		and
		\begin{align}
			W_{2} &=-\frac{1}{\sqrt{T}}\sum_{m\in\mathcal{N}}\sum_{s=0}^{b-1}sgn_{x}\left({\varepsilon}_{i_{m}+s+T+1}\right) + (\hat{\gamma}_{LAD}-\gamma_{0})\frac{1}{T^{1/2}}\sum_{m\in\mathcal{N}}\sum_{s=0}^{b-1}\frac{{y}_{i_{m}+s+T}}{\sigma_{i_{m}+s+T+1}}sgn_{x}^{\prime}({\varepsilon}_{*i_{m}+s+T+1}) \nonumber
			\\&=W_{21}+(\gamma_{0}-\hat{\gamma}_{LAD})W_{22}. 
		\end{align}
		Since $sgn_{x}^{\prime}({\cdot})$ is a bounded function for any positive and finite $x$, using the same techniques as for proving (\ref{squared of P12 + P22}), we obtain $E^*(W_{12} + W_{22})^2 = O_p(Tb)$, which gives \textcolor{black}{$(\gamma_0 - \hat{\gamma}_{LAD})(W_{12}+W_{22}) = o_{p^*}(1)$}. Consequently, we have
		\begin{equation}
			\frac{1}{\sqrt{T}}\sum_{t=1}^{T}sgn(\varepsilon_{t}^{*})=\frac{1}{\sqrt{T}%
			}\sum_{m\in\mathcal{P}}\sum_{s=0}^{b-1}sgn\left(  {\varepsilon}_{i_{m}%
				+s}\right)  +\frac{1}{\sqrt{T}}\sum_{m\in\mathcal{N}}\sum_{s=0}^{b-1}%
			sgn(-\varepsilon_{i_{m}+s+T+1}) + o_{p^*}(1).\label{SGNQ}%
		\end{equation}

		Let $\nu_{t}^{*}=\left(		\varepsilon_{t}^{*},sgn(\varepsilon_{t}^{*})\right)  ^{\prime}$. Note $E^{*}(\nu_{t}^{*})=0$, then 
		\begin{align}
			E^{*}(\frac{1}{{T}}\sum_{t=1}^{T}\nu_{t}^{*}\sum_{t=1}^{T}\nu_{t}^{*\prime})  &  =\left[
			\begin{array}
				[c]{cc}%
				E^{*}\left(  \frac{1}{\sqrt{T}}\sum_{t=1}^{T}\varepsilon_{t}^{*}\right)
				^{2} & \frac{1}{T}E^{*}\left(  \sum_{t=1}^{T}\varepsilon_{t}^{*}\right)
				\left(  \sum_{t=1}^{T}sgn(\varepsilon_{t}^{*})\right) \\
				\frac{1}{T}E^{*}\left(  \sum_{t=1}^{T}\varepsilon_{t}^{*}\right)
				\left(  \sum_{t=1}^{T}sgn(\varepsilon_{t}^{*})\right)  & E^{*}\left(
				\frac{1}{\sqrt{T}}\sum_{t=1}^{T}sgn(\varepsilon_{t}^{*})\right)  ^{2}%
			\end{array}
			\right] \nonumber
			\\ & := 
			\left[
			\begin{array}
				[c]{cc}%
				{\delta_{1}^{*}}^{2} & \delta_{12}^{*}\\
				\delta_{12}^{*} & {\delta_{2}^{*}}^{2}%
			\end{array}
			\right].
		\end{align}
		By (\ref{PPQ}),  $\frac{1}{\sqrt{T}}\sum_{t=1}^{T}\varepsilon_{t}^{*}$ can be approximated by the scaled sample mean of \textcolor{black}{a block bootstrap sample}. As shown by \cite{Kunsch1989} and Theorem 3.1 of \cite{Lahiri2003}, the variance of the scaled bootstrap sample mean based on the block bootstrap observations  converge to the corresponding population counterparts. \textcolor{black}{Further notice that sequence $\{-\varepsilon_1,\ldots, -\varepsilon_T, \varepsilon_1,\ldots,\varepsilon_T \}$ exhibits the same mixing dependence as  the series $\{\varepsilon_{t}\}_{t=1}^{T}$}. Consequently we have ${\delta_{1}^{*}}^{2}\overset{p^*}{\rightarrow}\delta_{1}^{2}$, where $\delta_{1}^{2}$ is the variance of $\frac{1}{\sqrt{T}}\sum_{t=1}^{T}\varepsilon_t$. Similarly, we also have ${\delta_{2}^{*}}^{2}\overset{p^*}{\rightarrow}\delta_{2}^{2}$. Finally, based on (\ref{PPQ}) and (\ref{SGNQ}) we have
		\begin{align*}
			\delta_{12}^{*}  &  =\frac{1}{T}E^{*}\left[  \left(  \sum_{m\in
				\mathcal{P}}\sum_{s=0}^{b-1}{\varepsilon}_{i_{m}+s}\right)  \left(  \sum
			_{m\in\mathcal{P}}\sum_{s=1}^{b}sgn\left(  {\varepsilon}_{i_{m}+s}\right)
			\right)  \right] \\
			&  -\frac{1}{T}E^{*}\left[  \left(  \sum_{m\in\mathcal{P}}\sum_{s=0}%
			^{b-1}{\varepsilon}_{i_{m}+s}\right)  \left(  \sum_{m\in\mathcal{N}}\sum
			_{s=0}^{b-1}sgn(\varepsilon_{i_{m}+s+T+1})\right)  \right] \\
			&  -\frac{1}{T}E^{*}\left[  \left(  \sum_{m\in\mathcal{N}}\sum_{s=0}%
			^{b-1}{\varepsilon}_{i_{m}+T+1+s}\right)  \left(  \sum_{m\in\mathcal{P}}%
			\sum_{s=1}^{b}sgn\left(  {\varepsilon}_{i_{m}+s}\right)  \right)  \right] \\
			&  +\frac{1}{T}E^{*}\left[  \left(  \sum_{m\in\mathcal{N}}\sum_{s=0}%
			^{b-1}{\varepsilon}_{i_{m}+T+1+s}\right)  \left(  \sum_{m\in\mathcal{N}}%
			\sum_{s=0}^{b-1}sgn(\varepsilon_{i_{m}+s+T+1})\right)  \right]  + o_{p^*}(1)\\
			&  =Cov_{1}-Cov_{2}-Cov_{3}+Cov_{4}+  o_{p^*}(1).
		\end{align*}
		
		Under the bootstrap probability law, it is evident that $Cov_{2} = Cov_{3} = 0$. For $Cov_{1}$, we have
		\begin{align} 
			Cov_{1} &  =\frac{1}{T}E^{*}\left[ \sum_{m\in\mathcal{P}}\left(  \sum_{s=0}^{b-1}|{\varepsilon}_{i_{m}+s}|+\sum_{s=0}^{b-1}\sum_{s^{\prime}\neq s}^{b-1}{\varepsilon}_{i_{m}+s}sgn({\varepsilon}_{i_{m}+s^{\prime}})\right)\right]  \nonumber
			\\&  
			= \frac{1}{T} \left[
			\left(\frac{1}{2}\right)^{k}\left(k C_{k}^0 + (k-1)C_{k}^1 + \ldots + C_{k}^{k-1} \right)E^*\left( \sum_{s=0}^{b-1}\left|\varepsilon_{i_m+s} \right|\right)  
			\right] \nonumber
			\\ & 
			+ \frac{1}{T} \left[
			\left(\frac{1}{2}\right)^{k}\left(k C_{k}^0 + (k-1)C_{k}^1 + \ldots + C_{k}^{k-1} \right)E^*\left( \sum_{s=0}^{b-1}\sum_{s^{\prime}\neq s}^{b-1}{\varepsilon}_{i_{m}+s}sgn({\varepsilon}_{i_{m}+s^{\prime}})\right)  
			\right],  
		\end{align}
		where $E^{*}|\varepsilon_{i_{m}+s}|=\frac{1}{T-b}\sum_{t=1}^{T-b}|\varepsilon_{t+s}|\overset{p}{\rightarrow}E|\varepsilon_{t}|$ given $i_m \in \mathcal{P}$ and
		\begin{align*}
			\sum_{s=0}^{b-1}&\sum_{s^{\prime}\neq s}^{b-1}E^{*}\left(  {\varepsilon}_{i_{m}+s}sgn({\varepsilon}_{i_{m}+s^{\prime}})\right)   
			\\ &  =\sum_{s=0}^{b-1}\sum_{s^{\prime}>s}^{b-1}E^{*}\left(  {\varepsilon}_{i_{m}+s}%
			sgn({\varepsilon}_{i_{m}+s^{\prime}})\right)  +\sum_{s=0}^{b-1}\sum_{s^{\prime}<s}^{b-1}E^{*}\left(  {\varepsilon}_{i_{m}+s}sgn({\varepsilon}_{i_{m}+s^{\prime}})\right) \\
			&  =\sum_{s=0}^{b-1}\sum_{s^{\prime}>s}^{b-1}\frac{1}{T-b}\sum_{t=1}
			^{T-b}\varepsilon_{t+s}sgn(\varepsilon_{t+s^{\prime}})+\sum_{s=0}^{b-1}\sum_{s^{\prime}<s}^{b-1}\frac{1}{T-b}\sum_{t=1}^{T-b}\varepsilon_{t+s}sgn(\varepsilon_{t+s^{\prime}})\\
			&  \overset{p}{\rightarrow}\sum_{s=0}^{b-1}\sum_{s^{\prime}>s}^{b-1}E(\varepsilon_{t}sgn(\varepsilon_{t+s^{\prime}-s}))+\sum_{s=0}^{b-1}\sum_{s^{\prime}<s}^{b-1}E(\varepsilon_{t+s-s^{\prime}}sgn(\varepsilon_{t})).
		\end{align*}
		
		For $Cov_4$, taking the same arguments as for $Cov_1$, we have
		\begin{align*}
			Cov_4 & = \frac{1}{T} \left[
			\left(\frac{1}{2}\right)^{k}\left( C_{k}^1 + \ldots + (k-1)C_{k}^{k-1} + kC_k^k \right)E^*\left( \sum_{s=0}^{b-1}\left|\varepsilon_{i_m+s+T+1} \right|\right)  
			\right]\nonumber
			\\ & 
			+ \frac{1}{T} \left[
			\left(\frac{1}{2}\right)^{k}\left( C_{k}^1 + \ldots + (k-1)C_{k}^{k-1} + kC_k^k \right)E^*\left( \sum_{s=0}^{b-1}\sum_{s^{\prime}\neq s}^{b-1}{\varepsilon}_{i_{m}+s+T+1}sgn({\varepsilon}_{i_{m}+s^{\prime}+T+1})\right)  
			\right].  
		\end{align*}
		Hence
		\begin{align*}
			\delta_{12}^{*} & = Cov_1 + Cov_4 + o_{p^*}(1)   
			\\ & \overset{p^*}{\rightarrow}
			\frac{1}{T}k\left[ \sum_{s=0}^{b-1}E|\varepsilon_t| + \sum_{s=0}^{b-1}\sum_{s^{\prime}>s}^{b-1}E(\varepsilon_{t}sgn(\varepsilon_{t+s^{\prime}-s}))+\sum_{s=0}^{b-1}\sum_{s^{\prime}<s}^{b-1}E(\varepsilon_{t+s-s^{\prime}}sgn(\varepsilon_{t}))
			\right]
			\\&  \overset{p}{\rightarrow}E|\varepsilon_{t}|+\sum_{s=2}^{\infty}E(\varepsilon_{1}sgn(\varepsilon_{s}))+\sum_{s=2}^{\infty}E(\varepsilon_{s}sgn(\varepsilon_{1}))=\delta_{12}.
		\end{align*}
		Thus, we obtain $E^{*}\left(\frac{1}{T}\sum_{t=1}^{T}\nu_{t}^{*}\sum_{t=1}^{T}\nu_{t}^{*\prime}\right)\overset{p^{*}}{\rightarrow}\Sigma$ as $1/b + b/T + 1/T \rightarrow 0$.
		
		Now, we prove the bivariate invariance principle based on Theorem 2.1 of \cite{Phillips1986}. To apply their theorem, we are required to verify each of the following conditions. Firstly, we have proved that $E^{*}\left(\frac{1}{T}\sum_{t=1}^{T}\nu_{t}^{*}\sum_{t=1}^{T}\nu_{t}^{*\prime}\right)\overset{p^{*}}{\rightarrow}\Sigma$ and  $E^{*}(\nu_{t}^{*}) = 0$. Secondly, $sup_{t}\left(  E^{*}|\varepsilon_{t}^{*}|^{p}\right)  < \infty$ for some $p > 2$ by $E|\varepsilon_{t}|^{p} <\infty$ of Assumption 2, and thus  $\{{\varepsilon_{t}^{*}}^{2}\}$ is uniformly integrable. Thirdly, it is straightforward that $E^{*}\left(\frac{1}{T}\sum_{t=z}^{T+z}\nu_{t}^{*}\sum_{t=z}^{T+z}\nu_{t}^{*\prime}\right)$ converges to $\Sigma$ for $z \geq1$ as  $min(z,T) \rightarrow\infty$ because  the partial sum of sequence $\{\nu_{z}^{*},\cdots, \nu_{T+z}^{*}\}$ can be decomposed into the sum of blocks, as shown by (\ref{boot decompose equ}). Finally, the pseudo series $\{\varepsilon_{t}^{*}\}$ satisfies the mixing condition of $\{\varepsilon_t\}$ since the dependence within each block is unchanged and the relationships between different blocks are \textcolor{black}{severed}, and the $\alpha$-mixing condition of size $-\beta p/(p-\beta)$ with $p>\beta > 2$ is more stringent than that of  size $-\beta/(\beta-2)$. As a result, we have
		\begin{equation}
			\frac{1}{\sqrt{T}}\sum_{t=1}^{\lfloor Tr \rfloor}\left(  \varepsilon_{t}^{*
			},sgn(\varepsilon_{t}^{*})\right)  ^{\prime}\overset{d^{*}}{\rightarrow
			}\left(  B_{1}(r),B_{2}(r)\right)  ^{\prime}:=\Sigma^{1/2}\left(
			W_{1}(r),W_{2}(r)\right)  ^{\prime}.\label{BBI}
		\end{equation}
		Furthermore, the result (3.4) holds true naturally by applying the same proof techniques as those used in Lemma 1. 		
	\end{proof}
	
	\subsection{Proof of Lemma 4}
	\label{section:proof lemma4}
	\begin{proof}[Proof of Lemma 4:]
		We still abbreviate $b_T$ as $b$ for convenience. For the feasible bootstrapped errors $\{\check{\varepsilon}_{t}^{*}\}_{t=1}^{T}$, we follow the proof in Subsection \ref{subsection proof of lemma 3} and make the following decomposition
		\begin{equation*}
			\frac{1}{\sqrt{T}}\sum_{t=1}^{T}\check{\varepsilon}_{t}^{*}=\frac{1}{\sqrt{T}%
			}\sum_{t=1}^{bk}{\check{\varepsilon}}_{t}^{*}+\frac{1}{\sqrt{T}}%
			\sum_{t=bk+1}^{T}{\check{\varepsilon}}_{t}^{*}=\check{P}+\check{G},
		\end{equation*}
		where $k=\lfloor(T-1)/b\rfloor$.
		
		Based on the sets $\mathcal{P}$ and $\mathcal{N}$ defined in Subsection \ref{subsection proof of lemma 3}, we have
		\begin{equation}
			\check{P}=\frac{1}{\sqrt{T}}\sum_{m\in \mathcal{P}}\sum_{s=0}^{b-1}\check{{%
					\varepsilon }}_{i_{m}+s}-\frac{1}{\sqrt{T}}\sum_{m\in \mathcal{N}%
			}\sum_{s=0}^{b-1}\check{{\varepsilon }}_{i_{m}+s+T+1}=\check{P}_{1}+\check{P}%
			_{2}. \notag
		\end{equation}%
		It follows that
		\begin{align*}
			\check{P}_{1}& =\frac{1}{\sqrt{T}}\sum_{m\in \mathcal{P}}\sum_{s=0}^{b-1}%
			\varepsilon _{i_{m}+s}+(\gamma _{0}-\hat{\gamma}_{LAD})\frac{1}{\sqrt{T}}%
			\sum_{m\in \mathcal{P}}\sum_{s=0}^{b-1}\frac{y_{i_{m}+s-1}}{\hat{\sigma}%
				_{i_{m}+s}} \\
			& +\frac{1}{\sqrt{T}}\sum_{m\in \mathcal{P}}\sum_{s=0}^{b-1}\varepsilon
			_{i_{m}+s}\left( \frac{\sigma _{i_{m}+s}}{\hat{\sigma}_{i_{m}+s}}-1\right)
			\\
			& =P_{11}+(\gamma _{0}-\hat{\gamma}_{LAD})\check{P}_{12}+\check{P}_{13},
		\end{align*}
		where $P_{11}$ has been defined in Subsection \ref{subsection proof of lemma 3}.
		For $\check{P}_{13}$, 
		\begin{equation*}
			\check{P}_{13}=\frac{1}{\sqrt{T}}\sum_{m\in \mathcal{P}}\sum_{s=0}^{b-1}%
			\varepsilon _{i_{m}+s}\frac{\left( \sigma _{i_{m}+s}-\hat{\sigma}%
				_{i_{m}+s}\right) }{\sigma _{i_{m}+s}}+\frac{1}{\sqrt{T}}\sum_{m\in \mathcal{%
					P}}\sum_{s=0}^{b-1}\varepsilon _{i_{m}+s}\frac{\left( \sigma _{i_{m}+s}-\hat{%
					\sigma}_{i_{m}+s}\right) ^{2}}{\hat{\sigma}_{i_{m}+s}\textcolor{black}{{\sigma}_{i_{m}+s}} }=\check{P}_{131}+%
			\check{P}_{132}.
		\end{equation*}
		By using Lemma \ref{Lemma other}, we have $\check{P}_{131}=o_{p}\left( 1\right) $, and by using Lemma \ref{Lemma other new2}, we have $\check{P}_{132}=o_{p}\left( 1\right) $. Hence $%
		\check{P}_{13}=o_{p}\left( 1\right) $.	 As a result, we obtain
		\begin{equation*}
			\check{P}_{1}=P_{11}+(\gamma _{0}-\hat{\gamma}_{LAD})\check{P}_{12}+o_{p}\left(
			1\right).
		\end{equation*}
		In a similar manner, we can show
		\begin{equation*}
			\check{P}_{2}=P_{21}+(\gamma _{0}-\hat{\gamma}_{LAD})\check{P}_{22}+o_{p}(1),
		\end{equation*}%
		where $P_{21}=-\frac{1}{\sqrt{T}}\sum_{m\in \mathcal{N}}\sum_{s=0}^{b-1}%
		\varepsilon _{i_{m}+s+T+1}$ and $\check{P}_{22}=-\frac{1}{\sqrt{T}}\sum_{m\in
			\mathcal{N}}\sum_{s=0}^{b-1}\frac{y_{i_{m}+s+T}}{\hat{\sigma}_{i_{m}+s+T+1}}.$
		We note that the sequence $\{\hat{\sigma}_t\}_{t=1}^T$ is nonstochastic \textcolor{black}{under} the bootstrap probability \textcolor{black}{measure}, thus  $(\gamma_{0}-\hat{\gamma}_{LAD})(\check{P}_{12}+\check{P}_{22}) = o_{p^*}(1)$ by the same techniques used for the proof of $(\gamma_{0}-\hat{\gamma}_{LAD})(P_{12}+P_{22})$ in Subsection \ref{subsection proof of lemma 3}. Similarly, we can show  $\check{G}=o_{p^\ast}(1)$. Thus, we obtain
		\begin{equation}
			\frac{1}{\sqrt{T}}\sum_{t=1}^{T}\check{{\varepsilon }}_{t}^{\ast }=\frac{1}{%
				\sqrt{T}}\sum_{t=1}^{T}{\varepsilon }_{t}^{\ast }+o_{p^{\ast }}(1).
			\label{partial sum converge1}
		\end{equation}
		By using the same arguments as for $\frac{1}{\sqrt{T}}\sum_{t=1}^{T}\check{{\varepsilon }}_{t}^{\ast }$ and the approximation techniques for $sgn(\cdot )$ in  (\ref{app}), we also have
		\begin{equation}
			\frac{1}{\sqrt{T}}\sum_{t=1}^{T}sgn(\check{\varepsilon}_{t}^{\ast })=\frac{1%
			}{\sqrt{T}}\sum_{t=1}^{T}sgn({\varepsilon }_{t}^{\ast })+o_{p^{\ast }}(1).
			\label{partial sum converge2}
		\end{equation}
		Based on (\ref{partial sum converge1}) and (\ref{partial sum converge2}), we have
		\begin{equation*}
			\frac{1}{\sqrt{T}}\sum_{t=1}^{\lfloor Tr \rfloor}\left( \check{\varepsilon}_{t}^{\ast
			},sgn(\check{\varepsilon}_{t}^{\ast })\right) ^{\prime }=\frac{1}{\sqrt{T}}%
			\sum_{t=1}^{\lfloor Tr \rfloor}\left( {\varepsilon }_{t}^{\ast },sgn({\varepsilon }%
			_{t}^{\ast })\right) ^{\prime }+o_{p^{\ast }}(1).
		\end{equation*} 
		Using the result in (\ref{BBI}), we obtain
		\begin{align}
			\frac{1}{\sqrt{T}}\sum_{t=1}^{\lfloor Tr \rfloor}\left( \check{\varepsilon}_{t}^{\ast
			},sgn(\check{\varepsilon}_{t}^{\ast })\right) ^{\prime }
			\overset{d^{\ast }}{%
				\rightarrow }\left( B_{1}(r),B_{2}(r)\right) ^{\prime }.
		\end{align}
		
		To prove the second convergence stated in Lemma 4,  we consider  
		\begin{align}
			\frac{1}{\sqrt{T}}\sum_{t=1}^{T}\left( \check{{\varepsilon }}_{t}^{\ast }\hat{\sigma}_{t} - \varepsilon_t^*\sigma _{t}\right) & = 	\frac{1}{\sqrt{T}}\sum_{t=1}^{T}\hat{\sigma}_{t}\left( \check{{\varepsilon }}_{t}^{\ast } - \varepsilon_t^*\right) + 	\frac{1}{\sqrt{T}}\sum_{t=1}^{T}\varepsilon_t^*\left( \hat{\sigma}_{t} - \sigma _{t}\right) = \mathcal{R}_1 + \mathcal{R}_2.
			\label{DRR}
		\end{align}
		
		For $\mathcal{R}_1$,
		\begin{equation*}
			\mathcal{R}_1=\frac{1}{\sqrt{T}}\sum_{t=1}^{bk}\hat{\sigma}_t{\check{\varepsilon}}_{t}^{*}+\frac{1}{\sqrt{T}}
			\sum_{t=bk+1}^{T}\hat{\sigma}_t{\check{\varepsilon}}_{t}^{*}=\mathcal{R}_{11}+\mathcal{R}_{12}.
		\end{equation*}
		Let the subscript  $\tilde{t}$ of $\hat{\sigma}_{\tilde{t}}$ denote an adapted index from the original sequence. For instance, if  $\{\check{\varepsilon}_t^*\}_{t=1}^{b}$ comes from the block $i_1$ and $1 \in \mathcal{P}$, then $\tilde{t} = 1,\ldots, b$ corresponds to $\check{\varepsilon}_{i_1}, \ldots, \check{\varepsilon}_{i_1 + b-1}$ respectively. If $\{\check{\varepsilon}_t^* \}_{t=b+1}^{2b}$ comes from the block $i_4$ and $4 \in \mathcal{N}$, then $\tilde{t} = b+1,\ldots, 2b$ corresponds to $\check{\varepsilon}_{i_4+T+1}, \ldots, \check{\varepsilon}_{i_4 + T+b}$.  Then $\mathcal{R}_{11}$ can be rewritten as
		\begin{equation}
			\mathcal{R}_{11}=\frac{1}{\sqrt{T}}\sum_{m\in \mathcal{P}}\sum_{s=0}^{b-1}\hat{\sigma}_{\tilde{t}}\check{{\varepsilon }}_{i_{m}+s}-\frac{1}{\sqrt{T}}\sum_{m\in \mathcal{N}%
			}\sum_{s=0}^{b-1}\hat{\sigma}_{\tilde{t}}\check{{\varepsilon }}_{i_{m}+s+T+1}=\mathcal{R}_{111}+\mathcal{R}_{112}. \notag
		\end{equation}
		
		For $\mathcal{R}_{111}$,
		\begin{align*}
			\mathcal{R}_{111}& =\frac{1}{\sqrt{T}}\sum_{m\in \mathcal{P}}\sum_{s=0}^{b-1}\hat{\sigma}_{\tilde{t}}
			\varepsilon _{i_{m}+s}+(\gamma _{0}-\hat{\gamma}_{LAD})\frac{1}{\sqrt{T}}\sum_{m\in \mathcal{P}}\sum_{s=0}^{b-1}\hat{\sigma}_{\tilde{t}}\frac{y_{i_{m}+s-1}}{\hat{\sigma}_{i_{m}+s}} \\
			& +\frac{1}{\sqrt{T}}\sum_{m\in \mathcal{P}}\sum_{s=0}^{b-1}\hat{\sigma}_{\tilde{t}}\varepsilon
			_{i_{m}+s}\left( \frac{\sigma _{i_{m}+s}}{\hat{\sigma}_{i_{m}+s}}-1\right)
			\\
			& =\mathcal{R}_{1111}+(\gamma _{0}-\hat{\gamma}_{LAD})\mathcal{R}_{1112}+\mathcal{R}_{1113}.
		\end{align*} 
		For $\mathcal{R}_{1113}$, rewrite it as
		\begin{align*}
			\mathcal{R}_{1113} &=\frac{1}{\sqrt{T}}\sum_{m\in \mathcal{P}}\sum_{s=0}^{b-1}%
			\hat{\sigma}_{\tilde{t}}\varepsilon _{i_{m}+s}\frac{\left( \sigma _{i_{m}+s}-\hat{\sigma}%
				_{i_{m}+s}\right) }{\sigma _{i_{m}+s}}+\frac{1}{\sqrt{T}}\sum_{m\in \mathcal{%
					P}}\sum_{s=0}^{b-1}\hat{\sigma}_{\tilde{t}}\varepsilon _{i_{m}+s}\frac{\left( \sigma _{i_{m}+s}-\hat{%
					\sigma}_{i_{m}+s}\right) ^{2}}{\hat{\sigma}_{i_{m}+s}\textcolor{black}{{\sigma}_{i_{m}+s}}}.
		\end{align*}
		The second term on R.H.S is $o_p(1)$ by using  Lemma \ref{lemma new other} (ii) and the techniques applied for Lemma \ref{Lemma other new2}. For the first term on R.H.S,  because ${\sigma}_{\tilde{t}}$ is bounded by Assumption 3, by using \textcolor{black}{Lemma \ref{lemma new other} (iii), (vi), (viii)} and \ref{Lemma other}, it can be shown that the first term is $o_p(1)$. Then $
		\mathcal{R}_{111} = 	\mathcal{R}_{1111} + (\gamma_{0}-\hat{\gamma}_{LAD})	\mathcal{R}_{1112} + o_p(1)
		$. In a similar way, we also have 
		$\mathcal{R}_{112} = 	\mathcal{R}_{1121} + (\gamma_{0}-\hat{\gamma}_{LAD})	\mathcal{R}_{1122} + o_p(1)
		$, where $\mathcal{R}_{1121} = - \frac{1}{\sqrt{T}}\sum_{m\in \mathcal{N}}\sum_{s=0}^{b-1}\hat{\sigma}_{\tilde{t}}{{\varepsilon }}_{i_{m}+s+T+1}$ and $\mathcal{R}_{1122} = -\frac{1}{\sqrt{T}}\sum_{m\in \mathcal{N}}\sum_{s=0}^{b-1}\hat{\sigma}_{\tilde{t}}\frac{{\varepsilon }_{i_{m}+s+T+1}}{\hat{\sigma}_{i_m+s+T+1}}$. Once again, because the sequence $\{ \hat{\sigma}_t\}$ is nonstochastic conditional on the original data, we follow the same logic as for the proof of $(\gamma_0 - \hat{\gamma}_{LAD})(\check{P}_{12}+\check{P}_{22})$, and then have $(\gamma_0 - \hat{\gamma}_{LAD})(\mathcal{R}_{1112}+\mathcal{R}_{1122})=o_{p^*}(1)$. In addition, it is not hard to show $\mathcal{R}_{12} = o_p(1)$. Consequently, we can conclude that
		$$
		\frac{1}{\sqrt{T}}\sum_{t=1}^T\hat{\sigma}_t \check{\varepsilon}_t^* = \frac{1}{\sqrt{T}}\sum_{t=1}^T\hat{\sigma}_t{\varepsilon}_t^* + o_{p^*}(1).
		$$	
		
		Next, we proceed to show $\mathcal{R}_2 = o_{p^*}(1)$. It follows that
		\begin{align*}
			E^*\left(\mathcal{R}_2^2 \right)  & = \frac{1}{T}\sum_{t=1}^{T} E^*\left(\varepsilon_t^{*2} \right)\left(\hat{\sigma}_t - \sigma_t \right)^2  
			+ \frac{1}{T}\sum_{t=1}^{T}\sum_{t^\prime \neq t}^T E^*\left(\varepsilon_t^*\varepsilon_{t^\prime}^* \right)\left(\hat{\sigma}_t - \sigma_t \right)\left(\hat{\sigma}_{t^\prime} - \sigma_{t^\prime} \right). 
		\end{align*} 
		For the situations where $t$ and $t^\prime$ are from the same block, we have   $E^*\left(\varepsilon_t^*\varepsilon_{t^\prime}^* \right) \overset{p}{\rightarrow} E(\varepsilon_t\varepsilon_{t^\prime})$, and for the situations where $t$ and $t^\prime$ are from different blocks, we have $E^*\left(\varepsilon_t^*\varepsilon_{t^\prime}^* \right) = 0 $. By using the mixing condition of $\varepsilon_t$, we then have 
		\begin{align*}
			E^*\left(\mathcal{R}_2^2 \right)  & \leq \frac{1}{T}\sum_{t=1}^{T} \left(\hat{\sigma}_t - \sigma_t \right)^2  
			+ max_t \left\vert\hat{\sigma}_t -{\sigma}_t \right\vert\frac{1}{T}\sum_{t=1}^{T} \left\vert\hat{\sigma}_t - \sigma_t \right\vert 
			\\ & \leq \frac{C}{T}\sum_{t=1}^{T}\left(\hat{\sigma}_t -\tilde{\sigma}_t \right)^2  + \frac{C}{T}\sum_{t=1}^{T}\left(\tilde{\sigma}_t -\check{\sigma}_t \right)^2
			+ \frac{C}{T}\sum_{t=1}^{T}\left(\check{\sigma}_t -\sigma_t \right)^2 
			\\ & + C\left[
			max_t
			\left\vert\hat{\sigma}_t -\tilde{\sigma}_t   \right\vert+ max_t\left\vert \tilde{\sigma}_t -\check{\sigma}_t   \right\vert + \frac{1}{T}\sum_{t=1}^T\left\vert \check{\sigma}_t - \sigma_t  \right\vert\right]
			\\ & = o_p(1)
		\end{align*}
		by using Lemma \ref{lemma new other}.  This completes the proof of Lemma 4. 				
	\end{proof}
	
	\bigskip
	
	
	\subsection{Proof of Theorems 3 and 4}
	\begin{proof}[Proof of Theorem 3:]
		Following the proof of  Theorem 2, and using the results in  Lemma 3, we immediately have Theorem 3.
	\end{proof}

	
	\begin{proof}[Proof of Theorem 4:]
		Following the proof of  Theorem 2, and using the results in  Lemma 4, we immediately have Theorem 4.
	\end{proof}
	
	\section{Proof of Theorems in Section 4}\label{section: proof 4}
	
	
	\subsection{Proof of Theorem 5}\label{subsection proof theorem5}
	
	\begin{proof}[Proof of Theorem 5:]
		
		(a) Consider the case of $d_t = 1 $. From (4.3), it follows that 
		\begin{align}
			\hat{\mu}(\bar{c}) = \left( 1 + \frac{\bar{c}^2}{T} \right)^{-1}\left(x_1 + \frac{\bar{c}}{T}\sum_{t=2}^{T}\left(\Delta x_t + \frac{\bar{c}}{T}x_{t-1} \right) \right)  = x_1 + O_p\left(\frac{1}{\sqrt{T}}\right),
			\label{hat mu = x1}
		\end{align}
		where $\Delta x_t = \sigma_t\varepsilon_t -\frac{c}{T}y_{t-1}$ under $\mathbb{H}_A$ and $y_{t} = O_p(\sqrt{T})$. Now we rewrite (4.4) as
		\begin{align}
			\sum_{t=1}^{T}\left| y_t^d - \gamma y_{t-1}^d \right| = \sum_{t=1}^{T} \left| u_t - vT^{-1}\left(y_{t-1} + \mu -\hat{\mu}(\bar{c}) \right)  + (1-\gamma_{0})\left(\mu -\hat{\mu}(\bar{c}) \right)  \right|,\label{cof}
		\end{align}
		where $v = T(\gamma -\gamma_0)$. Following the same arguments as in Section 2 of the main paper, we just need to derive the limiting distributions of   $\frac{1}{T}\sum_{t=1}^{T}\left(y_{t-1} + \mu -\hat{\mu}(\bar{c}) \right)sgn\left( u_t\right)$ and 
		$2\sum_{t=1}^{T}\int_{0}^{vT^{-1}\left( y_{t-1} + \mu -\hat{\mu}(\bar{c}) \right) }\left[I\left(  u_{t}\leq s\right)  -I\left(  u_{t}\leq0\right)  \right]  ds $ since $(1-\gamma_0)(\mu-\hat{\mu}(\bar{c}))= O_p(\frac{1}{T})$.
		Based on (\ref{yt-1 sgn alter}) and $\hat{\mu}(\bar{c}) - \mu = y_1 + O_p\left(\frac{1}{\sqrt{T}}\right) = O_p(1)$, we have  
		\begin{align}
			\frac{1}{T}\sum_{t=1}^{T}\left(y_{t-1} + \mu -\hat{\mu}(\bar{c}) \right)sgn\left( u_t\right) \overset{d}{\rightarrow} \int_{0}^{1}\mathcal{J}_{c,\sigma}(r)dB_{2}(r)+\Gamma,\label{CH1}
		\end{align}	
		where  $\Gamma$ is defined in Lemma 2. Furthermore, following the proof of (\ref{yt-1 sgn alter longer}), we can also  obtain 
		\begin{align}		2\sum_{t=1}^{T}\int_{0}^{vT^{-1}\left( y_{t-1} + \mu -\hat{\mu}(\bar{c}) \right) }\left[
			I\left(  u_{t}\leq s\right)  -I\left(  u_{t}\leq0\right)  \right]  ds \overset{d}{\rightarrow} v^{2}\left( f\left(  0\right)  \int_{0}^{1}\sigma^{-1}\left(  r\right) \mathcal{J}^2_{c,\sigma}\left(  r\right)  dr\right).\label{CH2}
		\end{align}
		Combining (\ref{CH1}) and (\ref{CH2}),  and using the fact that  (\ref{cof}) is minimized at $\hat{v} = T(\hat{\gamma}^{d}_{LAD} -\gamma_0)$, we obtain the convergence of $L_T^d$ in Theorem 5 (a). The results of $t_T^d$ are straightforward based on the asymptotic results of $L_T^d$. $\Box$
		
		(b) Consider the case of $d_t = (1,t)^\prime$. Define $D_T^\mu =diag(1,\sqrt{T})$, we have  
		\begin{align}
			D_T^\mu\left( \hat{\mu}(\bar{c}) - \mu \right)=\left( \begin{array}{c}
				y_1 \\ \frac{T^{-1/2}\sum_{t=2}^{T} \left( 1+(\bar{c}/T)(t-1)\right)  \left(\Delta y_t + (\bar{c}/T)y_{t-1} \right)}{T^{-1}\sum_{t=2}^{T} (1+(\bar{c}/T)(t-1))^2}
			\end{array} \right)  +o_p(1),	
			\label{D hat mu - mu}
		\end{align}
		by  following the proof of Theorem 5.1 in \cite{Boswijk2018}.
		
		Let $\hat{\mu}_i(\bar{c})$ and $\mu_i$ ($i = 1,2$) denote the $i$-th element of $\hat{\mu}(\bar{c})$ and $\mu$ respectively.  (\ref{D hat mu - mu}) means that  
		\begin{align}
			&\hat{\mu}_1(\bar{c}) - \mu_1 = y_1 + o_p(1)
			\label{hat mu1 - mu1},
			\\
			\sqrt{T}\left(\hat{\mu}_2(\bar{c}) - \mu_2 \right) \overset{d}{\rightarrow} \mathcal{M}(\bar{c},c,\sigma) := &\left( \int_{0}^{1}\left(1+\bar{c}r \right)^2dr  \right)^{-1} \int_{0}^{1}(1+\bar{c}r)\left( d\mathcal{J}_{c,\sigma}(r) + \bar{c}\mathcal{J}_{c,\sigma}(r)dr \right). 
			\label{hat mu2- mu2}
		\end{align}
		
		Now we let $v = T(\gamma - \gamma_{0})$ and rewrite (4.4) as 
		\begin{align}
			\sum_{t=1}^{T}\left| y_t^d - \gamma y_{t-1}^d \right|= \sum_{t=1}^{T} \left| u_t - vT^{-1}\left(y_{t-1} + \left( \mu -\hat{\mu}(\bar{c})\right)^\prime d_{t-1}  \right)  + \left(\mu -\hat{\mu}(\bar{c}) \right)^\prime d_t - \gamma_{0} \left(\mu -\hat{\mu}(\bar{c}) \right)^\prime d_{t-1}  \right|. \label{trof}
		\end{align}
		Similarly, following the same arguments as in Section 2, we need to derive the limiting distributions of  $ \frac{1}{T}\sum_{t=1}^{T}\left(y_{t-1} + \left( \mu -\hat{\mu}(\bar{c})\right)^\prime d_{t-1}  \right)sgn\left( u_t\right)$ and $2\sum_{t=1}^{T}\int_{0}^{X_{t-1}}\left[
		I\left(  u_{t}\leq s\right)  -I\left(  u_{t}\leq0\right)  \right]  ds$, where $X_{t-1} = vT^{-1}\left( y_{t-1} + \left( \mu -\hat{\mu}(\bar{c})\right)^\prime d_{t-1}  \right)  - (t-1)(1-\gamma_{0})(\mu_2-\hat{\mu}_2(\bar{c})) - (\mu_2-\hat{\mu}_2(\bar{c}))$.
		
		By (\ref{yt-1 sgn alter}), (\ref{hat mu1 - mu1}) and (\ref{hat mu2- mu2}), we have  
		\begin{align}
			\frac{1}{T}\sum_{t=1}^{T}&\left(y_{t-1} + \left( \mu -\hat{\mu}(\bar{c})\right)^\prime d_{t-1}  \right)sgn\left( u_t\right)
			\overset{d}{\rightarrow}
			\int_{0}^{1}\mathcal{J}_{c,\sigma}(r)dB_2(r) -\mathcal{M}(\bar{c},c,\sigma)\int_{0}^{1}rdB_2(r) +\Gamma,\label{LTH1}
		\end{align}
		where ${T^{-3/2}}\sum_{t=1}^{T}(t-1)  sgn\left( u_t\right) \overset{d}{\rightarrow} \int_{0}^{1}rdB_2(r)$. Similarly, by following the arguments in the proof of (\ref{yt-1 sgn alter longer}), we have
		\begin{align*}
			2\sum_{t=1}^{T}&\int_{0}^{X_{t-1} }\left[
			I\left(  u_{t}\leq s\right)   -  I\left(  u_{t}\leq0\right)   \right]  ds
			\\&= 
			\sum_{t=1}^{T}\sigma_t^{-1}f_{t-1}(0)\frac{v^2}{T^2} \left( y_{t-1} + \left( \mu_1 - \hat{\mu}_1(\bar{c})\right)  \right)^2  
			+ \sum_{t=1}^{T}\sigma_t^{-1}f_{t-1}(0)\frac{v^2}{T^2}(\mu_2 - \hat{\mu}_2(\bar{c}))^2(t-1)^2
			\\ & +  2\sum_{t=1}^{T}\sigma_t^{-1}f_{t-1}(0)\frac{v^2}{T^2}\left( y_{t-1} + \left( \mu_1 - \hat{\mu}_1(\bar{c})\right)  \right)(\mu_2 - \hat{\mu}_2(\bar{c}))(t-1)
			\\ & - 
			2\sum_{t=1}^{T}\sigma_t^{-1}f_{t-1}(0) \frac{v}{T}\left( y_{t-1} + \left( \mu -\hat{\mu}(\bar{c})\right)^\prime d_{t-1}  \right)\left( (t-1)(1-\gamma_{0})(\mu_2-\hat{\mu}_2(\bar{c})) + (\mu_2-\hat{\mu}_2(\bar{c}))\right) 
			\\ & + 
			\sum_{t=1}^{T}\sigma_t^{-1}f_{t-1}(0) \left( (t-1)(1-\gamma_{0})(\mu_2-\hat{\mu}_2(\bar{c})) + (\mu_2-\hat{\mu}_2(\bar{c}))\right)^2  + o_p(1)
			\\ &:= S_1(v) + S_2(v) + S_3(v)+ S_4(v) + S_5 + o_p(1).
		\end{align*}
		For $S_1(v)$, from the proof in (a), it is easy to show
		$$S_1(v)\overset{d}{\rightarrow} v^{2}\left( f\left(  0\right)  \int_{0}^{1}\sigma^{-1}\left(r\right)  \mathcal{J}^2_{c,\sigma}
		\left(r\right)  dr\right). $$
		For $S_2(v)$ and $S_3(v)$, by (\ref{hat mu2- mu2}),  we obtain  
		\begin{align*}
			S_2(v)\overset{d}{\rightarrow} v^2\mathcal{M}^2(\bar{c},c,\sigma)f(0) \int_{0}^{1}\sigma^{-1}(r) r^2dr,
			\\
			S_3(v) \overset{d}{\rightarrow} -2v^2\mathcal{M}(\bar{c},c,\sigma)f(0)\int_{0}^{1}\sigma^{-1}(r)r\mathcal{J}_{c,\sigma}(r)dr. 
		\end{align*}
		For $S_4(v)$,
		\begin{align*}
			S_4(v) \overset{d}{\rightarrow}2vf(0)M(\bar{c},c,\sigma)\int_{0}^{1}\sigma^{-1}(r)(1+cr)\mathcal{J}_{c,\sigma}(r)dr - 2vf(0)\mathcal{M}^2(\bar{c},c,\sigma)\int_0^1\sigma^{-1}(r)r(1+cr)dr.
		\end{align*}
		At last, $S_5$ is irrelevant to $v$. As a result, by  combining the limits of $S_1(v), S_2(v), S_3(v)$ and $S_4(v)$, along with (\ref{LTH1}), and by minimizing (\ref{trof}) at $\hat{v} = T(\hat{\gamma}^{d}_{LAD} -\gamma_0)$, we obtain the limit of $L_T^d$ (hence the limit of $t_T^d$) in Theorem 5 (b).
	\end{proof}

	\subsection{Proof of Theorem 6}
	\begin{proof}[Proof of Theorem 6:]
		Denote $\acute{\varepsilon}_t = \hat{u}_t^d/\acute{\sigma}_t$, where $\hat{u}_t^d$ is the demeaned/detrended LAD residual and $\acute{\sigma}_t = \sum_{s=1}^T w_{t,s}|\hat{u}_t^d|$. Denote the bootstrap residual by $\acute{\varepsilon}_t^\ast$, and we still abbreviate $b_T$ as $b$.
		Following the proof of Lemma 4, we have
		\begin{align*}
			\frac{1}{\sqrt{T}}\sum_{t=1}^{T}\acute{\varepsilon}_{t}^{*} 
			= \frac{1}{\sqrt{T}}\sum_{m\in\mathcal{P}}\sum_{s=0}^{b-1}\acute{{\varepsilon}}_{i_{m}+s}-\frac{1}{\sqrt{T}}\sum_{m\in\mathcal{N}}\sum_{s=0}^{b-1}\acute{{\varepsilon}}_{i_{m}+s+T+1} + o_p(1):=B_1 + B_2 + o_p(1), 
		\end{align*}
		where $k=\lfloor(T-1)/b\rfloor$, $\mathcal{P}$ and $\mathcal{N}$ are defined in Subsection \ref{subsection proof of lemma 3}.
		
		Next, we only demonstrate the bootstrap invariance principle for the detrended residuals ( since the situation for the demeaned residuals is analogous). For the detrended residuals $\hat{u}^d_t$, 
		\begin{align}
			\hat{u}^d_t=  u_t + (\gamma_0 - \hat{\gamma}_{LAD}^d) y_{t-1} + (1-\hat{\gamma}_{LAD}^d)\left((\mu_1 - \hat{\mu}_1(\bar{c})) + (\mu_2 - \hat{\mu}_2(\bar{c}))(t-1) \right)+ (\mu_2 - \hat{\mu}_2(\bar{c})).
			\label{ut approxi}
		\end{align}
		By using (\ref{ut approxi}), $B_1$ and $B_2$ can be decomposed as follows
		\begin{align*}
			\frac{1}{\sqrt{T}}\sum_{m\in\mathcal{P}}\sum_{s=0}^{b-1}\acute{\varepsilon}_{i_m+s} & =  \frac{1}{\sqrt{T}}\sum_{m\in\mathcal{P}}\sum_{s=0}^{b-1}{\varepsilon}_{i_m+s} + \frac{1}{\sqrt{T}}\sum_{m\in\mathcal{P}}\sum_{s=0}^{b-1}{\varepsilon}_{i_m+s}\left(\frac{\sigma_{i_m+s}}{\acute{\sigma}_{i_m+s}} -1 \right) \\ &+ \left[\left( 1- \hat{\gamma}_{LAD}^d\right)\left(\mu_1 -\hat{\mu}_1(\bar{c})\right) + \left(\mu_2 -\hat{\mu}_2(\bar{c})\right) \right] \frac{ 1 }{\sqrt{T}}\sum_{m\in\mathcal{P}}\sum_{s=0}^{b-1}\frac{1}{\acute{\sigma}_{i_m+s}}
			\\ &+ \left( 1- \hat{\gamma}_{LAD}^d\right)\left(\mu_2 -\hat{\mu}_2(\bar{c})\right)\frac{1 }{\sqrt{T}}\sum_{m\in\mathcal{P}}\sum_{s=0}^{b-1}\frac{i_{m} +s -1 }{\acute{\sigma}_{i_m+s}}
			\\ & + 
			\left( \gamma_{0} - \hat{\gamma}_{LAD}^d \right)\frac{1 }{\sqrt{T}}\sum_{m\in\mathcal{P}}\sum_{s=0}^{b-1}\frac{y_{i_m+s-1}}{\acute{\sigma}_{i_m+s}} 
			\\& := P_{11} + B_{12} + \Theta_1B_{13} + \Theta_2B_{14} + \Theta_3B_{15},
		\end{align*}
		\begin{align*}
			-\frac{1}{\sqrt{T}}\sum_{m\in\mathcal{N}}\sum_{s=0}^{b-1}\acute{\varepsilon}_{i_m+s+T+1}& =  -\frac{1}{\sqrt{T}}\sum_{m\in\mathcal{N}}\sum_{s=0}^{b-1}{\varepsilon}_{i_m+s+T+1} - \frac{1}{\sqrt{T}}\sum_{m\in\mathcal{N}}\sum_{s=0}^{b-1}{\varepsilon}_{i_m+s+T+1}\left(\frac{\sigma_{i_m+s+T+1}}{\acute{\sigma}_{i_m+s+T+1}} -1 \right) \\ &- \left[\left( 1- \hat{\gamma}_{LAD}^d\right)\left(\mu_1 -\hat{\mu}_1(\bar{c})\right) + \left(\mu_2 -\hat{\mu}_2(\bar{c})\right)\right]\frac{ 1 }{\sqrt{T}}\sum_{m\in\mathcal{N}}\sum_{s=0}^{b-1}\frac{1}{\acute{\sigma}_{i_m+s+T+1}}
			\\ &- \left( 1- \hat{\gamma}_{LAD}^d\right)\left(\mu_2 -\hat{\mu}_2(\bar{c})\right)\frac{1 }{\sqrt{T}}\sum_{m\in\mathcal{N}}\sum_{s=0}^{b-1}\frac{i_m +s+T}{\acute{\sigma}_{i_m+s+T+1}}
			\\ & -  \left( \gamma_{0} - \hat{\gamma}_{LAD}^d \right)\frac{1 }{\sqrt{T}}\sum_{m\in\mathcal{N}}\sum_{s=0}^{b-1}\frac{y_{i_m+s+T}}{\acute{\sigma}_{i_m+s+T+1}}
			\\& := P_{21} + B_{22} + \Theta_1B_{23} + \Theta_2B_{24} + \Theta_3B_{25},
		\end{align*}
		where $\Theta_1 =  \left( 1- \hat{\gamma}_{LAD}^d\right)\left(\mu_1 -\hat{\mu}_1(\bar{c})\right) + \left(\mu_2 -\hat{\mu}_2(\bar{c})\right) = O_p(\frac{1}{\sqrt{T}})$, $\Theta_2 = \left( 1- \hat{\gamma}_{LAD}^d\right)\left(\mu_2 -\hat{\mu}_2(\bar{c})\right) = O_p(\frac{1}{T^{3/2}})$, $\Theta_3 =  \gamma_{0} - \hat{\gamma}_{LAD}^d = O_p(\frac{1}{T})$ since $\gamma_0 - \hat{\gamma}_{LAD}^d = O_p(\frac{1}{T})$, $ \mu_1 - \hat{\mu}_1(\bar{c}) = O_p(1)$, and $\mu_2 - \hat{\mu}_2(\bar{c}) = O_p(\frac{1}{\sqrt{T}})$ by Theorem 5, (\ref{hat mu1 - mu1}) and (\ref{hat mu2- mu2}), respectively.
		
		Consider the term $\Theta_1(B_{13}+B_{23})$. Since $E^*\left(\sum_{s=0}^{b-1}\acute{\sigma}_{i_m+s}^{-1} \right)^2 = \frac{1}{T-b}\sum_{t=1}^{T-b}\left(\sum_{s=0}^{b-1}\acute{\sigma}_{t+s}^{-1} \right)^2 = O_p(b^2)$ given $i_m \in \mathcal{P}$, it follows that $E^*(B_{13}+B_{23})^2 = O_p(b)$ by applying the same techniques as those for (\ref{squared of P12 + P22}). As a result, 
		$\Theta_1(B_{13}+B_{23}) = O_{p^*}\left(\sqrt{\frac{b}{T}} \right) = o_{p^*}(1)$ .
		
		Consider the term $\Theta_2(B_{14}+B_{24})$. Similarly, we have $E^*(B_{14}+B_{24})^2 = O_p(T^2b)$ because of  $E^*\left(\sum_{s = 0}^{b-1}\frac{i_m+s-1}{\acute{\sigma}_{i_m+s}} \right)^2 = \frac{1}{T-b}\sum_{t=1}^{T-b}\left(\sum_{s=0}^{b-1}\frac{t+s-1}{\acute{\sigma}_{t+s}}\right)^2 = O_p(T^2b^2)$. Consequently,   $\Theta_2(B_{13}+B_{23}) = O_{p^*}\left(\sqrt{\frac{b}{T}}\right)= o_{p^*}(1)$.

		Consider the term $\Theta_3(B_{15}+B_{25})$. Following (\ref{squared of P12 + P22}), it is easy to obtain $\Theta_3(B_{15}+B_{25}) = o_{p^*}(1)$.

		Finally, for $B_{12}$ and $B_{22}$, we can show they are both $o_p(1)$, whose proofs are similar to that of $\check{P}_{13}$ in Subsection \ref{section:proof lemma4}. As a result, 
		\begin{align}
			\frac{1}{\sqrt{T}}\sum_{t=1}^{T}\acute{\varepsilon}_{t}^{*} = \frac{1}{\sqrt{T}}\sum_{t=1}^{T}{\varepsilon}_{t}^{*}  + o_{p^*}(1).
			\label{DAB}
		\end{align}
		Similar to the proof of (\ref{eq:sgn sum decompose}), we  also have 
		\begin{align}
			\frac{1}{\sqrt{T}}\sum_{t=1}^{T}sgn\left( \acute{\varepsilon}_{t}^{*}\right)  = \frac{1}{\sqrt{T}}\sum_{t=1}^{T}sgn\left( {\varepsilon}_{t}^{*}\right)   + o_{p^*}(1) .
			\label{DABS}
		\end{align}
		Based on (\ref{DAB}) and (\ref{DABS}), we obtain
		\begin{align}
			\frac{1}{\sqrt{T}}\sum_{t=1}^{\lfloor Tr \rfloor}\left(  \acute{\varepsilon}_{t}^{*		},sgn(\acute{\varepsilon}_{t}^{*})\right)  ^{\prime}\overset{d^{*}}{\rightarrow}\left(  B_{1}(r),B_{2}(r)\right) ^{\prime}, \label{bootstrap bi-invariance homo}
		\end{align}
		by using the first result in Lemma 3. Similar to the proof of $\frac{1}{\sqrt{T}}\sum_{t=1}^{T}\left( \check{{\varepsilon }}_{t}^{\ast }\hat{\sigma}_{t} - \varepsilon_t^*\sigma _{t}\right)= o_{p^*}(1)$ in (\ref{DRR}), we also have $\frac{1}{\sqrt{T}}\sum_{t=1}^{T}\left(\acute{\sigma}_{t}\acute{\varepsilon	}_{t}^{*} - \varepsilon_t^*\sigma _{t}\right)= o_{p^*}(1)$, which, together with  (\ref{DABS}), leads to 
		\begin{align}
			\frac{1}{\sqrt{T}}\sum_{t=1}^{\lfloor Tr \rfloor}\left(  \acute{\sigma}_{t}\acute{\varepsilon	}_{t}^{*},sgn(\acute{\varepsilon}_{t}^{*})\right)^{\prime} \overset{d^{*}}{\rightarrow}\left(  B_{\sigma}(r),B_{2}(r)\right)^{\prime},
			\label{bootstrap bi-invariance heter}
		\end{align}
		by using the second result in Lemma 3.
		
		
		Now we turn to obtain the limiting distribution of $\check{L}_T^{d*}= T(\check{\gamma}_{LAD}^{d*}-1)$, where $\check{\gamma}_{LAD}^{d*}$ is the bootstrapped LAD estimator of $\gamma_0$ using $y_t^{d*} = x_t^* - \hat{\mu}^*(\bar{c})^\prime d_t$, where  $x_t^* = \sum_{i=1}^t \acute{u}_i^*$ and $\acute{u}_i^* = \acute{\sigma}_i\acute{\varepsilon}_i^*$.

		(a) Consider the case of $d_t = 1$. From the proof of Theorem 5, we have
		\begin{align}
			\hat{\mu}^*(\bar{c}) = \left( 1 + \frac{\bar{c}^2}{T}\right)^{-1}\left({x}_1^* + \frac{\bar{c}}{T}\sum_{t=2}^{T}{\Delta {x}_t^* + \frac{\bar{c}}{T}{x}_{t-1}^*} \right)  = {x}_1^* + O_p\left(\frac{1}{\sqrt{T}}\right).
			\label{acute mu ast mean}
		\end{align}
		Then the LAD minimization problem can be reformed as
		\begin{align*}
			\sum_{t=1}^{T}\left| {y}_t^{d*} - \gamma {y}_{t-1}^{d*} \right| = \sum_{t=1}^{T} \left| \acute{u}_t^* - vT^{-1}\left({x}_{t-1}^* - \hat{\mu}^*(\bar{c}) \right)  \right|,
		\end{align*}
		where $v = T(\gamma - 1)$. By using the result in (\ref{bootstrap bi-invariance heter}), the techniques in proving Theorem 5 (a), and $\hat{\mu}^*(\bar{c}) = O_p(1)$ in (\ref{acute mu ast mean}), we can easily show that 
		\begin{align*}
			& \frac{1}{T}\sum_{t=1}^{T}\left({x}_{t-1}^*  -\hat{\mu}^*(\bar{c}) \right)sgn\left( \acute{\varepsilon}_t^*\right)\overset{d^*}{\rightarrow} \int_0^1 B_\sigma(r)dB_2(r) + \Gamma,
			\\& 2\sum_{t=1}^{T}\int_{0}^{vT^{-1}\left( {x}_{t-1}^* -\hat{\mu}^*(\bar{c}) \right) }\left[
			I\left(  \acute{u}_{t}^*\leq s\right)  -I\left( \acute{u}_{t}^*\leq0\right)  \right]  ds \overset{d^*}{\rightarrow} v^2f(0)\int_{0}^{1}\sigma^{-1}(r)B_\sigma^2(r) dr,
		\end{align*} 
		which leads to Theorem 6 (a) immediately. 
		
		(b) Consider the case of $d_t = (1,t)^\prime$. By (\ref{D hat mu - mu}), it is straightforward that $\hat{\mu}_1^*(\bar{c}) = {x}_1^* + o_p(1)$ and 
		\begin{align*}
			\sqrt{T}\hat{\mu}^*_2(\bar{c}) = \left( \frac{1}{T}\sum_{t=2}^{T} \left(1+\frac{\bar{c}}{T}(t-1)\right)^2\right)^{-1} \frac{1}{\sqrt{T}}\sum_{t=2}^{T}  \left( 1+\frac{\bar{c}}{T}(t-1)\right)  \left(\Delta {x}_t^* + \frac{\bar{c}}{T}{x}_{t-1}^* \right).
		\end{align*}
		By  using (\ref{bootstrap bi-invariance heter}) and the arguments for proving Theorem 5 (b), we obtain $\sqrt{T}\hat{\mu}_2^*(\bar{c}) \overset{d^*}{\rightarrow} \mathcal{M}(\bar{c},0,\sigma)$. Next, rewrite
		\begin{align*}
			\sum_{t=1}^{T}\left|{y}_t^{d*} -\gamma {y}_{t-1}^{d^*} \right|
			=\sum_{t=1}^{T}\left|\acute{u}_t^*  -vT^{-1}\left( {x}_{t-1}^* - \hat{\mu}_1^*(\bar{c}) - \hat{\mu}_2^*(\bar{c}) (t-1) \right)  -\hat{\mu}_2^*(\bar{c}) \right|,
		\end{align*}
		where $v = T(\gamma - 1)$. Following the same logic of proving Theorem 5 (b), we have
		\begin{align*}
			& \frac{1}{T}\sum_{t=1}^{T}\left({x}_{t-1}^* - \hat{\mu}_1^*(\bar{c}) - \hat{\mu}_2^*(\bar{c}) (t-1)  \right)sgn\left( \acute{\varepsilon}_t^*\right)   \overset{d^*}{\rightarrow}  \int_{0}^{1}B_{\sigma}(r)dB_2(r) -\mathcal{M}(\bar{c},0,\sigma)\int_{0}^{1}rdB_2(r) +\Gamma,
		\end{align*}
		\begin{align*}
			&2\sum_{t=1}^{T}\int_{0}^{vT^{-1}\left( {x}_{t-1}^* - \hat{\mu}_1^*(\bar{c}) - \hat{\mu}_2^*(\bar{c}) (t-1) \right) +\hat{\mu}_2^*(\bar{c}) }\left[
			I\left(  \acute{u}_{t}^*\leq s\right)  -I\left(  \acute{u}_{t}^*\leq 0\right)  \right]  ds
			\\ &
			\overset{d^*}{\rightarrow}
			v^2f(0)\left(\int_{0}^{1}\sigma^{-1}(r)B_{\sigma}^2(r) dr - 2\mathcal{M}(\bar{c},0,\sigma) \int_{0}^{1}\sigma^{-1}(r)rB_\sigma(r)dr + \mathcal{M}^2(\bar{c},0,\sigma) \int_{0}^{1}\sigma^{-1}(r)r^2dr \right) 
			\\& + 2vf(0)\mathcal{M}(\bar{c},0,\sigma)\int_0^1\sigma^{-1}(r)B_\sigma(r)dr - 2vf(0)\mathcal{M}^2(\bar{c},0,\sigma)\int_0^1\sigma^{-1}(r)rdr + irre,
		\end{align*}
		where $irre$ represents the terms irrelevant to $v$. Consequently, Theorem 5 (b) is obtained by using the above results.
	\end{proof}
	
	\section{Technical Lemmas}\label{section other lemma}
	
	\begin{lemma}
		Suppose Assumptions 1-5 hold true, let $y_{t}$ be determined by (2.1)	and (2.2) with $\gamma_{0}$ satisfying (2.13),  define $\tilde{\sigma}_t = \sum_{i=1}^{T}w_{t,i}|u_i|$ and $\check{\sigma}_t = \sum_{i=1}^{T}w_{t,i}\sigma_i$, as $T\rightarrow\infty$, we have 
		
		(i) $max_{1\leq t,i\leq T} w_{t,i} = O(\frac{1}{Th})$;
		
		(ii) $max_{1\leq t \leq T} \hat{\sigma}_t = O_p(1)$ and $\left(min_{1\leq t \leq T} \hat{\sigma}_t\right)^{-1} = O_p(1)$;
		
		(iii) $max_{1\leq t \leq T} \left\vert\hat{\sigma}_t - \tilde{\sigma}_t \right\vert  = O_p(\frac{1}{\sqrt{T}h})$;
		
		(iv) $\sum_{t=1}^T \left(\hat{\sigma}_t - \tilde{\sigma}_t\right)^2 = O_p(\frac{1}{Th^2})$;
		
		(v) $max_{1\leq t \leq T} E\left\vert \tilde{\sigma}_t - \check{\sigma}_t \right\vert^\delta = O(\frac{1}{T^{\delta/2}h^{\textcolor{black}{\delta}}}),\delta=2,4$;
		
		(vi) $max_{1\leq t\leq T}\left\vert \tilde{\sigma}_{t}-\check{\sigma}_{t}\right\vert=O_{p}\left( \frac{1}{T^{1/4}h^{1/2}}\right)$;
		
		(vii) $\frac{1}{T}\sum_{t=1}^T \left\vert\check{\sigma}_t -\sigma_t \right\vert^{\delta} = o(1),\delta=1,2,4$;
		
		(viii) $max_{t}|\check{\sigma}_t -\sigma_t| = O(1)$.
		\label{lemma new other}
	\end{lemma}

	\begin{lemma}
		Suppose Assumptions 1-5 hold true, let $y_{t}$ be determined by (2.1)	and (2.2) with $\gamma_{0}$ satisfying (2.13), as $T\rightarrow\infty$, we have 
		
		(i) $\check{Q}_1 = \frac{1}{\sqrt{T}}\sum_{t=1}%
		^{T}\varepsilon_{t} \sigma^{-1}_{t}\left(  \sigma_{t}-\check{\sigma}_{t}\right) = o_p(1) $;
		
		(ii) $\check{Q}_2 = \frac
		{1}{\sqrt{T}}\sum_{t=1}^{T}\varepsilon_{t}\sigma^{-1}_{t}\left(  \check{\sigma}_{t}%
		-\tilde{\sigma}_{t}\right)  = o_p(1)$;
		
		(iii) $\check{Q}_3 =\frac{1}{\sqrt{T}}\sum_{t=1}^{T}%
		\varepsilon_{t}\sigma^{-1}_{t}\left(  \tilde{\sigma}_{t}-\hat{\sigma}_{t}\right) =o_p(1) $.
		
		\label{Lemma other}
	\end{lemma}

	\begin{lemma}
		Suppose Assumptions 1-5 hold true, let $y_{t}$ be determined by (2.1)	and (2.2) with $\gamma_{0}$ satisfying (2.13), as  $T\rightarrow\infty$, we have 
		$\frac{1}{\sqrt{T}}\sum_{t=1}^{T}\varepsilon_{t}{\left( \sigma _{t}-%
			\hat{\sigma}_{t}\right) ^{2}}/{\left(\sigma _{t}\hat{\sigma}_{t}\right)} = o_p(1)$.
		
		
		\label{Lemma other new2}
	\end{lemma}

	\begin{proof}[Proof of Lemma \ref{lemma new other}:] The proof of (i), (ii), (vi), (vii) and (viii) are similar to those of Lemma A in  \cite{Xu2008}. The proof of (v) is similar to that of Lemma B5 in \cite{Zhu2019}. 
		
		(iii) By using (i), we have $|\hat{\sigma}_t - \tilde{\sigma}_t| \leq O({1}/{Th})\sum_{i=1}^T\left\vert |\hat{u}_i| - |u_i| \right\vert$, thus it follows that 
		\begin{align*}
			\mathop{max}_{1\leq t \leq T} \left\vert \hat{\sigma}_t - \tilde{\sigma}_t \right\vert \leq O(\frac{1}{Th}) \sum_{i=1}^T\left\vert |\hat{u}_i| - |u_i| \right\vert
			\leq O(\frac{1}{Th})\left\vert \hat{\gamma}_{LAD} - \gamma_0\right\vert\sum_{i=1}^T |y_{i-1}| = O_p(\frac{1}{\sqrt{T}h}), 
		\end{align*}
		where $ \hat{\gamma}_{LAD} - \gamma_0 = O_p(T^{-1})$ by Theorem 2 and $|T^{-1/2}y_{i-1}| = O_p(1)$ by (2.14). 
		
		(iv) Let $\hat{v} = T(\hat{\gamma}_{LAD}-\gamma_0)$,  we make the following decomposition:
		\begin{align*}
			\hat{\sigma}_{t}-\tilde{\sigma}_{t} &= \sum_{i=1}^T w_{ti}\left( |\hat{u}_i| - |u_i| \right) 
			= \sum_{i=1}^T w_{ti}\left(|u_i - \hat{v}y_{i-1}T^{-1}| - |{u}_i| \right)
			\\ &
			=  -T^{-1}\sum_{i=1}^T w_{ti} y_{i-1}sgn(u_i)  + 2\sum_{i=1}^{T}w_{ti}\int_0^{\hat{v}y_{i-1}T^{-1}}(I(u_i\leq s) - I(u_i \leq 0)) ds,
		\end{align*}
		where the third equality holds by Knight's identity (\cite{Knight1989}). Since $\{y_t\}$ is a (nearly) unit root process, it is easy to derive that 
		$T^{-1}\sum_{i=1}^T w_{ti} y_{i-1}sgn(u_i) = O_p(\frac{1}{Th})$ by (\ref{yt-1 sgn alter}) and result (i). For the second term on the right-hand side, similarly, we proceed by defining  $I_x(v)$ be a regular sequence of indicator function $I(v \leq 0)$ and thus have
		\begin{align*}
			\sum_{i=1}^{T}w_{ti}\int_0^{\hat{v}y_{i-1}T^{-1}}&(I_x(u_i -s) - I_x(u_{i})) ds = \sum_{i=1}^{T}w_{ti}\int_0^{\hat{v}y_{i-1}T^{-1}}sI^\prime_x(u_{*i}) ds
			\\ &
			\leq  C\sum_{i=1}^{T}w_{ti}\int_0^{\hat{v}y_{i-1}T^{-1}} s ds
			=  C\hat{v}^2 T^{-2} \sum_{i=1}^{T}w_{ti}\frac{1}{2} y_{i-1}^2 
			= O_p(\frac{1}{Th}),
		\end{align*}
		where $u_{*i}$ is a middle value between $u_i$ and $u_i -s$,  $I_x^\prime(\cdot)$ is the first derivative of $I_x(\cdot)$, and is  bounded for any finite $x$ (see \cite{Phillips1995}). As such, we obtain $\sum_{t=1}^{T}\left( \hat{\sigma}_{t}-\tilde{\sigma}_{t} \right)^2 = O_p(\frac{1}{Th^2}) $.
	\end{proof}
	
	\bigskip

	\begin{proof}[Proof of Lemma \ref{Lemma other}:]	 
		We first consider $\check{Q}_{1}$ and $\check{Q}_{3}$. Let ${N=\{1,2,\cdots,T\}}$, and let $N_{0}$ denote the set of discontinuous points of $\sigma (\cdot)$. For $\check{Q}_{1}$, it has the following decomposition
		\begin{equation*}
			\check{Q}_{1}=\frac{1}{\sqrt{T}}\sum_{t\in N\backslash N_{0}}^{T}\varepsilon
			_{t}\sigma^{-1}_{t}\left( \sigma _{t}-\check{\sigma}_{t}\right) +\frac{1}{\sqrt{T}}%
			\sum_{t\in N_{0}}^{T}\varepsilon _{t}\sigma^{-1}_{t}\left( \sigma _{t}-\check{\sigma}_{t}\right) :=\check{Q}_{11}+\check{Q}_{12}.
		\end{equation*}
		For $\check{Q}_{11}$ we have
		\begin{equation*}
			E\left( \check{Q}_{11}^{2}\right) \leq \left(\min_{t}{\sigma}_{t}\right)^{-2}\left( \max_{t\in N\backslash
				N_{0}}\left\vert \sigma _{t}-\check{\sigma}_{t}\right\vert \right) ^{2}\frac{%
				1}{T}\sum_{t\in N\backslash N_{0}}^{T}\sum_{t^{\prime }\in N\backslash
				N_{0}}^{T}\left\vert E\left( \varepsilon _{t}\varepsilon _{t^{\prime
			}}\right) \right\vert =o_{p}\left( 1\right),
		\end{equation*}%
		where $\max_{t\in N\backslash N_{0}}\left\vert \sigma _{t}-\check{\sigma}%
		_{t}\right\vert =o\left( 1\right) $ for continuous points, and $\frac{1}{T}\sum_{t\in
			N\backslash N_{0}}^{T}\sum_{t^{\prime }\in N\backslash N_{0}}^{T}\left\vert E\left( \varepsilon
		_{t}\varepsilon _{t^{\prime }}\right) \right\vert =O\left( 1\right) $ by the mixing condition in Assumption 2 and Corollary 15.3 of \cite{Davidson1994}. Similarly, for $\check{Q}_{12}$		
		\begin{equation*}
			E\left( \check{Q}_{12}^{2}\right) \leq \left(\min_{t}{\sigma}_{t}\right)^{-2}\left( \max_{t\in N_{0}}\left\vert
			\sigma_{t}-\check{\sigma}_{t}\right\vert \right) ^{2}\frac{1}{T}\sum_{t\in N_{0}}^{T}\sum_{t^{\prime }\in N_{0}}^{T}\left\vert E\left( \varepsilon
			_{t}\varepsilon _{t^{\prime }}\right) \right\vert = o_p\left( 1\right),
		\end{equation*}%
		where $max_{t\in  N_{0}}\left\vert \sigma _{t}-\check{\sigma}%
		_{t}\right\vert =O\left( 1\right) $ and $\frac{1}{T}\sum_{t\in
			N_{0}}^{T}\sum_{t^{\prime }\in N_{0}}^{T}\left\vert E\left( \varepsilon
		_{t}\varepsilon _{t^{\prime }}\right) \right\vert =O\left( h\right) $  since the cardinality of $\sum_{t\in N_{0}}^{T}$ is $O(Th)$. As a result, we can conclude that $\check{Q}_{1}=o_{p}\left( 1\right) .$
		
		For $\check{Q}_{3}$, it is easy to obtain
		\begin{equation*}
			\left\vert \check{Q}_{3}\right\vert \leq \left(\min_{t}{\sigma}_{t}\right)^{-1}\left( \frac{1}{T}%
			\sum_{t=1}^{T}\varepsilon _{t}^{2}\right) ^{1/2}\left( \sum_{t=1}^{T}\left(
			\tilde{\sigma}_{t}-\hat{\sigma}_{t}\right) ^{2}\right) ^{1/2}=O_{p}\left(
			\frac{1}{\sqrt{Th^{2}}}\right),
		\end{equation*}
		where $\frac{1}{T}\sum_{t=1}^{T}\varepsilon _{t}^{2}=O_{p}\left( 1\right) $ and $\sum_{t=1}^{T}\left( \tilde{\sigma}_{t}-\hat{\sigma}_{t}\right)
		^{2}=O_{p}\left( \frac{1}{Th^{2}}\right) $ by Lemma \ref{lemma new other} (iv). Hence $\check{Q}_{3}=o_{p}\left( 1\right).$
		
		Now we show that $\check{Q}_{2}=o_p(1)$. Define $\varphi_{Tt} =\sigma^{-1}_{t} \sum_{i=1}^{t-1}w_{t,i}\sigma_i(|\varepsilon_i|-1)$ and $\varpi_{Tt} = \sigma^{-1}_{t}\sum_{i=t+1}^{T}w_{t,i}\sigma_i(|\varepsilon_i|-1)$. Then $\check{Q}_{2}$ can be rewritten as
		\begin{align*}
			&\check{Q}_{2}=
			\frac{1}{\sqrt{T}}\sum_{t=1}^{T}\varepsilon_t\varphi_{Tt}+ \frac{1}{\sqrt{T}}\sum_{t=1}^{T}\varepsilon_t\varpi_{Tt}+  \frac{1}{\sqrt{T}}\sum_{t=1}^{T}\varepsilon_tw_{t,t}(|\varepsilon_t|-1):=\check{Q}_{21}+\check{Q}_{22}+\check{Q}_{23}.
		\end{align*}
		We first consider $\check{Q}_{23}$, it follows that 
		\begin{eqnarray*}
			E\left\vert \check{Q}_{23}\right\vert  \leq \frac{1}{\sqrt{T}}%
			\sum_{t=1}^{T}w_{t,t}E\left\vert \varepsilon _{t}(|\varepsilon
			_{t}|-1)\right\vert 
			\leq \frac{1}{\sqrt{T}}\left( \max_{t}w_{t,t}\right)\sum_{t=1}^{T}\left[ E\left( \varepsilon _{t}^{2}\right)
			E(|\varepsilon _{t}|-1)^{2}\right] ^{1/2}  =O_{p}(\frac{1}{\sqrt{T}h}).
		\end{eqnarray*}
		Hence we have $\check{Q}_{23}=o_p(1)$ by Markov's inequality. Next we show that  $\check{Q}_{21}$ and $\check{Q}_{22}$ are both $o_p(1)$.  Because $\check{Q}_{21}$ and $\check{Q}_{22}$ are similar, we only prove  $\check{Q}_{22}$ here. 
		For  $\check{Q}_{22}$ we have
		\begin{align}
			\check{Q}_{22}^2 = \frac{1}{T}\sum_{t=1}^{T}\sum_{t^\prime=1}^{T}\varepsilon_t\varepsilon_{t^\prime}\varpi_{Tt}\varpi_{Tt^\prime} :=\frac{1}{T}\Pi_1 +  \frac{1}{T}\Pi_2 + \frac{1}{T}\Pi_3, 
			\label{Pi decompose}
		\end{align}
		where $\Pi_1 =\sum_{t=1}^{T}\sum_{t^\prime > t}\varepsilon_t\varepsilon_{t^\prime}\varpi_{Tt}\varpi_{Tt^\prime} $, $\Pi_2 =\sum_{t=1}^{T}\sum_{t^\prime = t}\varepsilon_t\varepsilon_{t^\prime}\varpi_{Tt}\varpi_{Tt^\prime} $ and $\Pi_3 =\sum_{t^\prime=1}^{T}\sum_{t> t^\prime}\varepsilon_t\varepsilon_{t^\prime}\varpi_{Tt}\varpi_{Tt^\prime} $. For $\Pi_1$, choose a positive integer $l_T$ such that $l_T^{-1} = o(1)$, then 
		\begin{align}
			\Pi_1 & = \sum_{t=1}^{T}\sum_{t^\prime > t}\varepsilon_t\varepsilon_{t^\prime}\left( \varpi_{1Tt}+\varpi_{2Tt}\right) \left(  \varpi_{1Tt^\prime} +  \varpi_{2Tt^\prime} \right) \nonumber
			\\ & 
			= \sum_{t=1}^{T}\sum_{t^\prime > t}\varepsilon_t\varepsilon_{t^\prime} \varpi_{1Tt}  \varpi_{1Tt^\prime}   + \sum_{t=1}^{T}\sum_{t^\prime > t}\varepsilon_t\varepsilon_{t^\prime} \varpi_{1Tt}   \varpi_{2Tt^\prime}  \nonumber
			\\ & + 
			\sum_{t=1}^{T}\sum_{t^\prime > t}\varepsilon_t\varepsilon_{t^\prime}\varpi_{2Tt} \varpi_{1Tt^\prime} + \sum_{t=1}^{T}\sum_{t^\prime > t}\varepsilon_t\varepsilon_{t^\prime}\varpi_{2Tt} \varpi_{2Tt^\prime} \nonumber
			\\ &
			:= \Pi_{11} + \Pi_{12} + \Pi_{13} + \Pi_{14},
			\label{Pi1 decompose}
		\end{align}
		where $\varpi_{1Tt} = \sigma^{-1}_{t}\sum_{i=t+1}^{t^\prime -1}w_{t,i}\sigma_i(|\varepsilon_i|-1)$, $\varpi_{2Tt} = \sigma^{-1}_{t}\sum_{i=t^\prime}^{T}w_{t,i}\sigma_i(|\varepsilon_i|-1)$,  $\varpi_{1Tt^\prime} = \sigma^{-1}_{t}\sum_{i=t^\prime+1}^{t^\prime+l_T-1}w_{t^\prime,i}\sigma_i(|\varepsilon_i|-1)$ and $\varpi_{2Tt^\prime} =\sigma^{-1}_{t} \sum_{i=t^\prime+l_T}^{T}w_{t^\prime,i}\sigma_i(|\varepsilon_i|-1)$ .  
		
		Now we show $T^{-1}\Pi_{11},T^{-1}\Pi_{12},T^{-1}\Pi_{13}$ and $T^{-1}\Pi_{14}$ are all $o_p(1)$. Firstly, for $T^{-1}\Pi_{11}$, 
		\begin{align}
			\left|E\left(\frac{1}{T}\Pi_{11} \right)  \right|	& = \left| \frac{1}{T}\sum_{t=1}^{T}\sum_{t^\prime > t}E\left(\varepsilon_t\varepsilon_{t^\prime} \varpi_{1Tt}  \varpi_{1Tt^\prime} \right)  \right|
			\leq 
			\frac{1}{T}\sum_{t=1}^{T}\sum_{t^\prime>t}||\varepsilon_t||_4 || \varepsilon_{t^\prime}||_4 
			||\varpi_{1Tt}||_4 ||\varpi_{1Tt^\prime}||_4 \nonumber
			\\ & 
			\leq 
			O(\frac{1}{T^3h^2})\sum_{t=1}^{T}\sum_{t^\prime> t}\sqrt{(t^\prime-t-1)l_T} = O(\frac{\sqrt{l_T}}{\sqrt{T}h^2}), 
			\label{Pi 11}
		\end{align}
		where the first inequality holds by H$\ddot{\mathrm{o}}$lder's  inequality, and 
		\begin{align}
			&||\varpi_{1Tt}||_4 \leq C \sqrt{t^\prime - t-1}\left(\min_{t}{\sigma}_{t}\right)^{-1} \mathop{max}_{i}||w_{t,i}\sigma_i\left(|\varepsilon_i|-1 \right) ||_k = O(\frac{\sqrt{t^\prime - t-1}}{Th}), \label{shao yu theorem}
			\\ &
			||\varpi_{1Tt^\prime}||_4 \leq C \sqrt{l_T}\left(\min_{t}{\sigma}_{t}\right)^{-1} \mathop{max}_{i}||w_{t,i}\sigma_i\left(|\varepsilon_i|-1 \right) ||_k = O(\frac{\sqrt{l_T}}{Th}),
		\end{align}
		which are obtained by applying Theorem 4.1 of \cite{Shao1996}.

		Secondly, for $T^{-1}\Pi_{12}$, since $E(\varpi_{2Tt^\prime}) = 0$, we have
		\begin{align}
			\left|E\left(\frac{1}{T}\Pi_{12} \right)  \right|& = 
			\left| \frac{1}{T}\sum_{t=1}^{T}\sum_{t^\prime > t}Cov\left(\varepsilon_t\varepsilon_{t^\prime} \varpi_{1Tt},  \varpi_{2Tt^\prime} \right)  \right| \nonumber
			\leq 
			\frac{1}{T}\sum_{t=1}^{T}\sum_{t^\prime > t}||\varepsilon_t\varepsilon_{t^\prime} \varpi_{1Tt}||_{p/3}  ||\varpi_{2Tt^\prime}||_{p}  \alpha_{l_T}^{\frac{p-4}{p}} \nonumber
			\\ &
			\leq  
			\frac{1}{T}\sum_{t=1}^{T}\sum_{t^\prime > t}||\varepsilon_t||_p ||\varepsilon_{t^\prime}||_p ||\varpi_{1Tt}||_{p}  ||\varpi_{2Tt^\prime}||_{p}  \alpha_{l_T}^{\frac{p-4}{p}} \nonumber
			\leq 
			O(\frac{1}{T^3h^2})\sum_{t=1}^{T}\sum_{t^\prime > t} \sqrt{(t^\prime-t-1)(T-t^\prime-l_T)} \alpha_{l_T}^{\frac{p-4}{p}} \nonumber
			\\ & =
			O\left( \left( \frac{1}{l_T}\right) ^{4 + \frac{p-4}{p}\left(\delta_0-\frac{4p}{p-4}  \right)  }\right) O(h^{-2}),
			\label{Pi 12}
		\end{align}
		where the first inequality holds by Davydov's inequality in \cite{Davydov1968}, the second inequality holds by H$\ddot{\mathrm{o}}$lder's  inequality, and third inequality follows by the similar arguments as for (\ref{shao yu theorem}), the last equality holds for some $\delta_0 \geq p\beta/(p-\beta)$ by mixing condition in Assumption 6. For notation simplicity, denote $\iota_0 = (p-4)p^{-1}\left(\delta_0 - 4p(p-4)^{-1} \right) > 0$ because of  $p>\beta > 4$. As such, we have $\left|E\left({T}^{-1}\Pi_{12} \right)  \right|\leq O\left(\frac{1}{l_T^{4+\iota}h^2}\right)$. 
		
		Thirdly, for $T^{-1}\Pi_{13}$, we have
		\begin{align*}
			\left|E\left(\frac{1}{T}\Pi_{13} \right)  \right|	&= 
			\left| \frac{1}{T}\sum_{t=1}^{T}\sum_{t^\prime > t}E\left(\varepsilon_t\varepsilon_{t^\prime} \left( \varpi_{21Tt} + \varpi_{22Tt}\right)   \varpi_{1Tt^\prime} \right)  \right|
			\\ &
			= \left| \frac{1}{T}\sum_{t=1}^{T}\sum_{t^\prime > t}
			\left[E\left(\varepsilon_t\varepsilon_{t^\prime}  \varpi_{21Tt}    \varpi_{1Tt^\prime} \right) + Cov\left(\varepsilon_t\varepsilon_{t^\prime}   \varpi_{1Tt^\prime},\varpi_{22Tt} \right)\right] \right|,
		\end{align*}
		where $\varpi_{21Tt} =\sigma^{-1}_{t} \sum_{i = t^\prime}^{t^\prime + 2 l_T -3}w_{t,i}\sigma_i(|\varepsilon_i|-1)$ and $\varpi_{22Tt} = \sigma^{-1}_{t}\sum_{i = t^\prime + 2 l_T -2}^{T}w_{t,i}\sigma_i(|\varepsilon_i|-1)$. Following the similar arguments as those for (\ref{Pi 11}) and (\ref{Pi 12}), we have
		\begin{align*}
			\left| \frac{1}{T}\sum_{t=1}^{T}\sum_{t^\prime > t}
			E\left(\varepsilon_t\varepsilon_{t^\prime}  \varpi_{21Tt}    \varpi_{1Tt^\prime} \right)\right|  & \leq \frac{1}{T}\sum_{t=1}^{T}\sum_{t^\prime>t}||\varepsilon_t||_4 || \varepsilon_{t^\prime}||_4 
			||\varpi_{21Tt}||_4 ||\varpi_{1Tt^\prime}||_4 
			\\ & 
			\leq 
			O(\frac{1}{T^3h^2})\sum_{t=1}^{T}\sum_{t^\prime> t}\sqrt{(2l_T -3 )(l_T-1)} = O(\frac{{l_T}}{{T}h^2}),
		\end{align*}
		and 
		\begin{align*}
			\left| \frac{1}{T}\sum_{t=1}^{T}\sum_{t^\prime > t}
			Cov\left(\varepsilon_t\varepsilon_{t^\prime}   \varpi_{1Tt^\prime},\varpi_{22Tt} \right)\right|& \leq  \frac{1}{T}\sum_{t=1}^{T}\sum_{t^\prime > t}||\varepsilon_t||_p ||\varepsilon_{t^\prime}||_p ||\varpi_{1Tt^\prime}||_{p}  ||\varpi_{22Tt}||_{p}  \alpha_{l_T}^{\frac{p-4}{p}}
			\\&
			\leq  
			O(\frac{1}{T^3h^2})\sum_{t=1}^{T}\sum_{t^\prime > t} \sqrt{(l_T-1)(T-t^\prime-2l_T-1)}\alpha_{l_T}^{\frac{p-4}{p}}
			\\ &
			\leq O(\frac{\sqrt{l_T}\alpha_{l_T}^{\frac{p-4}{p}}}{\sqrt{T}h^2}) = O(\frac{1}{\sqrt{T}l_T^{7/2+\iota_0} h^2}).
		\end{align*}
		Hence we have that 
		\begin{align}
			\left|E\left(\frac{1}{T}\Pi_{13} \right) \right| = O(\frac{{l_T}}{{T}h^2}) + O(\frac{1}{\sqrt{T}l_T^{7/2 +\iota_0} h^2}).
			\label{Pi 13}
		\end{align}
		
		Lastly, for $T^{-1}\Pi_{14} $, it is easy to derive that
		\begin{align*}
			\left|E\left(\frac{1}{T}\Pi_{14} \right)  \right|&
			= \left| \frac{1}{T}\sum_{t=1}^{T}\sum_{t^\prime > t}E\left(\varepsilon_t\varepsilon_{t^\prime} \left( \varpi_{23Tt} + \varpi_{24Tt}\right)   \varpi_{2Tt^\prime} \right)  \right|
			\\ &
			= \left| \frac{1}{T}\sum_{t=1}^{T}\sum_{t^\prime > t}
			\left[Cov\left(\varepsilon_t\varepsilon_{t^\prime}  \varpi_{23Tt},    \varpi_{2Tt^\prime} \right) + Cov\left(\varepsilon_t\varepsilon_{t^\prime},   \varpi_{24Tt}\varpi_{2Tt^\prime} \right) + E\left(\varepsilon_t\varepsilon_{t^\prime} \right) E\left(\varpi_{24Tt}\varpi_{2Tt^\prime} \right) \right] \right|,
		\end{align*}
		where  $\varpi_{23Tt} = \sigma^{-1}_{t}\sum_{i = t^\prime}^{t^\prime + l_T/2 -1}w_{t,i}\sigma_i(|\varepsilon_i|-1)$ and $\varpi_{24Tt} = \sigma^{-1}_{t}\sum_{i = t^\prime + l_T/2}^{T}w_{t,i}\sigma_i(|\varepsilon_i|-1)$. For the last term on the right-hand side in the second equality, by using Corollary 15.3 in \cite{Davidson1994}, it can be shown that  $\left| \frac{1}{T}\sum_{t=1}^{T}\sum_{t^\prime > t}
		E\left(\varepsilon_t\varepsilon_{t^\prime} \right) E\left(\varpi_{24Tt}\varpi_{2Tt^\prime} \right)\right|= O(\frac{1}{Th^2})$. 
		By the similar techniques as those for (\ref{Pi 12}), we immediately obtain  $\left| \frac{1}{T}\sum_{t=1}^{T}\sum_{t^\prime > t}
		Cov\left(\varepsilon_t\varepsilon_{t^\prime}   \varpi_{23Tt},\varpi_{2Tt^\prime} \right)\right| = O(\frac{1}{\sqrt{T}l_T^{7/2+\iota_0} h^2})$ and $\left| \frac{1}{T}\sum_{t=1}^{T}\sum_{t^\prime > t}
		Cov\left(\varepsilon_t\varepsilon_{t^\prime},   \varpi_{24Tt}\varpi_{2Tt^\prime} \right)\right| = O(\frac{1}{l_T^{4+\iota_0} h^2}).$ As a result,  
		\begin{align}
			\left|E\left(\frac{1}{T}\Pi_{14} \right) \right| =  O(\frac{1}{\sqrt{T}l_T^{7/2+\iota_0} h^2})+O(\frac{1}{{l}_T^{4+\iota_0}h^2}) + O(\frac{1}{Th^2}).
			\label{Pi 14}
		\end{align}
		
		Take $l_T = h^{-1/2}$, by (\ref{Pi 11}), (\ref{Pi 12}), (\ref{Pi 13}), (\ref{Pi 14}) and Markov's inequality, we conclude that
		$$
		\frac{1}{T}\Pi_1 = O_p\left( \frac{1}{\sqrt{T}h^{{9}/{4}}} + \frac{1}{Th^{5/2}}+h^{\iota_0/2} + \frac{1}{Th^2} + \frac{1}{\sqrt{T}h^{1/4-\iota_0/2}} \right)  = o_p(1),
		$$ 
		as $Th^{9/2}\rightarrow \infty$. In a similar manner, we can also obtain $T^{-1}\Pi_2=T^{-1}\Pi_3=o_p$(1). Consequently, we get $\check{Q}_{22}=o_p(1)$. Furthermore, it can be shown that $\check{Q}_{21}=o_p(1)$ by taking the same argument to those of proving $\check{Q}_{22}$. By combining the results $\check{Q}_{21}=o_p(1)$, $\check{Q}_{22}=o_p(1)$ and $\check{Q}_{23}=o_p(1)$, it follows that $\check{Q}_{2}=o_p(1)$. This completes the proof. 
	\end{proof}

	\bigskip
	\begin{proof}[Proof of Lemma \ref{Lemma other new2}:]  Observe that
		\begin{align*}
			\frac{1}{\sqrt{T}}\sum_{t=1}^{T}\varepsilon _{t}\frac{\left( \sigma _{t}-%
				\hat{\sigma}_{t}\right) ^{2}}{\sigma _{t}\hat{\sigma}_{t}}& =\frac{1}{\sqrt{T%
			}}\sum_{t=1}^{T}\varepsilon _{t}\frac{\left( \tilde{\sigma}_{t}-\hat{\sigma}%
				_{t}\right) ^{2}}{\sigma _{t}\hat{\sigma}_{t}}+\frac{1}{\sqrt{T}}%
			\sum_{t=1}^{T}\varepsilon _{t}\frac{\left( \check{\sigma}_{t}-\tilde{\sigma}%
				_{t}\right) ^{2}}{\sigma _{t}\hat{\sigma}_{t}} \\
			& +\frac{1}{\sqrt{T}}\sum_{t=1}^{T}\varepsilon _{t}\frac{\left( \sigma _{t}-%
				\check{\sigma}_{t}\right) ^{2}}{\sigma _{t}\hat{\sigma}_{t}}+\frac{2}{\sqrt{T%
			}}\sum_{t=1}^{T}\varepsilon _{t}\frac{\left( \tilde{\sigma}_{t}-\hat{\sigma}%
				_{t}\right) \left( \check{\sigma}_{t}-\tilde{\sigma}_{t}\right) }{\sigma _{t}%
				\hat{\sigma}_{t}} \\
			& +\frac{2}{\sqrt{T}}\sum_{t=1}^{T}\varepsilon _{t}\frac{\left( \tilde{\sigma%
				}_{t}-\hat{\sigma}_{t}\right) \left( \sigma _{t}-\check{\sigma}_{t}\right) }{%
				\sigma _{t}\hat{\sigma}_{t}}+\frac{2}{\sqrt{T}}\sum_{t=1}^{T}\varepsilon _{t}%
			\frac{\left( \check{\sigma}_{t}-\tilde{\sigma}_{t}\right) \left( \sigma _{t}-%
				\check{\sigma}_{t}\right) }{\sigma _{t}\hat{\sigma}_{t}} \\
			& =\mathcal{V}_{1}+\mathcal{V}_{2}+\mathcal{V}_{3}+2\mathcal{V}_{4}+2%
			\mathcal{V}_{5}+2\mathcal{V}_{6}\text{.}
		\end{align*}
		
		Now we demonstrate that $\mathcal{V}_{i}$, for $i=1,\ldots,6$, are all $o_p(1)$.
		For $\mathcal{V}_{1}$, it follows that
		\begin{equation*}
			\left\vert \mathcal{V}_{1}\right\vert \leq \left( \min_{t}\sigma
			_{t}\right) ^{-1}\left( \min_{t}\hat{\sigma}_{t}\right) ^{-1}\left(
			\max_{t}\left\vert \tilde{\sigma}_{t}-\hat{\sigma}_{t}\right\vert \right)
			^{2}\frac{1}{\sqrt{T}}\sum_{t=1}^{T}\left\vert \varepsilon _{t}\right\vert
			=O_{p}\left( \frac{1}{\sqrt{T}h^{2}}\right),
		\end{equation*}%
		where $\left( \min_{t}{\sigma}_{t}\right)^{-1} = O(1)$ by Assumption 3, $\left( \min_{t}\hat{\sigma}_{t}\right)^{-1}  =O_p(1)$ and $\max_{t}\left\vert \tilde{\sigma}_{t}-\hat{\sigma}%
		_{t}\right\vert =O_{p}\left(\frac{1}{\sqrt{T}h} \right) $ by Lemma \ref{lemma new other} (ii) and (iii), and  $\frac{1}{T}\sum_{t=1}^{T}\left\vert \varepsilon
		_{t}\right\vert =O_{p}\left( 1\right) $.
		
		For $\mathcal{V}_{2}$,  it follows that $\mathcal{V}_{2} \leq \left( \min_{t}\sigma_{t}\right) ^{-1}\left( \min_{t}\hat{\sigma}_{t}\right) ^{-1} \frac{1}{\sqrt{T}}\sum_{t=1}^{T}\left\vert\varepsilon_t\right\vert\left(
		\check{\sigma}_{t}-\tilde{\sigma}_{t}\right)^2$, and 
		\begin{equation*}
			E\left(\frac{1}{\sqrt{T}} \sum_{t=1}^{T}\left\vert\varepsilon_t\right\vert\left(\check{\sigma}_{t}-\tilde{\sigma}_{t}\right)^2\right) \leq \frac{1}{\sqrt{T}}		\sum_{t=1}^{T}\left( E\left( \varepsilon _{t}^{2}\right) \max_{t}E\left(
			\check{\sigma}_{t}-\tilde{\sigma}_{t}\right) ^{4}\right) ^{1/2}=O\left(
			\frac{1}{\sqrt{T}h^{2}}\right),
		\end{equation*}%
		where $\max_{t}E\left( \check{\sigma}_{t}-\tilde{\sigma}_{t}\right) ^{4}=O\left(\frac{1}{T^2h^4}
		\right)$ by Lemma \ref{lemma new other} (v). 
		
		For $\mathcal{V}_{3}$, we have
		\begin{eqnarray*}
			E\left( \mathcal{V}_{3}^{2}\right)  &\leq &\frac{1}{T}\sum_{t=1}^{T}\frac{
				E\left( \varepsilon _{t}^{2}\right) \left( \sigma _{t}-\check{\sigma}_{t}\right) ^{4}}{\sigma _{t}^{2}\hat{\sigma}_{t}^{2}}+\frac{2}{T}\sum_{t=1}^{T}\sum_{t^{\prime }=t+1}^{T}\frac{\left\vert E\left( \varepsilon_{t}\varepsilon _{t^{\prime }}\right) \right\vert \left( \sigma _{t}-\check{\sigma}_{t}\right) ^{2}\left( \sigma _{t^{\prime }}-\check{\sigma}_{t^{\prime }}\right) ^{2}}{\sigma _{t}^{2}\hat{\sigma}_{t}^{2}} \\
			&\leq &  C\left(\min_t \sigma_t \right)^{-2} \left(\min_t \hat{\sigma}_t \right)^{-2}  \left[\frac{1}{T}\sum_{t=1}^{T}\left( \sigma _{t}-\check{\sigma}_{t}\right)
			^{4}+\frac{1}{T}\sum_{t=1}^{T}\left( \sigma _{t}-\check{\sigma}_{t}\right)
			^{2}\sum_{t^{\prime }=t+1}^{T}\left\vert E\left( \varepsilon _{t}\varepsilon
			_{t^{\prime }}\right) \right\vert \right] \\
			&=&o_p\left( 1\right),
		\end{eqnarray*}%
		where $\frac{1}{T}\sum_{t=1}^{T}\left( \sigma _{t}-\check{\sigma}_{t}\right)
		^{i}=o\left( 1\right)$ for $i=2,4$  by Lemma \ref{lemma
			new other} (vii), and $\sum_{t^{\prime }=t+1}^{T}\left\vert E\left(
		\varepsilon _{t}\varepsilon _{t^{\prime }}\right) \right\vert =O\left(
		1\right) $ by \textcolor{black}{the} mixing condition.
		
		For $\mathcal{V}_{4}$, by H$\ddot{\mathrm{o}}$lder's inequality we have
		\begin{equation*}
			\left\vert \mathcal{V}_{4}\right\vert \leq \left(\min_t \sigma_t \right)^{-1} \left(\min_t \hat{\sigma}_t \right)^{-1} \left[ \frac{1}{T}%
			\sum_{t=1}^{T}\varepsilon _{t}^{2}\left( \check{\sigma}_{t}-\tilde{\sigma}%
			_{t}\right) ^{2}\right] ^{1/2}\left[ \sum_{t=1}^{T}\left( \tilde{\sigma}_{t}-%
			\hat{\sigma}_{t}\right) ^{2}\right] ^{1/2}=O_{p}\left( \frac{1}{Th^{2}}%
			\right),
		\end{equation*}%
		where $\sum_{t=1}^{T}\left( \tilde{\sigma}_{t}-\hat{\sigma}_{t}\right)
		^{2}=O_{p}\left( \frac{1}{Th^2}\right) $ by Lemma \ref{lemma new other}
		(iv), and $\frac{1}{T}\sum_{t=1}^{T}\varepsilon _{t}^{2}\left( \check{\sigma}%
		_{t}-\tilde{\sigma}_{t}\right) ^{2}=O_{p}\left( \frac{1}{Th^2}\right) $ since
		\begin{eqnarray*}
			\frac{1}{T}\sum_{t=1}^{T}E\left[ \varepsilon _{t}^{2}\left( \check{\sigma}%
			_{t}-\tilde{\sigma}_{t}\right) ^{2}\right]  \leq \frac{1}{T}%
			\sum_{t=1}^{T}\left( E\left( \varepsilon _{t}^{4}\right) E\left( \check{%
				\sigma}_{t}-\tilde{\sigma}_{t}\right) ^{4}\right) ^{1/2} 
			\leq C \left( \max_{t}E\left( \check{\sigma}_{t}-\tilde{\sigma}_{t}\right)
			^{4}\right) ^{1/2}=O\left( \frac{1}{Th^{2}}\right),
		\end{eqnarray*}%
		where $\max_{t}E\left( \check{\sigma}_{t}-\tilde{\sigma}_{t}\right)
		^{4} = O\left(\frac{1}{T^2h^4}\right)$ from  Lemma \ref{lemma new other} (v).
		
		In a similar manner, for $\mathcal{V}_5$, we also have
		\begin{equation*}
			\left\vert \mathcal{V}_{5}\right\vert \leq  \left(\min_t \sigma_t \right)^{-1} \left(\min_t \hat{\sigma}_t \right)^{-1}\left[ \frac{1}{T}%
			\sum_{t=1}^{T}\varepsilon _{t}^{2}\left( \sigma _{t}-\check{\sigma}%
			_{t}\right) ^{2}\right] ^{1/2}\left[ \sum_{t=1}^{T}\left( \tilde{\sigma}_{t}-%
			\hat{\sigma}_{t}\right) ^{2}\right] ^{1/2}=o_{p}\left( \frac{1}{\sqrt{T}h
			}\right),
		\end{equation*}%
		by Lemma \ref{lemma new other} (iv) and $E\left( \frac{1}{T}
		\sum_{t=1}^{T}\varepsilon _{t}^{2}\left( \sigma _{t}-\check{\sigma}_{t}\right) ^{2}\right) \leq \frac{C}{T}\sum_{t=1}^{T}\left( \sigma _{t}-%
		\check{\sigma}_{t}\right) ^{2}=o\left( 1\right) $ by Lemma \ref{lemma new other} (vii).
		
		Finally, for $\mathcal{V}_{6}$, we have $\mathcal{V}_6 \leq \left( \min_{t}\sigma_{t}\right) ^{-1}\left(\min_{t}\hat{\sigma}_{t}\right)^{-1} (\max_t |\tilde{\sigma}_t - \check{\sigma}_t|) \frac{1}{\sqrt{T}}\sum_{t=1}^T|\varepsilon_t\left( \sigma_t -\check{\sigma}_t\right)  |$, and
		\begin{align*}
			E\left( \frac{1}{\sqrt{T}}\sum_{t=1}^T|\varepsilon_t\left( \sigma_t -\check{\sigma}_t\right)  |\right)^2  &= \frac{1}{T}\sum_{t=1}^{T}			\sum_{t^{\prime }=1}^{T}E\left( \left\vert  \varepsilon_{t}\right\vert\left\vert \varepsilon_{t^{\prime }}\right\vert\right)  \left\vert \sigma _{t}-\check{\sigma}_{t}\right\vert \left\vert
			\sigma _{t^{\prime }}-\check{\sigma}_{t^{\prime }}\right\vert  \\
			&\leq  E\left( \varepsilon _{t}^{2}\right) \frac{1}{T}
			\sum_{t=1}^{T}\left( \sigma _{t}-\check{\sigma}_{t}\right) ^{2}  + 
			\frac{C}{T}\sum_{t=1}^{T}  \left\vert \sigma _{t}-\check{\sigma}_{t}\right\vert			\sum_{t^{\prime }=t + 1}^{T} 
			\left\vert E\left(\left\vert \varepsilon_t\right\vert \left\vert \varepsilon_{t^\prime}\right\vert  \right) \right\vert					
			\\&=o\left( 1\right),
		\end{align*}%
		where $\max_{t}\left\vert \check{\sigma}_{t}-\tilde{\sigma}_{t}\right\vert =O_p\left(\frac{1}{T^{1/4}h^{1/2}}\right)$, $\frac{1}{T}\sum_{t=1}^{T}\left\vert \sigma _{t}-\check{\sigma}_{t}\right\vert ^{\delta}=o\left( 1\right) $ for $\delta=1,2$ by Lemma \ref{lemma new other} (vi) and (vii), and $\{|\varepsilon_t|\}$ is a mixing process with the same size of $\{\varepsilon_t\}$.
		
		By	Combining the results of $\mathcal{V}_{1},\mathcal{V}_{2},\mathcal{V}_{3},%
		\mathcal{V}_{4},\mathcal{V}_{5}$ and $\mathcal{V}_{6},$ we obtain $\frac{1}{\sqrt{T}}\sum_{t=1}^{T}\varepsilon _{t}{\left( \sigma _{t}-\hat{\sigma}_{t}\right) ^{2}}/{(\sigma _{t}\hat{\sigma}_{t})}=o_{p}\left( 1\right) $, this completes the proof. 		
	\end{proof}

	{}  
	
\end{appendices}

	\end{sloppypar}	
\end{document}